\documentclass[acmsmall, screen, nonacm]{acmart}
\renewcommand\footnotetextcopyrightpermission[1]{} % removes footnote with conference information in first column
\pdfoutput=1
\usepackage{graphicx}

\setcopyright{rightsretained}
\acmPrice{}
\acmDOI{10.1145/3434284}
\acmYear{2021}
\copyrightyear{2021}
\acmSubmissionID{popl21main-p12-p}
\acmJournal{PACMPL}
\acmVolume{5}
\acmNumber{POPL}
\acmArticle{3}
\acmMonth{1}

\usepackage[hide=true]{marginalia}
\usepackage{subcaption}
\usepackage{mathpartir}
\usepackage{pifont}
\usepackage{url}
\usepackage{mdframed}
\usepackage{tikz-cd}
\usepackage{wrapfig}
\usetikzlibrary{calc,fit,positioning}
\usepackage{pgfplots}
\usepackage{enumitem}
\usepackage{microtype}

\usepackage[nameinlink]{cleveref}
\Crefname{figure}{Fig.}{Figures}
\Crefformat{section}{\S#2#1#3}
\Crefrangeformat{section}{\S#3#1#4-#5#2#6}

\newcommand{\IR}{\mathcal{R}}
\newcommand{\R}{\mathbb{R}}
\newcommand{\bool}{\mathbb{B}}

\newcommand{\nat}{\mathbb{N}}
\newcommand{\clang}{\ensuremath{\lambda_C}}
\newcommand{\dlang}{\ensuremath{\lambda_S}}

\newcommand{\cend}[2]{\int_{#1} #2}
\newcommand{\coend}[2]{\int^{#1} #2}

\newcommand{\Comment}[1]{}

\newcommand{\gor}{\ \ \mid\ \ }

\newcommand{\kint}[2]{[#1, #2]}
\newcommand{\denote}[1]{\llbracket #1 \rrbracket}

\newcommand{\sto}{\leadsto}

\newcommand{\cat}[1]{\textbf{\textup{#1}}}
\newcommand{\stimes}{\times}

\newcommand{\Cvx}{\mathbf{C}}
\newcommand{\y}[1]{\overline{#1}}

\newcommand{\codethensec}{\vspace{-2em}}
\newcommand{\codethenpar}{\vspace{-1em}}

\definecolor{dblue}{rgb}{0, 0, 0.8}
\definecolor{dred}{rgb}{0.8, 0, 0}
\definecolor{contcolor}{rgb}{0,0,0}
\definecolor{smoothcolor}{rgb}{0,0,0}

\newcommand{\Yes}{\textcolor{dblue}{\ding{51}}}
\newcommand{\No}{\textcolor{dred}{{\ding{55}}}}

\newcommand{\smooth}[1]{\textcolor{smoothcolor}{#1}}
\newcommand{\cont}[1]{\textcolor{contcolor}{#1}}
\newcommand\restr[2]{{% we make the whole thing an ordinary symbol
  \left.\kern-\nulldelimiterspace % automatically resize the bar with \right
  #1 % the function
  \vphantom{|} % pretend it's a little taller at normal size
  \right|_{#2} % this is the delimiter
  }}

\newenvironment{smoothc}{\color{smoothcolor}}{\ignorespacesafterend}
\newenvironment{contc}{\color{contcolor}}{\ignorespacesafterend}

\newenvironment{qsmooth}{\left\llbracket \begingroup\color{smoothcolor}}{\endgroup\right\rrbracket_\cat{HAD}\ignorespacesafterend}
\newenvironment{qvect}{\left\llbracket \begingroup\color{smoothcolor}}{\endgroup\right\rrbracket_\cat{AD}\ignorespacesafterend}

\newenvironment{uqcont}{\left\lfloor \begingroup\color{contcolor}}{\endgroup\right\rfloor\ignorespacesafterend}
\newenvironment{uq}{\left\lfloor}{\right\rfloor\ignorespacesafterend}

\newmdenv[backgroundcolor=lightgray!20, linewidth=0pt, skipabove=0em]{repl}
\newmdenv[backgroundcolor=orange!20, linewidth=0pt, skipbelow=0em]{prompt}

\newcommand*{\ESHRequirePackage}[2][]
  {\IfFileExists{#2.sty}
    {\RequirePackage[#1]{#2}}
    {\PackageError{ESH}{Missing LaTeX dependency: please install #2}
      {ESH requires relsize, ulem, and xcolor to work properly.  On Debian derivatives, these packages are provided by texlive-recommended and texlive-latex-extra.}}}

\ESHRequirePackage{xcolor} % TeXLive-recommended
\ESHRequirePackage[normalem]{ulem} % TeXLive-recommended

\ESHRequirePackage{relsize} % TeXLive-latex-extra

\newcommand*{\ESHFontSize}{}
\newcommand*{\ESHFontFamily}{\ttfamily}

\newcommand*{\ESHInlineFontSize}{\ESHFontSize}
\newcommand*{\ESHInlineFontFamily}{\ESHFontFamily}
\newcommand*{\ESHInlineFont}{\ESHInlineFontSize\ESHInlineFontFamily}

\newcommand*{\ESHInlineBlockFontSize}{\ESHFontSize}
\newcommand*{\ESHInlineBlockFontFamily}{\ESHFontFamily}
\newcommand*{\ESHInlineBlockFont}{\ESHInlineBlockFontSize\ESHInlineBlockFontFamily}

\newcommand*{\ESHBlockFontSize}{\ESHFontSize}
\newcommand*{\ESHBlockFontFamily}{\ESHFontFamily}
\newcommand*{\ESHBlockFont}{\ESHBlockFontSize\ESHBlockFontFamily}

\newcommand*{\ESHFallbackFontFamily}{\ESHFontFamily}
\newcommand*{\ESHFallbackFont}{\ESHFallbackFontFamily}

\newcommand*{\ESHNoHyphens}{\hyphenpenalty=10000}

\newcommand*{\ESHConstantSpace}{\spaceskip=\fontdimen2\font\xspaceskip=0pt}

\newdimen\ESHtempdim%
\newcommand*{\ESHCenterInWidthOf}[2]
  {\settowidth\ESHtempdim{#1}%
   \makebox[\ESHtempdim][c]{#2}}

\DeclareRobustCommand*{\ESHText}[1]{\ifmmode{\textnormal{#1}}\else{#1}\fi}

\newcommand*{\ESHIfFontChar}[1]
  {\iffontchar\font`#1{#1}\else{\ESHFallbackFont#1}\fi}

\makeatletter
\@ifundefined{XeTeXinterchartoks}
  {\@ifundefined{luatexversion}
     {% Not using XeTeX, nor LuaTeX: fall back to just using \ESHFallbackFont
      \def\ESHWithFallback#1{\ESHFallbackFont#1}}
     {% Using LuaTeX
      \def\ESHWithFallback#1{\ESHIfFontChar{#1}}}}
  {% Using XeTeX
   \def\ESHWithFallback#1{%
     \ifnum\XeTeXfonttype\font>0% Graphite, OpenType, or AAT font
       \ESHIfFontChar{#1}%
     \else% Legacy TeX font
       \setbox0=\hbox{\tracinglostchars=0\kern1sp#1\expandafter}%
       \ifnum\lastkern=1{\ESHFallbackFont#1}\else{#1}\fi
     \fi}}
\makeatother

\DeclareRobustCommand*{\ESHInlineSpecialChar}[1]
  {{\ESHInlineFontFamily\ESHWithFallback{#1}}}
\DeclareRobustCommand*{\ESHBlockSpecialChar}[1]
  {{\ESHCenterInWidthOf{\ESHBlockFontFamily{a}}{\ESHBlockFontFamily\ESHWithFallback{#1}}}}
\DeclareRobustCommand*{\ESHInlineUnicodeSubstitution}[1]
  {{\ESHInlineFontFamily#1}}
\DeclareRobustCommand*{\ESHBlockUnicodeSubstitution}[1]
  {{\ESHCenterInWidthOf{\ESHBlockFontFamily{a}}{\ESHBlockFontFamily#1}}}
\DeclareRobustCommand*{\ESHMathSymbol}[1]{\ensuremath{#1}}

\makeatletter
\newlength{\ESHCurFontSize}
\newcommand*{\ESHSetCurFontSize}{\setlength{\ESHCurFontSize}{\f@size pt}}
\makeatother

\DeclareRobustCommand*{\ESHInlineRaise}[2]
  {\ESHSetCurFontSize\raisebox{#1\ESHCurFontSize}{\relsize{-2}#2}}
\DeclareRobustCommand*{\ESHBlockRaise}[2]
  {\rlap{\ESHInlineRaise{#1}{#2}}\hphantom{#2}}

\newlength{\ESHBaselineskip}
\DeclareRobustCommand*{\ESHBlockStrut}[1]
  {\rule{0pt}{#1\ESHBaselineskip}}

\makeatletter
\DeclareRobustCommand*{\ESHUnderline}[1]
  {\bgroup\UL@setULdepth\markoverwith{#1\rule[-\ULdepth]{2pt}{0.4pt}}\ULon}
\DeclareRobustCommand*{\ESHUnderwave}[1]
  {\bgroup\UL@setULdepth\markoverwith{#1\raisebox{-\ULdepth}{\raisebox{-.5\height}{\ESHSmallWaveFont\char58}}}\ULon}
\font\ESHSmallWaveFont=lasyb10 scaled 400
\makeatother

\newlength{\ESHBoxSep}
\newdimen\ESHBoxTempDim%

\newcommand*{\ESHInlineSlantItalic}[1]{\textit{#1}}
\newcommand*{\ESHBlockSlantItalic}[1]{{\itshape{#1}}} % No italic correction

\newcommand*{\ESHBreakingSpace}{\ }
\newcommand*{\ESHNonbreakingSpace}{\nobreakspace}

\let\ESHSpecialChar\ignorespaces%
\let\ESHUnicodeSubstitution\ignorespaces%
\let\ESHRaise\ignorespaces%
\let\ESHBol\ignorespaces%
\let\ESHEmptyLine\ignorespaces%
\let\ESHEol\ignorespaces%
\let\ESHSpace\ignorespaces%
\let\ESHDash\ignorespaces%
\let\ESHSlantItalic\ignorespaces%
\DeclareRobustCommand*{\ESHInlineInternalSetup}
  {\def\ESHSpecialChar{\ESHInlineSpecialChar}\def\ESHUnicodeSubstitution{\ESHInlineUnicodeSubstitution}%
   \def\ESHRaise{\ESHInlineRaise}\def\ESHSlantItalic{\ESHInlineSlantItalic}%
   \def\ESHStrut{\relax}\def\ESHBol{\relax}\def\ESHEmptyLine{\mbox{}}\def\ESHEol{\newline}%
   \def\ESHSpace{\ESHBreakingSpace}%
   \def\ESHDash{-}}
\DeclareRobustCommand*{\ESHInlineBlockInternalSetup}
  {\def\arraystretch{1}%
   \def\ESHSpecialChar{\ESHBlockSpecialChar}\def\ESHUnicodeSubstitution{\ESHBlockUnicodeSubstitution}%
   \def\ESHRaise{\ESHBlockRaise}\def\ESHSlantItalic{\ESHBlockSlantItalic}%
   \setlength{\ESHBaselineskip}{\baselineskip}\def\ESHStrut{\ESHBlockStrut}%
   \def\ESHBol{\-}\def\ESHEmptyLine{\mbox{}}\def\ESHEol{\cr}%
   \def\ESHSpace{\ESHNonbreakingSpace}\def\ESHDash{\hbox{-}\nobreak}}
\DeclareRobustCommand*{\ESHBlockInternalSetup}
  {\def\ESHSpecialChar{\ESHBlockSpecialChar}\def\ESHUnicodeSubstitution{\ESHBlockUnicodeSubstitution}%
   \def\ESHRaise{\ESHBlockRaise}\def\ESHSlantItalic{\ESHBlockSlantItalic}%
   \setlength{\ESHBaselineskip}{\baselineskip}\def\ESHStrut{\ESHBlockStrut}%
   \def\ESHBol{\-}\def\ESHEmptyLine{\mbox{}}\def\ESHEol{\newline}%
   \def\ESHSpace{\ESHNonbreakingSpace}\def\ESHDash{\hbox{-}\nobreak}}

\newcommand*{\ESHInlineBasicSetup}
  {\leavevmode\ESHNoHyphens\ESHInlineFont}
\newcommand*{\ESHInlineBlockBasicSetup}
  {\ESHNoHyphens\ESHInlineBlockFont\ESHConstantSpace}
\newcommand*{\ESHBlockBasicSetup}
  {\setlength{\parindent}{0pt}\raggedright\ESHNoHyphens%
   \ESHBlockFont\ESHConstantSpace}

\newcommand*{\ESHHook}{}
\newcommand*{\ESHInlineHook}{\ESHHook}
\newcommand*{\ESHInlineBlockHook}{\ESHHook}
\newcommand*{\ESHBlockHook}{\ESHHook}

\makeatletter
\DeclareRobustCommand*{\ESHInline}[1]
  {\bgroup\ESHText{\ESHInlineInternalSetup\ESHInlineBasicSetup\ESHInlineHook#1}\egroup}

\newcommand*{\ESHInlineBlockVerticalAlignment}{c}

\newenvironment{ESHInlineBlock}[1][\ESHInlineBlockVerticalAlignment]
  {\bgroup\ESHInlineBlockInternalSetup\ESHInlineBlockBasicSetup\ESHInlineBlockHook\begin{tabular}[#1]{@{}l@{}}}
  {\end{tabular}\egroup}

\newlength{\ESHSkip}
\newcommand*{\ESHNoBreakAddVSpace}[1]{\addpenalty{\@M}\addvspace{#1}}

\newenvironment{ESHBlock}
  {\par\ESHNoBreakAddVSpace{\ESHSkip}\bgroup\ESHBlockInternalSetup\ESHBlockBasicSetup\ESHBlockHook}
  {\par\egroup\addvspace{\ESHSkip}}
\makeatother

\makeatletter
\def\ESHpvEnterVerbMode{\let\do\@makeother\dospecials}
\makeatother

\def\ESHpvNotFoundMessage#1{No highlighting found for "#1"; falling back to verbatim.}
\def\ESHpvNotFound#1{\PackageWarning{ESH}{\ESHpvNotFoundMessage{#1}}}

\def\ESHpvDefineSubstitution#1#2{\expandafter\def\csname #1-\detokenize{#2}\endcsname}

\def\ESHpvSubstitute#1#2#3#4{%
  \expandafter\ifx\csname #1-\detokenize{#2}\endcsname\relax
  \ESHpvNotFound{#3}\texttt{#4}%
  \else
  \csname #1-\detokenize{#2}\expandafter\endcsname
  \fi}

\newcommand*{\ESHSaveBoxName}[1]{ESHSaveBox:#1}

\makeatletter
\newcommand{\ESHEnsureSaveBox}[1]
  {\@ifundefined{#1}{\expandafter\newsavebox\csname#1\endcsname}{}}
\makeatother

\newcommand{\ESHUseVerb}[1]{\expandafter\usebox\csname\ESHSaveBoxName{#1}\endcsname}

\DeclareRobustCommand*{\ESHInlineRaise}[2]
  {\ESHSetCurFontSize\raisebox{#1\ESHCurFontSize}{\relsize{-1}\hspace{0.1em}#2}}
\DeclareRobustCommand*{\ESHBlockRaise}[2]
  {\ESHInlineRaise{#1}{#2}}

\makeatletter
\@ifundefined{XeTeXinterchartoks}{% Minimal pdfLaTeX setup
  \usepackage[T1]{fontenc}
  \usepackage[utf8]{inputenc}
}{% Regular XeLaTeX setup (recommended)
  \usepackage{fontspec}

  \newfontfamily{\XITSMath}[Path=fonts/,
                            UprightFont=*.otf,
                            BoldFont=*bold.otf,
                            Mapping=tex-ansi]{xits-math}

  \newfontfamily{\UbuntuMono}[Path=fonts/,
                              UprightFont=*-R.ttf,
                              BoldFont=*-B.ttf,
                              ItalicFont=*-RI.ttf,
                              BoldItalicFont=*-BI.ttf,
                              Mapping=tex-ansi]{UbuntuMono}

  \renewcommand{\ESHFontFamily}{\UbuntuMono}
  \renewcommand{\ESHFallbackFontFamily}{\XITSMath}
  \setlength{\ESHSkip}{0.3\baselineskip}

\renewenvironment{ESHBlock}
  {\trivlist\item\ESHNoBreakAddVSpace{\ESHSkip}\bgroup\ESHBlockInternalSetup\ESHBlockBasicSetup\ESHBlockHook}
  {\endtrivlist\egroup\addvspace{\ESHSkip}\noindent\hspace{-0.23em}}

}

\newcommand{\smoothpaperend}{\bibliographystyle{plainnat}
\bibliography{refs}

\appendix

\pagebreak
}
\newcommand{\smoothappendixsection}[1]{\section{#1}}

\newcommand{\paperonly}[1]{#1}
\newcommand{\thesisonly}[1]{}

\title{\dlang{}: Computable Semantics for Differentiable Programming with Higher-Order Functions and Datatypes}

\author{Benjamin Sherman}
\orcid{0000-0002-9569-3393}
\affiliation{%
  \institution{MIT}
  \country{USA}
}
\email{sherman@csail.mit.edu}

\author{Jesse Michel}
\affiliation{%
  \institution{MIT}
  \country{USA}
}
\email{jmmichel@mit.edu}

\author{Michael Carbin}
\affiliation{%
  \institution{MIT}
  \country{USA}
}
\email{mcarbin@csail.mit.edu}


\begin{document}

\begin{abstract}
Deep learning is moving towards increasingly sophisticated optimization objectives that employ higher-order functions, such as integration, continuous optimization, and root-finding.
Since differentiable programming frameworks such as PyTorch and TensorFlow do not have first-class representations of these functions, developers must reason about the semantics of such objectives and manually translate them to differentiable~code.

We present a differentiable programming language, \dlang{}, that is the first to deliver a semantics for higher-order functions, higher-order derivatives, and Lipschitz but nondifferentiable functions.
Together, these features enable \dlang{} to expose differentiable, higher-order functions for integration, optimization, and root-finding as first-class functions with automatically computed derivatives.
\dlang{}'s semantics is computable, meaning that values can be computed to arbitrary precision,
and we implement \dlang{} as an embedded language in Haskell.

We use \dlang{} to construct novel differentiable libraries for representing probability distributions, implicit surfaces, and generalized parametric surfaces -- all as instances of higher-order datatypes --
and present case studies that rely on computing the derivatives of these higher-order functions and datatypes.
In addition to modeling existing differentiable algorithms, such as a differentiable ray tracer for implicit surfaces, without requiring any user-level differentiation code, we demonstrate new differentiable algorithms, such as the Hausdorff distance of generalized parametric surfaces.
\end{abstract}

\maketitle

\section{Introduction}
\label{smooth:introduction}

Deep learning is centered on optimizing objectives $\smooth{\ell : \Theta \to \R}$ over some parameter space $\Theta$ by \emph{gradient descent}, following the derivative of $\ell$ at some particular value $\theta \in \Theta$ to move in a direction that decreases $\ell(\theta)$.
Before \NA{deep-learning practitioners} adopted frameworks such as TensorFlow and PyTorch,
creating a new model (i.e., parameter space and objective) was a laborious and error-prone endeavor,
since it involved manually determining and computing the derivative of the objective.
The advent of deep-learning frameworks that provide \emph{automatic differentiation} (\emph{AD})---the automated computation of derivatives of a function given just the definition of the function itself---has made creating and modifying models much easier:
a user simply writes the objective and its derivative is computed automatically.
As a result, progress in deep learning has rapidly accelerated -- a testament to the value of programming-language abstractions.

However, the creativity of deep-learning practitioners has exceeded the capabilities of current AD frameworks:
practitioners have devised objectives that current AD frameworks cannot handle directly.
A simple example is an objective including an expectation over a probability distribution whose parameters may vary, like this:
\[
\ell(\theta) = \mathbb{E}_{x \sim \mathcal{N}(\mu(\theta), \sigma^2(\theta))}[f(x)].
\]
If this $\ell$ is translated naïvely to PyTorch, by approximating the expectation with Monte Carlo sampling, the automatically generated derivative will be incorrect.
Numerous algorithms have been proposed to compute the derivatives of objectives that average over parameterized probability distributions \citep{naesseth2016reparameterization, NIPS2018_7326, jang2016categorical, jankowiak2018pathwise}.
How does one compute derivatives of objectives like these in general?
No existing differentiable-programming semantics has tackled the problem of differentiating through expectations such as these.

Other objectives are sufficiently complex that they do not even beg an incorrect naïve implementation. Objectives $\ell$ that optimize over compact sets $\Delta$
\[
\ell(\theta) = \max_{\delta \in \Delta} f(\theta, \delta)
\]
arise in \NA{adversarial contexts, including adversarial training and generative adversarial networks (GANs)}.
Conceptually, optimizing this objective with gradient-based techniques requires a semantics for a differentiable max operation over a compact set, which, to date, has not been covered in the literature on the semantics of differentiable programs.
Devising the appropriate derivative for these kinds of objectives is an object of current study \citep{wang2020on, lorraine2019optimizing}.

Sometimes, an objective involves \emph{root-finding},
\[
\ell(\theta) = \text{let } x \text{ be such that } g(\theta, x) = 0 \text{ in } f(\theta, x).
\]
This arises in learning implicit surfaces, with applications both to learning the decision boundaries of classifiers as well as to reconstructing surfaces from point-cloud data or other visual data.
How to compute the derivative of objectives like this is a key contribution of several papers \citep{atzmon2019controlling, niemeyer2019differentiable, bai2019deep}.

What do these objectives all have in common? They all involve \emph{higher-order functions}: their definitions introduce variables that are subject to integration, optimization, or root-finding.
Not only are these three operations troublesome in practice, but no semantics of differentiable programming has yet addressed them.

\paragraph{Approach.} We present \dlang{}, a differentiable programming language that includes higher-order functions for integration, optimization, and root-finding.
A key technical challenge is that these functions are higher-order and our semantic approach must wed higher-order functions with higher-order derivatives and nonsmooth functions to encompass these and other modern deep learning objectives.
As a toy example\fTBD{consider an example that differentiates more from Edalat}, consider computing the derivative $f'(0.6)$ of the function
\Comment{
\footnote{
\TBD{Remove green coloring}
Mathematical expressions written in \smooth{green} are to be interpreted in the internal language of \dlang{}, meaning that everything is implicitly \emph{smoothish} (defined in \Cref{ho-semantics}).
\citet[Chapter VI]{sheaves} provide an overview of the internal language of a topos.
}
}
\begin{align*}
f(c) &\triangleq \int_0^1 \text{ReLU}(x - c)\ dx,
\end{align*}
where $\text{ReLU}(x) = \max(0, x)$.
We can compute $f'(0.6) = -0.4$ with the \dlang{} expression
\begin{center}
\begin{minipage}[t]{0.8\textwidth}
\begin{prompt}
\ESHInline{\ESHBol{}eps=1e\ESHDash{}2{>}\ESHSpace{}deriv\ESHSpace{}(\textcolor[HTML]{346604}{\ESHUnicodeSubstitution{\ESHMathSymbol{\lambda}}}\ESHSpace{}c\ESHSpace{}\ESHUnicodeSubstitution{\ESHMathSymbol{\Rightarrow}}\ESHSpace{}integral01\ESHSpace{}(\textcolor[HTML]{346604}{\ESHUnicodeSubstitution{\ESHMathSymbol{\lambda}}}\ESHSpace{}x\ESHSpace{}\ESHUnicodeSubstitution{\ESHMathSymbol{\Rightarrow}}\ESHSpace{}relu\ESHSpace{}(x\ESHSpace{}\ESHDash{}\ESHSpace{}c)))\ESHSpace{}0.6}
\end{prompt}
\begin{repl}
\ESHInline{\ESHBol{}[\ESHDash{}0.407,\ESHSpace{}\ESHDash{}0.398]}
\end{repl}
\end{minipage}
\end{center}

\vspace{.5em}
{\noindent}where there is a type \ESHInline{\ESHBol{}\textcolor[HTML]{204A87}{\ESHUnicodeSubstitution{\ESHMathSymbol{\Re}}}} for real numbers,
a function \ESHInline{\ESHBol{}relu\ESHSpace{}:\ESHSpace{}\textcolor[HTML]{204A87}{\ESHUnicodeSubstitution{\ESHMathSymbol{\Re}}}\ESHSpace{}\ESHUnicodeSubstitution{\ESHMathSymbol{\rightarrow}}\ESHSpace{}\textcolor[HTML]{204A87}{\ESHUnicodeSubstitution{\ESHMathSymbol{\Re}}}} for ReLU,
a higher-order function \ESHInline{\ESHBol{}integral01\ESHSpace{}:\ESHSpace{}(\textcolor[HTML]{204A87}{\ESHUnicodeSubstitution{\ESHMathSymbol{\Re}}}\ESHSpace{}\ESHUnicodeSubstitution{\ESHMathSymbol{\rightarrow}}\ESHSpace{}\textcolor[HTML]{204A87}{\ESHUnicodeSubstitution{\ESHMathSymbol{\Re}}})\ESHSpace{}\ESHUnicodeSubstitution{\ESHMathSymbol{\rightarrow}}\ESHSpace{}\textcolor[HTML]{204A87}{\ESHUnicodeSubstitution{\ESHMathSymbol{\Re}}}} for integration over the unit interval [0,1],
and a higher-order function for differentiation of real-valued functions \ESHInline{\ESHBol{}deriv\ESHSpace{}:\ESHSpace{}(\textcolor[HTML]{204A87}{\ESHUnicodeSubstitution{\ESHMathSymbol{\Re}}}\ESHSpace{}\ESHUnicodeSubstitution{\ESHMathSymbol{\rightarrow}}\ESHSpace{}\textcolor[HTML]{204A87}{\ESHUnicodeSubstitution{\ESHMathSymbol{\Re}}})\ESHSpace{}\ESHUnicodeSubstitution{\ESHMathSymbol{\rightarrow}}\ESHSpace{}\textcolor[HTML]{204A87}{\ESHUnicodeSubstitution{\ESHMathSymbol{\Re}}}\ESHSpace{}\ESHUnicodeSubstitution{\ESHMathSymbol{\rightarrow}}\ESHSpace{}\textcolor[HTML]{204A87}{\ESHUnicodeSubstitution{\ESHMathSymbol{\Re}}}}.
The result can be queried to any precision, returning an interval guaranteed to include the true answer.
Here, the precision is specified in the prompt as \ESHInline{\ESHBol{}eps=1e\ESHDash{}2}.

\dlang{} is the first language that gives semantics to such an operation
and moreover is the first to support its computation to arbitrary precision.
Note that, in order to determine this derivative, we must evaluate the derivative of the ReLU function everywhere from -0.6 to 0.4, which includes 0, where ReLU is not (classically) differentiable.

Our work is unique compared to related work in supporting higher-order functions, higher-order derivatives, and nondifferentiable functions.
We combine the use of Clarke derivatives in \citet{edalat2013} to support nondifferentiable functions, the diffeological approach of \citet{vakar} to support higher-order functions, and the derivative towers of \citet{elliott-higher-ad} to support higher-order derivatives.
\Cref{smooth:related-work} covers related work in more detail.
Merging these techniques gives us a platform to accomplish the contributions described below.

\paragraph{Contributions.}
\label{smooth:contributions}

We present \emph{\dlang{}}, a differentiable programming language whose types are (generalized) smooth spaces and whose functions are (generalized) smooth maps.
%As shorthand, we call a type or expression \emph{smoothish} if it is representable in \dlang{}.
Our contributions are:
\begin{enumerate}
\item The first semantics for a differentiable programming language that admits all of the~following: 1) higher-order functions (\Cref{ho-semantics}), 2) higher-order derivatives (\Cref{fo-semantics}), and 3) Lipschitz but nonsmooth functions, such as \texttt{min}, \texttt{max}, and ReLU (\Cref{fo-semantics}).
\item The first semantics for differentiable integration, optimization, and root-finding (\Cref{ho-semantics}), enabled by the features above.

\item An implementation of this semantics, including implementations for higher-order functions such as integration (\Cref{smooth:implementation}).
	Our implementation is based directly on a constructive categorical semantics that demonstrates how these constructs can be computed to arbitrary precision.

\item New smooth libraries for constructing and computing on three higher-order datatypes: probability distributions, implicit surfaces, and generalized parametric surfaces (\Cref{smooth:applications}).
\end{enumerate}

\dlang{}'s semantics allows computation with and reasoning about the derivatives of higher-order functions, such as integration, optimization, and root-finding.
\dlang{} elucidates foundational principles for how to program with smooth values in a sound, arbitrarily precise manner, including which operations are possible to compute soundly and which are not.
While in many cases \dlang{}  is not practically efficient, in some cases, programs can serve as executable specifications to guide programming in other frameworks, to validate separately developed systems, and to suggest new functionality that could be added to other differentiable programming frameworks.

\section{An Introduction to \dlang{}}
\label{dlang-intro}

\begin{figure}[htbp]

\begin{subfigure}[t]{0.29\textwidth}
\begin{center}
\begin{tikzpicture}
\draw[fill=black] (1.3, 1) circle(0.2em);
\draw[fill=black] (0, 0) circle(0.2em);
\draw[thick, fill=black, fill opacity=0.1] (1.3, -3/4) circle (1);
\draw[dashed, thick, red] (0, 0) -- (0.63856, 0) -- (1.3, 1);
\draw[dashed, thick, red, opacity=0.3] (0, 0) -- (0.45, 0) -- (1.3, 1);
\draw[thick, fill=black, opacity=0.3, fill opacity=0.03] (1.3, -0.55) circle (1);
\draw[thick,->] (1.3, -3/4) -- (1.3, -0.55);
\end{tikzpicture}
\end{center}
\caption{A ray of light from a source above bounces off a circle before hitting a camera. How does the brightness change when the circle is moved up?}
\label{fig:ray-tracing-intro}
\end{subfigure}
\hfill
\begin{subfigure}{0.66\textwidth}
\small
\begin{ESHBlock}
\ESHBol{}\textcolor[HTML]{204A87}{type}\ESHSpace{}Surface\ESHSpace{}A\ESHSpace{}=\ESHSpace{}A\ESHSpace{}\ESHUnicodeSubstitution{\ESHMathSymbol{\rightarrow}}\ESHSpace{}\textcolor[HTML]{204A87}{\ESHUnicodeSubstitution{\ESHMathSymbol{\Re}}}\ESHEol
\ESHBol{}\ESHEmptyLine{}\ESHEol
\ESHBol{}firstRoot\ESHSpace{}:\ESHSpace{}(\textcolor[HTML]{204A87}{\ESHUnicodeSubstitution{\ESHMathSymbol{\Re}}}\ESHSpace{}\ESHUnicodeSubstitution{\ESHMathSymbol{\rightarrow}}\ESHSpace{}\textcolor[HTML]{204A87}{\ESHUnicodeSubstitution{\ESHMathSymbol{\Re}}})\ESHSpace{}\ESHUnicodeSubstitution{\ESHMathSymbol{\rightarrow}}\ESHSpace{}\textcolor[HTML]{204A87}{\ESHUnicodeSubstitution{\ESHMathSymbol{\Re}}}\ESHSpace{}\ESHSlantItalic{\textcolor[HTML]{5F615C}{!\ESHSpace{}language\ESHSpace{}primitive}}\ESHEol
\ESHBol{}\ESHEmptyLine{}\ESHEol
\ESHBol{}\textcolor[HTML]{346604}{let}\ESHSpace{}dot\ESHSpace{}(x\ESHSpace{}y\ESHSpace{}:\ESHSpace{}\textcolor[HTML]{204A87}{\ESHUnicodeSubstitution{\ESHMathSymbol{\Re}}}\ESHRaise{0.30}{2})\ESHSpace{}:\ESHSpace{}\textcolor[HTML]{204A87}{\ESHUnicodeSubstitution{\ESHMathSymbol{\Re}}}\ESHSpace{}=\ESHSpace{}x[0]\ESHSpace{}*\ESHSpace{}y[0]\ESHSpace{}+\ESHSpace{}x[1]\ESHSpace{}*\ESHSpace{}y[1]\ESHEol
\ESHBol{}\textcolor[HTML]{346604}{let}\ESHSpace{}scale\ESHSpace{}(c\ESHSpace{}:\ESHSpace{}\textcolor[HTML]{204A87}{\ESHUnicodeSubstitution{\ESHMathSymbol{\Re}}})\ESHSpace{}(x\ESHSpace{}:\ESHSpace{}\textcolor[HTML]{204A87}{\ESHUnicodeSubstitution{\ESHMathSymbol{\Re}}}\ESHRaise{0.30}{2})\ESHSpace{}:\ESHSpace{}\textcolor[HTML]{204A87}{\ESHUnicodeSubstitution{\ESHMathSymbol{\Re}}}\ESHRaise{0.30}{2}\ESHSpace{}=\ESHSpace{}(c\ESHSpace{}*\ESHSpace{}x[0],\ESHSpace{}c\ESHSpace{}*\ESHSpace{}x[1])\ESHEol
\ESHBol{}\textcolor[HTML]{346604}{let}\ESHSpace{}norm2\ESHSpace{}(x\ESHSpace{}:\ESHSpace{}\textcolor[HTML]{204A87}{\ESHUnicodeSubstitution{\ESHMathSymbol{\Re}}}\ESHRaise{0.30}{2})\ESHSpace{}:\ESHSpace{}\textcolor[HTML]{204A87}{\ESHUnicodeSubstitution{\ESHMathSymbol{\Re}}}\ESHSpace{}=\ESHSpace{}x[0]\ESHRaise{0.30}{2}\ESHSpace{}+\ESHSpace{}x[1]\ESHRaise{0.30}{2}\ESHEol
\ESHBol{}\textcolor[HTML]{346604}{let}\ESHSpace{}normalize\ESHSpace{}(x\ESHSpace{}:\ESHSpace{}\textcolor[HTML]{204A87}{\ESHUnicodeSubstitution{\ESHMathSymbol{\Re}}}\ESHRaise{0.30}{2})\ESHSpace{}:\ESHSpace{}\textcolor[HTML]{204A87}{\ESHUnicodeSubstitution{\ESHMathSymbol{\Re}}}\ESHRaise{0.30}{2}\ESHSpace{}=\ESHSpace{}scale\ESHSpace{}(1\ESHSpace{}/\ESHSpace{}sqrt\ESHSpace{}(norm2\ESHSpace{}x))\ESHSpace{}x\ESHEol
\ESHBol{}\ESHEmptyLine{}\ESHEol
\ESHBol{}deriv\ESHSpace{}:\ESHSpace{}(\textcolor[HTML]{204A87}{\ESHUnicodeSubstitution{\ESHMathSymbol{\Re}}}\ESHSpace{}\ESHUnicodeSubstitution{\ESHMathSymbol{\rightarrow}}\ESHSpace{}\textcolor[HTML]{204A87}{\ESHUnicodeSubstitution{\ESHMathSymbol{\Re}}})\ESHSpace{}\ESHUnicodeSubstitution{\ESHMathSymbol{\rightarrow}}\ESHSpace{}(\textcolor[HTML]{204A87}{\ESHUnicodeSubstitution{\ESHMathSymbol{\Re}}}\ESHSpace{}\ESHUnicodeSubstitution{\ESHMathSymbol{\rightarrow}}\ESHSpace{}\textcolor[HTML]{204A87}{\ESHUnicodeSubstitution{\ESHMathSymbol{\Re}}})\ESHSpace{}\ESHSlantItalic{\textcolor[HTML]{5F615C}{!\ESHSpace{}library\ESHSpace{}function}}\ESHEol
\ESHBol{}\textcolor[HTML]{346604}{let}\ESHSpace{}gradient\ESHSpace{}(f\ESHSpace{}:\ESHSpace{}\textcolor[HTML]{204A87}{\ESHUnicodeSubstitution{\ESHMathSymbol{\Re}}}\ESHRaise{0.30}{2}\ESHSpace{}\ESHUnicodeSubstitution{\ESHMathSymbol{\rightarrow}}\ESHSpace{}\textcolor[HTML]{204A87}{\ESHUnicodeSubstitution{\ESHMathSymbol{\Re}}})\ESHSpace{}(x\ESHSpace{}:\ESHSpace{}\textcolor[HTML]{204A87}{\ESHUnicodeSubstitution{\ESHMathSymbol{\Re}}}\ESHRaise{0.30}{2})\ESHSpace{}:\ESHSpace{}\textcolor[HTML]{204A87}{\ESHUnicodeSubstitution{\ESHMathSymbol{\Re}}}\ESHRaise{0.30}{2}\ESHSpace{}=\ESHEol
\ESHBol{}\ESHSpace{}\ESHSpace{}(deriv\ESHSpace{}(\textcolor[HTML]{346604}{\ESHUnicodeSubstitution{\ESHMathSymbol{\lambda}}}\ESHSpace{}z\ESHSpace{}:\ESHSpace{}\textcolor[HTML]{204A87}{\ESHUnicodeSubstitution{\ESHMathSymbol{\Re}}}\ESHSpace{}\ESHUnicodeSubstitution{\ESHMathSymbol{\Rightarrow}}\ESHSpace{}f\ESHSpace{}(z,\ESHSpace{}x[1]))\ESHSpace{}x[0],\ESHEol
\ESHBol{}\ESHSpace{}\ESHSpace{}\ESHSpace{}deriv\ESHSpace{}(\textcolor[HTML]{346604}{\ESHUnicodeSubstitution{\ESHMathSymbol{\lambda}}}\ESHSpace{}z\ESHSpace{}:\ESHSpace{}\textcolor[HTML]{204A87}{\ESHUnicodeSubstitution{\ESHMathSymbol{\Re}}}\ESHSpace{}\ESHUnicodeSubstitution{\ESHMathSymbol{\Rightarrow}}\ESHSpace{}f\ESHSpace{}(x[0],\ESHSpace{}z))\ESHSpace{}x[1])
\end{ESHBlock}
\codethenpar
\caption{Basic definitions used in \ESHInline{\ESHBol{}raytrace} below.}
\label{code:surfaces}
\end{subfigure}

\vspace{2em}

%%tangent A B : (A -> B) -> (Tan A -> Tan B) ! in the language
%%deriv : (real -> real) -> (real -> real) ! built in to language
%% let deriv (f : real -> real) (x : real) : real = snd (tangent_R.to (tangent f (tangent_R.from (x, 1)))

\begin{subfigure}{\textwidth}
\small
\begin{ESHBlock}
\ESHBol{}\ESHSlantItalic{\textcolor[HTML]{5F615C}{!\ESHSpace{}camera\ESHSpace{}assumed\ESHSpace{}to\ESHSpace{}be\ESHSpace{}at\ESHSpace{}the\ESHSpace{}origin}}\ESHEol
\ESHBol{}\textcolor[HTML]{346604}{let}\ESHSpace{}raytrace\ESHSpace{}(s\ESHSpace{}:\ESHSpace{}Surface\ESHSpace{}(\textcolor[HTML]{204A87}{\ESHUnicodeSubstitution{\ESHMathSymbol{\Re}}}\ESHRaise{0.30}{2}))\ESHSpace{}(lightPos\ESHSpace{}:\ESHSpace{}\textcolor[HTML]{204A87}{\ESHUnicodeSubstitution{\ESHMathSymbol{\Re}}}\ESHRaise{0.30}{2})\ESHSpace{}(rayDirection\ESHSpace{}:\ESHSpace{}\textcolor[HTML]{204A87}{\ESHUnicodeSubstitution{\ESHMathSymbol{\Re}}}\ESHRaise{0.30}{2})\ESHSpace{}:\ESHSpace{}\textcolor[HTML]{204A87}{\ESHUnicodeSubstitution{\ESHMathSymbol{\Re}}}\ESHSpace{}=\ESHEol
\ESHBol{}\ESHSpace{}\ESHSpace{}\textcolor[HTML]{346604}{let}\ESHSpace{}tStar\ESHSpace{}=\ESHSpace{}firstRoot\ESHSpace{}(\textcolor[HTML]{346604}{\ESHUnicodeSubstitution{\ESHMathSymbol{\lambda}}}\ESHSpace{}t\ESHSpace{}:\ESHSpace{}\textcolor[HTML]{204A87}{\ESHUnicodeSubstitution{\ESHMathSymbol{\Re}}}\ESHSpace{}\ESHUnicodeSubstitution{\ESHMathSymbol{\Rightarrow}}\ESHSpace{}s\ESHSpace{}(scale\ESHSpace{}t\ESHSpace{}rayDirection))\ESHSpace{}\textcolor[HTML]{346604}{in}\ESHEol
\ESHBol{}\ESHSpace{}\ESHSpace{}\textcolor[HTML]{346604}{let}\ESHSpace{}y\ESHSpace{}=\ESHSpace{}scale\ESHSpace{}tStar\ESHSpace{}rayDirection\ESHSpace{}\textcolor[HTML]{346604}{in}\ESHSpace{}\textcolor[HTML]{346604}{let}\ESHSpace{}normal\ESHSpace{}:\ESHSpace{}\textcolor[HTML]{204A87}{\ESHUnicodeSubstitution{\ESHMathSymbol{\Re}}}\ESHRaise{0.30}{2}\ESHSpace{}=\ESHSpace{}\ESHDash{}\ESHSpace{}gradient\ESHSpace{}s\ESHSpace{}y\ESHSpace{}\textcolor[HTML]{346604}{in}\ESHEol
\ESHBol{}\ESHSpace{}\ESHSpace{}\textcolor[HTML]{346604}{let}\ESHSpace{}lightToSurf\ESHSpace{}=\ESHSpace{}y\ESHSpace{}\ESHDash{}\ESHSpace{}lightPos\ESHSpace{}\textcolor[HTML]{346604}{in}\ESHEol
\ESHBol{}\ESHSpace{}\ESHSpace{}max\ESHSpace{}0\ESHSpace{}(dot\ESHSpace{}(normalize\ESHSpace{}normal)\ESHSpace{}(normalize\ESHSpace{}lightToSurf))\ESHEol
\ESHBol{}\ESHSpace{}\ESHSpace{}/\ESHSpace{}(norm2\ESHSpace{}y\ESHSpace{}*\ESHSpace{}norm2\ESHSpace{}lightToSurf)
\end{ESHBlock}
\codethenpar
\caption{A \dlang{} function for differentiable ray tracing of implicit surfaces.}
\label{code:raytracer}
\end{subfigure}

\caption{A library for differentiable ray tracing and scene representation.}
\label{fig:diff-ray-tracing}
\end{figure}

We demonstrate \dlang{}'s core functionality by implementing a simple differentiable \emph{ray tracer}, an algorithm that generates an image of a scene as viewed by a camera by tracing how rays of light emanate from a light source, bounce off the scene, and then enter the camera's aperture.
\emph{Differentiable ray tracing} is a new deep-learning technique that propagates derivatives through image rendering algorithms, permitting the use of inverse graphics to solve computer-vision tasks \citep{diffray, niemeyer2019differentiable}. These techniques optimize the parameters of a scene representation to make the image generated by the ray tracer more closely match a target image.\fTBD{MC: could use a little more grounding, such as explain what a task is}

As a simple example, consider computing the brightness of a particular scene at a particular direction,
using the  \dlang{} library for representing scenes and a function for performing ray tracing, both of which we present in \Cref{fig:diff-ray-tracing}:

\vspace{-0.8em}
\begin{center}
\begin{minipage}[t]{0.88\textwidth}
\begin{prompt}
\ESHInline{\ESHBol{}eps=1e\ESHDash{}5{>}\ESHSpace{}raytrace\ESHSpace{}(circle\ESHSpace{}(1,\ESHSpace{}\ESHDash{}3/4)\ESHSpace{}1)\ESHSpace{}(1,\ESHSpace{}1)\ESHSpace{}(1,\ESHSpace{}0)}
\end{prompt}
\begin{repl}
\ESHInline{\ESHBol{}[2.587289,\ESHSpace{}2.587299]}
\end{repl}
\end{minipage}
\end{center}
\Cref{fig:ray-tracing-intro} depicts the computation at hand. The camera is located at the origin \ESHInline{\ESHBol{}(0,\ESHSpace{}0)}, the circle is centered at \ESHInline{\ESHBol{}(1,\ESHSpace{}\ESHDash{}3/4)} and has radius \ESHInline{\ESHBol{}1}, the light source is at \ESHInline{\ESHBol{}(1,\ESHSpace{}1)}, and we consider a ray pointing horizontally to the right from the camera, in the direction \ESHInline{\ESHBol{}(1,\ESHSpace{}0)}.
The computation returns an interval and the \ESHInline{\ESHBol{}eps=1e\ESHDash{}5} specifies the precision tolerance, such that the interval-valued result, \ESHInline{\ESHBol{}[2.587289,\ESHSpace{}2.587299]}, has a width at most $10^{-5}$.
Our implementation guarantees that whenever it returns a finite-width interval, the true, real-valued result is contained within that interval.

\dlang{} permits differentiation of any functions in the language, so we can compute how the brightness would change if the circle were moved up by an infinitesimal amount:

\vspace{-.75em}
\begin{center}
\begin{minipage}[t]{\textwidth}
\begin{prompt}
\ESHInline{\ESHBol{}eps=1e\ESHDash{}3{>}\ESHSpace{}deriv\ESHSpace{}(\textcolor[HTML]{346604}{\ESHUnicodeSubstitution{\ESHMathSymbol{\lambda}}}\ESHSpace{}y\ESHSpace{}:\ESHSpace{}\textcolor[HTML]{204A87}{\ESHUnicodeSubstitution{\ESHMathSymbol{\Re}}}\ESHSpace{}\ESHUnicodeSubstitution{\ESHMathSymbol{\Rightarrow}}\ESHSpace{}raytrace\ESHSpace{}(circle\ESHSpace{}(0,\ESHSpace{}y)\ESHSpace{}1)\ESHSpace{}(1,\ESHSpace{}1)\ESHSpace{}(1,\ESHSpace{}0))\ESHSpace{}(\ESHDash{}3/4)}
\end{prompt}
\begin{repl}
\ESHInline{\ESHBol{}[1.3477,\ESHSpace{}1.3484]}
\end{repl}
\end{minipage}
\end{center}
The \dlang{} function \ESHInline{\ESHBol{}deriv\ESHSpace{}:\ESHSpace{}(\textcolor[HTML]{204A87}{\ESHUnicodeSubstitution{\ESHMathSymbol{\Re}}}\ESHSpace{}\ESHUnicodeSubstitution{\ESHMathSymbol{\rightarrow}}\ESHSpace{}\textcolor[HTML]{204A87}{\ESHUnicodeSubstitution{\ESHMathSymbol{\Re}}})\ESHSpace{}\ESHUnicodeSubstitution{\ESHMathSymbol{\rightarrow}}\ESHSpace{}(\textcolor[HTML]{204A87}{\ESHUnicodeSubstitution{\ESHMathSymbol{\Re}}}\ESHSpace{}\ESHUnicodeSubstitution{\ESHMathSymbol{\rightarrow}}\ESHSpace{}\textcolor[HTML]{204A87}{\ESHUnicodeSubstitution{\ESHMathSymbol{\Re}}})} computes the derivative of a scalar-valued real function.
The result indicates that when the circle is moved up infinitesimally from its current location, the brightness increases infinitesimally at a rate of $\sim$1.35 units brightness per unit distance the circle is moved up.

Several changes occur when the circle is moved up that affect the image brightness.
The point at which the light ray bounces off the circle moves closer to the camera,
decreasing the distance from the camera to the circle (increasing brightness) but increasing the distance from the light to the camera (decreasing brightness).
Both the direction of the surface normal of the circle at the point where the light deflects and the direction from the light source to that point change, increasing the angle between the surface normal of the circle and the light ray (decreasing brightness).
Automatic differentiation automatically takes all of these effects into account.

Figure \ref{code:raytracer} shows the implementation of the differentiable ray tracing in \dlang{}.
The function \ESHInline{\ESHBol{}firstRoot\ESHSpace{}:\ESHSpace{}(\textcolor[HTML]{204A87}{\ESHUnicodeSubstitution{\ESHMathSymbol{\Re}}}\ESHSpace{}\ESHUnicodeSubstitution{\ESHMathSymbol{\rightarrow}}\ESHSpace{}\textcolor[HTML]{204A87}{\ESHUnicodeSubstitution{\ESHMathSymbol{\Re}}})\ESHSpace{}\ESHUnicodeSubstitution{\ESHMathSymbol{\rightarrow}}\ESHSpace{}\textcolor[HTML]{204A87}{\ESHUnicodeSubstitution{\ESHMathSymbol{\Re}}}} in the definition of \ESHInline{\ESHBol{}raytrace} computes the distance that the light travels from the scene to the camera.
Given a function \ESHInline{\ESHBol{}f\ESHSpace{}:\ESHSpace{}\textcolor[HTML]{204A87}{\ESHUnicodeSubstitution{\ESHMathSymbol{\Re}}}\ESHSpace{}\ESHUnicodeSubstitution{\ESHMathSymbol{\rightarrow}}\ESHSpace{}\textcolor[HTML]{204A87}{\ESHUnicodeSubstitution{\ESHMathSymbol{\Re}}}}, \ESHInline{\ESHBol{}firstRoot\ESHSpace{}f} performs \emph{root finding}, computing $\min \{ x \in [0, 1] \mid \ESHInline{\ESHBol{}f}(x) = 0 \}$.
\dlang{}'s higher-order functions for root finding are novel, and accordingly,
\dlang{}'s ability to express differentiable ray tracing of implicit surfaces (embodied in \ESHInline{\ESHBol{}raytrace}) without needing any custom code for specifying derivatives.

The differentiable ray tracer \ESHInline{\ESHBol{}raytrace} critically depends on \dlang{}'s unique support for higher-order functions, higher-order derivatives, and Lipschitz but nondifferentiable functions such as min, max, and ReLU.
We now provide a brief introduction to these three features.

\subsection{Higher-Order Functions}

The \ESHInline{\ESHBol{}raytrace} function must compute the distance the ray of light travels from the scene to the camera, represented by the let-definition \ESHInline{\ESHBol{}tStar} in \ESHInline{\ESHBol{}raytrace}.
When applied to the scene \ESHInline{\ESHBol{}circle\ESHSpace{}(1,\ESHSpace{}y)\ESHSpace{}1}, the definition reduces to
\begin{ESHBlock}
\ESHBol{}\textcolor[HTML]{346604}{let}\ESHSpace{}tStar\ESHSpace{}y\ESHSpace{}=\ESHSpace{}firstRoot\ESHSpace{}(\textcolor[HTML]{346604}{\ESHUnicodeSubstitution{\ESHMathSymbol{\lambda}}}\ESHSpace{}t\ESHSpace{}:\ESHSpace{}\textcolor[HTML]{204A87}{\ESHUnicodeSubstitution{\ESHMathSymbol{\Re}}}\ESHSpace{}\ESHUnicodeSubstitution{\ESHMathSymbol{\Rightarrow}}\ESHSpace{}1\ESHSpace{}\ESHDash{}\ESHSpace{}y\ESHRaise{0.30}{2}\ESHSpace{}\ESHDash{}\ESHSpace{}(t\ESHSpace{}\ESHDash{}\ESHSpace{}1)\ESHRaise{0.30}{2})
\end{ESHBlock}
The function \ESHInline{\ESHBol{}firstRoot\ESHSpace{}:\ESHSpace{}(\textcolor[HTML]{204A87}{\ESHUnicodeSubstitution{\ESHMathSymbol{\Re}}}\ESHSpace{}\ESHUnicodeSubstitution{\ESHMathSymbol{\rightarrow}}\ESHSpace{}\textcolor[HTML]{204A87}{\ESHUnicodeSubstitution{\ESHMathSymbol{\Re}}})\ESHSpace{}\ESHUnicodeSubstitution{\ESHMathSymbol{\rightarrow}}\ESHSpace{}\textcolor[HTML]{204A87}{\ESHUnicodeSubstitution{\ESHMathSymbol{\Re}}}} is a higher-order function since it takes a function as input.
In order to admit a function like this in a differentiable programming language, the language must be able to compute how the result of \ESHInline{\ESHBol{}firstRoot} changes when there is an infinitesimal perturbation to its input function.
In this example, we want to know how \ESHInline{\ESHBol{}tStar} changes when \ESHInline{\ESHBol{}y} changes.
To answer this, define \ESHInline{\ESHBol{}f\ESHSpace{}t\ESHSpace{}y\ESHSpace{}=\ESHSpace{}1\ESHSpace{}\ESHDash{}\ESHSpace{}y\ESHRaise{0.30}{2}\ESHSpace{}\ESHDash{}\ESHSpace{}(t\ESHSpace{}\ESHDash{}\ESHSpace{}1)\ESHRaise{0.30}{2}}.
Then \ESHInline{\ESHBol{}tStar} finds a solution for the variable \ESHInline{\ESHBol{}t} to the equation \ESHInline{\ESHBol{}f\ESHSpace{}t\ESHSpace{}y\ESHSpace{}=\ESHSpace{}0}.
So whatever change is induced by changing \ESHInline{\ESHBol{}y} must be counterbalanced by changing \ESHInline{\ESHBol{}tStar}.
\dlang{}'s semantics validate the equation (for values of \ESHInline{\ESHBol{}y} giving well-defined roots)
\begin{ESHBlock}
\ESHBol{}deriv\ESHSpace{}tStar\ESHSpace{}y\ESHSpace{}=\ESHSpace{}\ESHDash{}\ESHSpace{}deriv\ESHSpace{}(\textcolor[HTML]{346604}{\ESHUnicodeSubstitution{\ESHMathSymbol{\lambda}}}\ESHSpace{}y0\ESHSpace{}:\ESHSpace{}\textcolor[HTML]{204A87}{\ESHUnicodeSubstitution{\ESHMathSymbol{\Re}}}\ESHSpace{}\ESHUnicodeSubstitution{\ESHMathSymbol{\Rightarrow}}\ESHSpace{}f\ESHSpace{}(tStar\ESHSpace{}y)\ESHSpace{}y0)\ESHSpace{}y\ESHSpace{}/\ESHEol
\ESHBol{}\ESHSpace{}\ESHSpace{}\ESHSpace{}\ESHSpace{}\ESHSpace{}\ESHSpace{}\ESHSpace{}\ESHSpace{}\ESHSpace{}\ESHSpace{}\ESHSpace{}\ESHSpace{}\ESHSpace{}\ESHSpace{}\ESHSpace{}\ESHSpace{}\ESHSpace{}\ESHSpace{}\ESHSpace{}deriv\ESHSpace{}(\textcolor[HTML]{346604}{\ESHUnicodeSubstitution{\ESHMathSymbol{\lambda}}}\ESHSpace{}t\ESHSpace{}:\ESHSpace{}\textcolor[HTML]{204A87}{\ESHUnicodeSubstitution{\ESHMathSymbol{\Re}}}\ESHSpace{}\ESHUnicodeSubstitution{\ESHMathSymbol{\Rightarrow}}\ESHSpace{}f\ESHSpace{}t\ESHSpace{}y)\ESHSpace{}(tStar\ESHSpace{}y)
\end{ESHBlock}
This equation for the derivative of root finding is known as the \emph{implicit function theorem}.
By the rules of calculus, we can further simplify this to
\[
\ESHInline{\ESHBol{}deriv\ESHSpace{}tStar\ESHSpace{}y\ESHSpace{}=\ESHSpace{}\ESHDash{}\ESHSpace{}y\ESHSpace{}/\ESHSpace{}(tStar\ESHSpace{}y\ESHSpace{}\ESHDash{}\ESHSpace{}1)}.
\]
Note that the semantics of \dlang{} ensures that these equations are indeed program equivalences: one can substitute one expression for the other within the context of a larger expression without affecting its meaning.
Indeed, taking \ESHInline{\ESHBol{}y\ESHSpace{}=\ESHSpace{}\ESHDash{}3/4}, and evaluating both sides of the expression above in \dlang{} produces compatible answers, roughly $-1.1$, which indicates that moving the circle up decreases the distance that the light travels from the circle to the camera.

We implement the \ESHInline{\ESHBol{}firstRoot} function as a language primitive by specifying not only how \ESHInline{\ESHBol{}firstRoot} acts on values but also how derivatives propagate through it, via the implicit function theorem (see \Cref{higher-order-primitives} for more detail).

\subsection{Higher-Order Derivatives}

The brightness of the image computed by the \ESHInline{\ESHBol{}raytrace} function depends on the angle at which the ray of light deflects as it bounces off the circle, so we need to know which direction the circle faces where the light hits it, which is known as the \emph{surface normal}.
In the code for \ESHInline{\ESHBol{}raytrace}, the surface normal is computed as
\begin{ESHBlock}
\ESHBol{}\textcolor[HTML]{346604}{let}\ESHSpace{}normal\ESHSpace{}:\ESHSpace{}\textcolor[HTML]{204A87}{\ESHUnicodeSubstitution{\ESHMathSymbol{\Re}}}\ESHRaise{0.30}{2}\ESHSpace{}=\ESHSpace{}\ESHDash{}\ESHSpace{}gradient\ESHSpace{}s\ESHSpace{}y
\end{ESHBlock}
Consider, for instance, the unit circle centered at $(0,0)$, i.e., \ESHInline{\ESHBol{}circle\ESHSpace{}(0,\ESHSpace{}0)\ESHSpace{}1}, given by the function $f(x, y) = 1 - x^2 - y^2$.
The surface normal is given by the negative gradient,
\[
- \nabla f (x, y)
= - \left( \frac{\partial f}{\partial x}, \frac{\partial f}{\partial y}\right)
= (2x, 2y)
\]
So, for instance, the point $(1/\sqrt{2}, 1/\sqrt{2})$ on the upper-right of the circle has a surface normal that points up and to the right, in the direction $(2/\sqrt{2}, 2/\sqrt{2})$.

Note that, in the \ESHInline{\ESHBol{}raytrace} code itself, this gradient computation requires the computation of derivatives of the implicitly defined surface in order to compute the image brightness.
Accordingly, computing the derivative of the image brightness with respect to an infinitesimal perturbation in the scene requires computing the second derivatives of the implicitly defined surface with respect to its arguments. Thus, higher-order differentiation is a valuable language feature.

In \dlang{}, differentiation is a first-class programming construct,
so higher-order differentiation is naturally supported,
as we can compute higher-order derivatives by applying the \ESHInline{\ESHBol{}deriv\ESHSpace{}:\ESHSpace{}(\textcolor[HTML]{204A87}{\ESHUnicodeSubstitution{\ESHMathSymbol{\Re}}}\ESHSpace{}\ESHUnicodeSubstitution{\ESHMathSymbol{\rightarrow}}\ESHSpace{}\textcolor[HTML]{204A87}{\ESHUnicodeSubstitution{\ESHMathSymbol{\Re}}})\ESHSpace{}\ESHUnicodeSubstitution{\ESHMathSymbol{\rightarrow}}\ESHSpace{}\textcolor[HTML]{204A87}{\ESHUnicodeSubstitution{\ESHMathSymbol{\Re}}}\ESHSpace{}\ESHUnicodeSubstitution{\ESHMathSymbol{\rightarrow}}\ESHSpace{}\textcolor[HTML]{204A87}{\ESHUnicodeSubstitution{\ESHMathSymbol{\Re}}}} function multiple times.
Note that some approaches to differentiable programming do not support higher-order differentiation (see \Cref{ad-semantics}) and thus do not have differentiation as a first-class construct.
Higher-order derivatives are also used for numerical integration, in optimization algorithms, and in other contexts.

The requirement to support higher-order derivatives means that language primitives, such as \ESHInline{\ESHBol{}firstRoot}, must specify not only how they act on values but also how derivatives \emph{of all orders} propagate through them.

\subsection{Nondifferentiability}

Note that the \ESHInline{\ESHBol{}raytrace} code uses the built-in function \ESHInline{\ESHBol{}max\ESHSpace{}:\ESHSpace{}\textcolor[HTML]{204A87}{\ESHUnicodeSubstitution{\ESHMathSymbol{\Re}}}\ESHSpace{}\ESHUnicodeSubstitution{\ESHMathSymbol{\rightarrow}}\ESHSpace{}\textcolor[HTML]{204A87}{\ESHUnicodeSubstitution{\ESHMathSymbol{\Re}}}\ESHSpace{}\ESHUnicodeSubstitution{\ESHMathSymbol{\rightarrow}}\ESHSpace{}\textcolor[HTML]{204A87}{\ESHUnicodeSubstitution{\ESHMathSymbol{\Re}}}} in computing the image brightness.
If the light source is behind the scene, the dot product of the surface normal and the vector from the light to the surface will be negative, but the brightness should be 0, rather than this negative value.
Hence, we clamp the value to be at least zero by applying \ESHInline{\ESHBol{}max\ESHSpace{}0}.
Note that this function is exactly the rectified linear unit (ReLU) that is common in deep learning:
\begin{ESHBlock}
\ESHBol{}\textcolor[HTML]{346604}{let}\ESHSpace{}relu\ESHSpace{}(x\ESHSpace{}:\ESHSpace{}\textcolor[HTML]{204A87}{\ESHUnicodeSubstitution{\ESHMathSymbol{\Re}}})\ESHSpace{}:\ESHSpace{}\textcolor[HTML]{204A87}{\ESHUnicodeSubstitution{\ESHMathSymbol{\Re}}}\ESHSpace{}=\ESHSpace{}max\ESHSpace{}0\ESHSpace{}x
\end{ESHBlock}
ReLU is not differentiable at 0.
When we compute its derivative at 0 in \dlang{}, we get a \emph{nonmaximal} result.
That means that, for sufficiently fine ($\le 1$) precision tolerances, we get nontermination:
\begin{center}
\begin{minipage}[t]{0.48\textwidth}
\begin{prompt}
\ESHInline{\ESHBol{}eps=1e\ESHDash{}1{>}\ESHSpace{}deriv\ESHSpace{}relu\ESHSpace{}0}
\end{prompt}
\begin{repl}
(nontermination)
\end{repl}
\end{minipage}
\begin{minipage}[t]{0.48\textwidth}
\begin{prompt}
\ESHInline{\ESHBol{}eps=2{>}\ESHSpace{}deriv\ESHSpace{}relu\ESHSpace{}0}
\end{prompt}
\begin{repl}
\ESHInline{\ESHBol{}[0.0,\ESHSpace{}1.0]}
\end{repl}
\end{minipage}
\end{center}
\vspace{.5em}
The interval approximations never converge to intervals smaller than $[0, 1]$.
The type \ESHInline{\ESHBol{}\textcolor[HTML]{204A87}{\ESHUnicodeSubstitution{\ESHMathSymbol{\Re}}}} contains, in addition to the real numbers, nonmaximal elements such as this one, which we name $[0, 1]$, e.g., we find that $\text{ReLU}'(0) = [0, 1]$.

Differentiable programming frameworks such as PyTorch admit min and max operations, but they are unsound, in the sense that one can define $f(x) = \max(x, 0) + \min(0, x)$, which is the identity function, but compute in PyTorch that $f'(0) = 2$, whereas it should be $f'(0) = 1$.
Because of this issue, most differentiable programming semantics leave the derivative of $\max$ undefined at~0.

However, \dlang{}'s interval-valued semantics for functions like $\max$ enables productive computational functionality that the partiality approach would not permit.
For instance, suppose rather than having a point light source for ray-tracing, we instead have a line light source, so we integrate over the entire line, using the primitive higher-order function \ESHInline{\ESHBol{}integral01\ESHSpace{}:\ESHSpace{}(\textcolor[HTML]{204A87}{\ESHUnicodeSubstitution{\ESHMathSymbol{\Re}}}\ESHSpace{}\ESHUnicodeSubstitution{\ESHMathSymbol{\rightarrow}}\ESHSpace{}\textcolor[HTML]{204A87}{\ESHUnicodeSubstitution{\ESHMathSymbol{\Re}}})\ESHSpace{}\ESHUnicodeSubstitution{\ESHMathSymbol{\rightarrow}}\ESHSpace{}\textcolor[HTML]{204A87}{\ESHUnicodeSubstitution{\ESHMathSymbol{\Re}}}}, where for \ESHInline{\ESHBol{}f\ESHSpace{}:\ESHSpace{}\textcolor[HTML]{204A87}{\ESHUnicodeSubstitution{\ESHMathSymbol{\Re}}}\ESHSpace{}\ESHUnicodeSubstitution{\ESHMathSymbol{\rightarrow}}\ESHSpace{}\textcolor[HTML]{204A87}{\ESHUnicodeSubstitution{\ESHMathSymbol{\Re}}}}, \ESHInline{\ESHBol{}integral01\ESHSpace{}f} computes the integral of \ESHInline{\ESHBol{}f} over the unit interval, $\int_0^1 \ESHInline{\ESHBol{}f}(x)\ dx$.
For simplicity, consider a camera located at $(0, 1)$ pointing downwards at a flat surface that stretches from $(-1, y)$ to $(1, y)$, with a light source stretching from $(1, 0)$ to $(1, 1)$.
Furthermore, let us disregard the effect of brightness decreasing when the light travels longer distances, such that the brightness is
\begin{ESHBlock}
\ESHBol{}\textcolor[HTML]{346604}{let}\ESHSpace{}brightness\ESHSpace{}(y\ESHSpace{}:\ESHSpace{}\textcolor[HTML]{204A87}{\ESHUnicodeSubstitution{\ESHMathSymbol{\Re}}})\ESHSpace{}:\ESHSpace{}\textcolor[HTML]{204A87}{\ESHUnicodeSubstitution{\ESHMathSymbol{\Re}}}\ESHSpace{}=\ESHEol
\ESHBol{}\ESHSpace{}\ESHSpace{}integral01\ESHSpace{}(\textcolor[HTML]{346604}{\ESHUnicodeSubstitution{\ESHMathSymbol{\lambda}}}\ESHSpace{}y0\ESHSpace{}:\ESHSpace{}\textcolor[HTML]{204A87}{\ESHUnicodeSubstitution{\ESHMathSymbol{\Re}}}\ESHSpace{}\ESHUnicodeSubstitution{\ESHMathSymbol{\Rightarrow}}\ESHSpace{}max\ESHSpace{}0\ESHSpace{}((y0\ESHSpace{}\ESHDash{}\ESHSpace{}y)\ESHSpace{}/\ESHSpace{}sqrt\ESHSpace{}(1\ESHSpace{}+\ESHSpace{}(y0\ESHSpace{}\ESHDash{}\ESHSpace{}y)\ESHRaise{0.30}{2})))
\end{ESHBlock}
When $0 \le \ESHInline{\ESHBol{}y} \le 1$, the integrand will be nondifferentiable with respect to \ESHInline{\ESHBol{}y} at the point where \ESHInline{\ESHBol{}y0\ESHSpace{}=\ESHSpace{}y}.
For instance, taking \ESHInline{\ESHBol{}y0\ESHSpace{}=\ESHSpace{}y\ESHSpace{}=\ESHSpace{}1/2}, we find that the derivative of the integrand is
\[
\ESHInline{\ESHBol{}deriv\ESHSpace{}(\textcolor[HTML]{346604}{\ESHUnicodeSubstitution{\ESHMathSymbol{\lambda}}}\ESHSpace{}y\ESHSpace{}:\ESHSpace{}\textcolor[HTML]{204A87}{\ESHUnicodeSubstitution{\ESHMathSymbol{\Re}}}\ESHSpace{}\ESHUnicodeSubstitution{\ESHMathSymbol{\Rightarrow}}\ESHSpace{}max\ESHSpace{}0\ESHSpace{}((1/2\ESHSpace{}\ESHDash{}\ESHSpace{}y)\ESHSpace{}/\ESHSpace{}sqrt\ESHSpace{}(1\ESHSpace{}+\ESHSpace{}(1/2\ESHSpace{}\ESHDash{}\ESHSpace{}y)\ESHRaise{0.30}{2}))))\ESHSpace{}(1/2)\ESHSpace{}=\ESHSpace{}[\ESHDash{}1,\ESHSpace{}0]}.
\]
When \ESHInline{\ESHBol{}y0} is just greater than \ESHInline{\ESHBol{}y}, the derivative will be near $-1$, but when \ESHInline{\ESHBol{}y0} is just less than \ESHInline{\ESHBol{}y}, the derivative will be near $0$.
Because the derivative at this point is a bounded interval, rather than a completely undefined result, it ends up being soundly neglected when it is integrated over:
\vspace{-.75em}
\begin{center}
\begin{minipage}[t]{0.8\textwidth}
\begin{prompt}
\ESHInline{\ESHBol{}eps=1e\ESHDash{}3{>}\ESHSpace{}deriv\ESHSpace{}brightness\ESHSpace{}(1/2)}
\end{prompt}
\begin{repl}
\ESHInline{\ESHBol{}[\ESHDash{}0.4476,\ESHSpace{}\ESHDash{}0.4469]}
\end{repl}
\end{minipage}
\end{center}

{\noindent}The expression \ESHInline{\ESHBol{}deriv\ESHSpace{}brightness\ESHSpace{}(1/2)} is indeed maximal, meaning that it can be evaluated to arbitrary precision.
Were the derivative of the integrand to be undefined rather than interval-valued, \ESHInline{\ESHBol{}deriv\ESHSpace{}brightness\ESHSpace{}(1/2)} would necessarily need to be undefined as well, but with these semantics, we can soundly compute the correct derivative.

This generalized notion of derivative that works for ReLU is based on \emph{Clarke's generalized derivative} \citep{clarke1990}.
The basic idea can be motivated by the desire for continuity and robustness in the numerical computation.
The derivative of ReLU is 1 for numbers imperceptibly greater than 0, and the derivative is 0 for numbers imperceptibly smaller than 0, so the derivative of ReLU at 0 should be consistent with those nearby answers.
The \emph{specialization relation} $\sqsubseteq$ on \ESHInline{\ESHBol{}\textcolor[HTML]{204A87}{\ESHUnicodeSubstitution{\ESHMathSymbol{\Re}}}} formalizes this notion of compatible behavior, where we have $[0, 1] \sqsubseteq 0$ and $[0, 1] \sqsubseteq 1$.
We will prove a consistency theorem for our language (\Cref{ho-soundness}) that says that derivatives are always compatible, i.e., related by $\sqsubseteq$, with the infinitesimal rates of change indicated by its value-level operation.

\section{Syntax and Semantics of \dlang{}}
\label{sec:syntax}

\begin{figure}[htbp]
\small
\begin{subfigure}[t]{0.42\textwidth}
\begin{ESHBlock}
\ESHBol{}0,\ESHSpace{}1,\ESHSpace{}2,\ESHSpace{}...\ESHSpace{}:\ESHSpace{}\textcolor[HTML]{204A87}{\ESHUnicodeSubstitution{\ESHMathSymbol{\Re}}}\ESHEol
\ESHBol{}(+),\ESHSpace{}(\ESHDash{}),\ESHSpace{}(*),\ESHSpace{}(/)\ESHSpace{}:\ESHSpace{}\textcolor[HTML]{204A87}{\ESHUnicodeSubstitution{\ESHMathSymbol{\Re}}}\ESHSpace{}\ESHUnicodeSubstitution{\ESHMathSymbol{\rightarrow}}\ESHSpace{}\textcolor[HTML]{204A87}{\ESHUnicodeSubstitution{\ESHMathSymbol{\Re}}}\ESHSpace{}\ESHUnicodeSubstitution{\ESHMathSymbol{\rightarrow}}\ESHSpace{}\textcolor[HTML]{204A87}{\ESHUnicodeSubstitution{\ESHMathSymbol{\Re}}}\ESHEol
\ESHBol{}max\ESHSpace{}:\ESHSpace{}\textcolor[HTML]{204A87}{\ESHUnicodeSubstitution{\ESHMathSymbol{\Re}}}\ESHSpace{}\ESHUnicodeSubstitution{\ESHMathSymbol{\rightarrow}}\ESHSpace{}\textcolor[HTML]{204A87}{\ESHUnicodeSubstitution{\ESHMathSymbol{\Re}}}\ESHSpace{}\ESHUnicodeSubstitution{\ESHMathSymbol{\rightarrow}}\ESHSpace{}\textcolor[HTML]{204A87}{\ESHUnicodeSubstitution{\ESHMathSymbol{\Re}}}\ESHEol
\ESHBol{}sin,\ESHSpace{}exp\ESHSpace{}:\ESHSpace{}\textcolor[HTML]{204A87}{\ESHUnicodeSubstitution{\ESHMathSymbol{\Re}}}\ESHSpace{}\ESHUnicodeSubstitution{\ESHMathSymbol{\rightarrow}}\ESHSpace{}\textcolor[HTML]{204A87}{\ESHUnicodeSubstitution{\ESHMathSymbol{\Re}}}\ESHEol
\ESHBol{}\ESHEmptyLine{}\ESHEol
\ESHBol{}integral01\ESHSpace{}:\ESHSpace{}(\textcolor[HTML]{204A87}{\ESHUnicodeSubstitution{\ESHMathSymbol{\Re}}}\ESHSpace{}\ESHUnicodeSubstitution{\ESHMathSymbol{\rightarrow}}\ESHSpace{}\textcolor[HTML]{204A87}{\ESHUnicodeSubstitution{\ESHMathSymbol{\Re}}})\ESHSpace{}\ESHUnicodeSubstitution{\ESHMathSymbol{\rightarrow}}\ESHSpace{}\textcolor[HTML]{204A87}{\ESHUnicodeSubstitution{\ESHMathSymbol{\Re}}}\ESHEol
\ESHBol{}cutRoot\ESHSpace{}:\ESHSpace{}(\textcolor[HTML]{204A87}{\ESHUnicodeSubstitution{\ESHMathSymbol{\Re}}}\ESHSpace{}\ESHUnicodeSubstitution{\ESHMathSymbol{\rightarrow}}\ESHSpace{}\textcolor[HTML]{204A87}{\ESHUnicodeSubstitution{\ESHMathSymbol{\Re}}})\ESHSpace{}\ESHUnicodeSubstitution{\ESHMathSymbol{\rightarrow}}\ESHSpace{}\textcolor[HTML]{204A87}{\ESHUnicodeSubstitution{\ESHMathSymbol{\Re}}}\ESHEol
\ESHBol{}firstRoot\ESHSpace{}:\ESHSpace{}(\textcolor[HTML]{204A87}{\ESHUnicodeSubstitution{\ESHMathSymbol{\Re}}}\ESHSpace{}\ESHUnicodeSubstitution{\ESHMathSymbol{\rightarrow}}\ESHSpace{}\textcolor[HTML]{204A87}{\ESHUnicodeSubstitution{\ESHMathSymbol{\Re}}})\ESHSpace{}\ESHUnicodeSubstitution{\ESHMathSymbol{\rightarrow}}\ESHSpace{}\textcolor[HTML]{204A87}{\ESHUnicodeSubstitution{\ESHMathSymbol{\Re}}}\ESHEol
\ESHBol{}max01\ESHSpace{}:\ESHSpace{}(\textcolor[HTML]{204A87}{\ESHUnicodeSubstitution{\ESHMathSymbol{\Re}}}\ESHSpace{}\ESHUnicodeSubstitution{\ESHMathSymbol{\rightarrow}}\ESHSpace{}\textcolor[HTML]{204A87}{\ESHUnicodeSubstitution{\ESHMathSymbol{\Re}}})\ESHSpace{}\ESHUnicodeSubstitution{\ESHMathSymbol{\rightarrow}}\ESHSpace{}\textcolor[HTML]{204A87}{\ESHUnicodeSubstitution{\ESHMathSymbol{\Re}}}\ESHEol
\ESHBol{}argmax01\ESHSpace{}:\ESHSpace{}(\textcolor[HTML]{204A87}{\ESHUnicodeSubstitution{\ESHMathSymbol{\Re}}}\ESHSpace{}\ESHUnicodeSubstitution{\ESHMathSymbol{\rightarrow}}\ESHSpace{}\textcolor[HTML]{204A87}{\ESHUnicodeSubstitution{\ESHMathSymbol{\Re}}})\ESHSpace{}\ESHUnicodeSubstitution{\ESHMathSymbol{\rightarrow}}\ESHSpace{}\textcolor[HTML]{204A87}{\ESHUnicodeSubstitution{\ESHMathSymbol{\Re}}}
\end{ESHBlock}
\end{subfigure}
\hfill
\begin{subfigure}[t]{0.55\textwidth}
\begin{ESHBlock}
\ESHBol{}tangent\ESHSpace{}A\ESHSpace{}B\ESHSpace{}:\ESHSpace{}(A\ESHSpace{}\ESHUnicodeSubstitution{\ESHMathSymbol{\rightarrow}}\ESHSpace{}B)\ESHSpace{}\ESHUnicodeSubstitution{\ESHMathSymbol{\rightarrow}}\ESHSpace{}Tan\ESHSpace{}A\ESHSpace{}\ESHUnicodeSubstitution{\ESHMathSymbol{\rightarrow}}\ESHSpace{}Tan\ESHSpace{}B\ESHEol
\ESHBol{}tangentValue\ESHSpace{}A\ESHSpace{}:\ESHSpace{}Tan\ESHSpace{}A\ESHSpace{}\ESHUnicodeSubstitution{\ESHMathSymbol{\rightarrow}}\ESHSpace{}A\ESHEol
\ESHBol{}\ESHEmptyLine{}\ESHEol
\ESHBol{}record\ESHSpace{}(\ESHUnicodeSubstitution{\ESHMathSymbol{\cong}})\ESHSpace{}A\ESHSpace{}B\ESHSpace{}=\ESHSpace{}\{\ESHSpace{}to\ESHSpace{}:\ESHSpace{}A\ESHSpace{}\ESHUnicodeSubstitution{\ESHMathSymbol{\rightarrow}}\ESHSpace{}B,\ESHEol
\ESHBol{}\ESHSpace{}\ESHSpace{}\ESHSpace{}\ESHSpace{}\ESHSpace{}\ESHSpace{}\ESHSpace{}\ESHSpace{}\ESHSpace{}\ESHSpace{}\ESHSpace{}\ESHSpace{}\ESHSpace{}\ESHSpace{}\ESHSpace{}\ESHSpace{}\ESHSpace{}\ESHSpace{}\ESHSpace{}from\ESHSpace{}:\ESHSpace{}B\ESHSpace{}\ESHUnicodeSubstitution{\ESHMathSymbol{\rightarrow}}\ESHSpace{}A\ESHSpace{}\}\ESHEol
\ESHBol{}tangent\_R\ESHSpace{}:\ESHSpace{}Tan\ESHSpace{}\textcolor[HTML]{204A87}{\ESHUnicodeSubstitution{\ESHMathSymbol{\Re}}}\ESHSpace{}\ESHUnicodeSubstitution{\ESHMathSymbol{\cong}}\ESHSpace{}\textcolor[HTML]{204A87}{\ESHUnicodeSubstitution{\ESHMathSymbol{\Re}}}\ESHSpace{}*\ESHSpace{}\textcolor[HTML]{204A87}{\ESHUnicodeSubstitution{\ESHMathSymbol{\Re}}}\ESHEol
\ESHBol{}tangentProd\ESHSpace{}A\ESHSpace{}B\ESHSpace{}:\ESHSpace{}Tan\ESHSpace{}(A\ESHSpace{}*\ESHSpace{}B)\ESHSpace{}\ESHUnicodeSubstitution{\ESHMathSymbol{\cong}}\ESHSpace{}Tan\ESHSpace{}A\ESHSpace{}*\ESHSpace{}Tan\ESHSpace{}B\ESHEol
\ESHBol{}tangentTo\_R\ESHSpace{}A\ESHSpace{}:\ESHSpace{}Tan\ESHSpace{}(A\ESHSpace{}\ESHUnicodeSubstitution{\ESHMathSymbol{\rightarrow}}\ESHSpace{}\textcolor[HTML]{204A87}{\ESHUnicodeSubstitution{\ESHMathSymbol{\Re}}})\ESHSpace{}\ESHUnicodeSubstitution{\ESHMathSymbol{\cong}}\ESHSpace{}(A\ESHSpace{}\ESHUnicodeSubstitution{\ESHMathSymbol{\rightarrow}}\ESHSpace{}\textcolor[HTML]{204A87}{\ESHUnicodeSubstitution{\ESHMathSymbol{\Re}}})\ESHSpace{}*\ESHSpace{}(A\ESHSpace{}\ESHUnicodeSubstitution{\ESHMathSymbol{\rightarrow}}\ESHSpace{}\textcolor[HTML]{204A87}{\ESHUnicodeSubstitution{\ESHMathSymbol{\Re}}})
\end{ESHBlock}
\end{subfigure}

\vspace{-2em}
\caption[\dlang{} constants and their types.]{\dlang{} constants and their types.}
\label{dlang-constants}
\end{figure}

\dlang{} is the simply-typed lambda calculus with the constants shown in \Cref{dlang-constants}.
These include basic operators, such as arithmetic and trigonometric operators, higher-order operators, and primitives to compute derivatives.
The syntax permits polymorphic type signatures, but semantically we treat polymorphism at the metatheoretic level.

\paragraph{Higher-Order Operators.} The function \ESHInline{\ESHBol{}integral01} gives the Riemannian integral of a function on the interval $[0,1]$.
\ESHInline{\ESHBol{}max01} maximizes a function over the interval $[0, 1]$,
and \ESHInline{\ESHBol{}argmax01} finds its maximizing argument.
\ESHInline{\ESHBol{}cutRoot} finds the root of a function \ESHInline{\ESHBol{}f\ESHSpace{}:\ESHSpace{}\textcolor[HTML]{204A87}{\ESHUnicodeSubstitution{\ESHMathSymbol{\Re}}}\ESHSpace{}\ESHUnicodeSubstitution{\ESHMathSymbol{\rightarrow}}\ESHSpace{}\textcolor[HTML]{204A87}{\ESHUnicodeSubstitution{\ESHMathSymbol{\Re}}}}, assuming that it has a single root and is negative for smaller values and positive for larger values.
\ESHInline{\ESHBol{}firstRoot}, on input \ESHInline{\ESHBol{}f\ESHSpace{}:\ESHSpace{}\textcolor[HTML]{204A87}{\ESHUnicodeSubstitution{\ESHMathSymbol{\Re}}}\ESHSpace{}\ESHUnicodeSubstitution{\ESHMathSymbol{\rightarrow}}\ESHSpace{}\textcolor[HTML]{204A87}{\ESHUnicodeSubstitution{\ESHMathSymbol{\Re}}}},
finds the first root of \ESHInline{\ESHBol{}f} on a region starting at 0.% (if one exists).
%It may not terminate on borderline cases.
%Its derivatives are given by the implicit function theorem.

\paragraph{Derivatives.}
\TBD{Depends on knowledge of forward-mode derivatives and tangent spaces.}
\ESHInline{\ESHBol{}tangent} is a first-class function that computes derivatives, where
the type function \ESHInline{\ESHBol{}Tan} gives the space of \NA{tangent bundles} over a space; conceptually, a space
of pairs of values and derivatives.
The function \ESHInline{\ESHBol{}tangentValue} projects the value part of this tangent bundle.
The isomorphisms of tangent bundles --  i.e., \ESHInline{\ESHBol{}tangent\_R}, \ESHInline{\ESHBol{}tangentProd}, and \ESHInline{\ESHBol{}tangentTo\_R} -- assist with manipulating the information that corresponds to the derivative part of the tangent bundle when it is possible for certain spaces.
To concretize the concept behind these isomorphisms, we now present the implementation of \ESHInline{\ESHBol{}deriv} from~\Cref{fig:diff-ray-tracing},
which uses \ESHInline{\ESHBol{}tangent} and these isomorphisms:
\begin{ESHBlock}
\ESHBol{}\textcolor[HTML]{346604}{let}\ESHSpace{}deriv\ESHSpace{}(f\ESHSpace{}:\ESHSpace{}\textcolor[HTML]{204A87}{\ESHUnicodeSubstitution{\ESHMathSymbol{\Re}}}\ESHSpace{}\ESHUnicodeSubstitution{\ESHMathSymbol{\rightarrow}}\ESHSpace{}\textcolor[HTML]{204A87}{\ESHUnicodeSubstitution{\ESHMathSymbol{\Re}}})\ESHSpace{}(x\ESHSpace{}:\ESHSpace{}\textcolor[HTML]{204A87}{\ESHUnicodeSubstitution{\ESHMathSymbol{\Re}}})\ESHSpace{}:\ESHSpace{}\textcolor[HTML]{204A87}{\ESHUnicodeSubstitution{\ESHMathSymbol{\Re}}}\ESHSpace{}=\ESHEol
\ESHBol{}\ESHSpace{}\ESHSpace{}snd\ESHSpace{}(tangent\_R.to\ESHSpace{}(tangent\ESHSpace{}f\ESHSpace{}(tangent\_R.from\ESHSpace{}(x,\ESHSpace{}1))))
\end{ESHBlock}
This implementation calls \ESHInline{\ESHBol{}tangent} with $f$ and a query for the derivative of $f$ at $x$ in the direction $1$.
The query is a tangent bundle constructed with the isomorphism \ESHInline{\ESHBol{}tangent\_R} from the pair $(x, 1)$.
\ESHInline{\ESHBol{}deriv} then projects out the derivative part of \ESHInline{\ESHBol{}tangent}'s result, using \ESHInline{\ESHBol{}tangent\_R} in the opposite direction and the standard second projection on binary products.

%\ESHInline{\ESHBol{}A\ESHSpace{}\ESHUnicodeSubstitution{\ESHMathSymbol{\cong}}\ESHSpace{}B} indicates an isomorphism: a pair of maps \ESHInline{\ESHBol{}A\ESHSpace{}\ESHUnicodeSubstitution{\ESHMathSymbol{\rightarrow}}\ESHSpace{}B} and \ESHInline{\ESHBol{}B\ESHSpace{}\ESHUnicodeSubstitution{\ESHMathSymbol{\rightarrow}}\ESHSpace{}A} that are mutually inverse.
%However, only in certain spaces can we easily manipulate tangent bundles assist with that.
%  by using \ESHInline{\ESHBol{}tangent\_R} and \ESHInline{\ESHBol{}tangentProd}, we get \ESHInline{\ESHBol{}Tan\ESHSpace{}(\textcolor[HTML]{204A87}{\ESHUnicodeSubstitution{\ESHMathSymbol{\Re}}}\textasciicircum{}n)\ESHSpace{}\ESHUnicodeSubstitution{\ESHMathSymbol{\cong}}\ESHSpace{}\textcolor[HTML]{204A87}{\ESHUnicodeSubstitution{\ESHMathSymbol{\Re}}}\textasciicircum{}n\ESHSpace{}*\ESHSpace{}\textcolor[HTML]{204A87}{\ESHUnicodeSubstitution{\ESHMathSymbol{\Re}}}\textasciicircum{}n} for any \ESHInline{\ESHBol{}n}.

\paragraph{Semantics.} Over the next sections, we develop the full syntax and semantics of \dlang{}. In \Cref{fo-semantics}, we describe a first-order (i.e., no higher-order functions) differentiable language that supports Clarke semantics and higher-order derivatives, by defining a Cartesian monoidal category \cat{AD}.
In \Cref{ho-semantics}, we will define semantics for the higher-order language \dlang{} by taking a category of presheaves, \cat{HAD}, over \cat{AD}.
We defer computability concerns to \Cref{smooth:implementation}.

\section{Semantics of a First-Order Differentiable Language (\cat{AD})}
\label{fo-semantics}

In this section, we describe a first-order (i.e., no higher-order functions) differentiable language that supports Clarke semantics and higher-order derivatives, by defining a Cartesian monoidal category \cat{AD}.
\Cref{smooth:fo-syntax} presents the syntax and typing rules for the language for \cat{AD}.
The $\ast$ type represents the unit type, having a single value $!$ in it.
Given any object $K \in \cat{AD}$, the type expression $\begin{uq}K\end{uq}$ represents the type whose semantics is $K$.
Given any arrow $f : \denote{\tau_1} \sto \denote{\tau_2}$ of \cat{AD} and given some expression $\Gamma \vdash e : \tau_1$ the syntax $\begin{uq}f\end{uq}(e)$ applies the map $f$ to the result of $e$.
When the constants are binary operators like $+$ and $\times$, we permit syntactic sugar to write them infix, such that, e.g., $e_1 + e_2$ is shorthand for $+(e_1, e_2)$.
The syntax $\frac{\partial e_y}{\partial x}\mid_{x = e_x} \cdot e_{dx}$ computes the directional derivative of $e_y$ with respect to $x$ at $x = e_x$ in the direction of infinitesimal perturbation $e_{dx}$.

\providecommand{\rulesep}{\unskip\ \vrule\ }

\begin{figure}
\begin{subfigure}[t]{0.38\textwidth}
\center{\emph{Syntax}}
\begin{align*}
\text{variables } x
\\ \text{types } \tau &\mathbin{::=} \ast \gor \tau_1 \times \tau_2 \gor \begin{uq}K\end{uq}
\\ \text{contexts } \Gamma &\mathbin{::=} \cdot \gor \Gamma, x : \tau
\\ \text{functions } f &\in \mathsf{Arr}(\cat{AD})
\\ \text{expressions } e &\mathbin{::=} x \gor \begin{uq}f\end{uq}(e)
\\ &\gor ! \gor (e, e)
\\ &\gor \frac{\partial e}{\partial x} \mid_{x = e} \cdot e
\\ &\gor \text{let } x \triangleq e \text{ in } e
\end{align*}
\end{subfigure}
\rulesep
\begin{subfigure}[t]{0.58\textwidth}
\center{\emph{Typing rules}}
\begin{mathpar}
\inferrule*
{(x : \tau) \in \Gamma}
{\Gamma \vdash x : \tau}

\inferrule*
{\Gamma \vdash e : \tau_1
\\ f : \denote{\tau_1} \sto \denote{\tau_2}}
{\Gamma \vdash \begin{uq}f\end{uq}(e) : \tau_2}

\inferrule*
{ }
{\Gamma \vdash \,! : \ast}

\inferrule*
{\Gamma \vdash e_1 : \tau_1
\\ \Gamma \vdash e_2 : \tau_2}
{\Gamma \vdash (e_1, e_2) : \tau_1 \times \tau_2}

\inferrule*
{\Gamma \vdash e_1 : \tau_1
\\ \Gamma, x : \tau_1 \vdash e_2 : \tau_2}
{\Gamma \vdash \text{let } x \triangleq e_1 \text{ in } e_2 : \tau_2}

\inferrule*
{\Gamma, x : \tau_1 \vdash e_y : \tau_2
\\ \Gamma \vdash e_x : \tau_1
\\ \Gamma \vdash e_{dx} : \tau_1
}
{\Gamma \vdash \frac{\partial e_y}{\partial x}\mid_{x = e_x} \cdot e_{dx} : \tau_2}
\end{mathpar}
\end{subfigure}
\caption{Syntax and typing rules for the language for \cat{AD}.}
\label{smooth:fo-syntax}
\end{figure}

\begin{figure}
\begin{subfigure}[t]{0.25\textwidth}
\center{\emph{Types}}
\begin{mathpar}
\inferrule*
{\tau \text{ type} }
{\denote{\tau} \in \mathsf{Ob}(\cat{AD})}
\end{mathpar}
\begin{align*}
\denote{\ast} &\triangleq 1_\cat{AD}
\\ \denote{\tau_1 \times \tau_2} &\triangleq \denote{\tau_1} \times \denote{\tau_2}
\\ \denote{\begin{uq}K\end{uq}} &\triangleq K
\end{align*}
\end{subfigure}
\rulesep
\begin{subfigure}[t]{0.23\textwidth}
\center{\emph{Contexts}}
\begin{mathpar}
\inferrule*
{\Gamma \text{ context} }
{\denote{\Gamma} \in \mathsf{Ob}(\cat{AD})}
\end{mathpar}
\begin{align*}
\denote{\cdot} &\triangleq 1_\cat{AD}
\\ \denote{\Gamma, x : \tau} &\triangleq \denote{\Gamma} \times \denote{\tau}
\end{align*}
\end{subfigure}
\rulesep
\begin{subfigure}[t]{0.46\textwidth}
\center{\emph{Terms}}
\begin{mathpar}
\inferrule*
{\Gamma \vdash e : \tau}
{\denote{e} : \denote{\Gamma} \sto \denote{\tau}  }
\end{mathpar}
\begin{align*}
%\denote{\Gamma \vdash x : \tau} &\triangleq \text{projection from $\Gamma$ to $\tau$} \\
\denote{\begin{uq}f\end{uq}(e)} &\triangleq
  f \circ \denote{e}
\\ \denote{!} &\triangleq \,!
\\ \denote{(e_1, e_2)} &\triangleq \langle \denote{e_1}, \denote{e_2} \rangle
\\ \denote{\text{let } x \triangleq e_1 \text{ in } e_2 } &\triangleq \denote{e_2} \circ \langle \text{id}, \denote{e_1} \rangle
\\ \denote{\frac{\partial e_y}{\partial x}\mid_{x = e_x} \cdot e_{dx}} &\triangleq
\denote{e_y}' \circ \langle \langle \mathrm{id}, e_x \rangle , \langle 0, e_{dx} \rangle \rangle
\end{align*}

\end{subfigure}
\caption[Semantics of the language for \cat{AD}.]{Semantics of the language for \cat{AD}.}
\vspace{-0.8em}
\label{smooth:fig:fo-semantics}
\end{figure}

\Cref{smooth:fig:fo-semantics} presents the semantics for the language for \cat{AD}, which we explain in this section.
Our semantics of derivatives is phrased in terms of \emph{Clarke's generalized derivative} \citep{clarke1990},
which enables capturing differentiable properties of locally Lipschitz but nonsmooth functions such as $\max$, $\min$, and $\mathsf{ReLU}$.
We will now present the background material we use to define the semantics of the language for \cat{AD}.

\subsection{Preliminaries}

A \emph{domain} $D$ is a set with a partial-order structure $\sqsubseteq$ that supports \emph{directed joins} $\bigsqcup_{d \in S} d$, which are just joins of \emph{directed} subsets $S \subseteq D$, which are those subsets such that if $x, y \in D$, then there is some $z \in D$ such that $x \sqsubseteq z$ and $y \sqsubseteq z$.
We call the partial-order relation $\sqsubseteq$ \emph{specialization}.
The relation $x \sqsubseteq y$ intuitively means that $x$ behaves in a way that is compatible with how $y$ behaves.
An element $x \in D$ is \emph{maximal} if for any $y \in D$, if $x \sqsubseteq y$, then $y \sqsubseteq x$.

Define
\[
\IR \triangleq \{ [a, b] \mid a, b \in \R, a \le b \} \cup \{ \R \}
\]
as the domain of \emph{interval reals}, partially ordered ($\sqsubseteq$) by reverse set inclusion.
Its maximal elements are the intervals of the form $[a, a]$, which we often just write as $a$.
Arithmetic operations can be extended from $\R$ to $\IR$ (see, e.g., \citet{edalat2004}).
Note that $\R$ serves as a bottom element, and we refer to it with the symbol $\bot$.
For any vector space $V$ (over $\R$), let $\Cvx(V)$ be the set of nonempty convex sets in $V$, with an order relation $\sqsubseteq$ also corresponding to reverse inclusion.
Note that $V$ serves as a bottom element, and we refer to it with the symbol $\bot$.
%Define $\Cvx(V) \triangleq \Cvx(V) \cup \{ V \}$ to add in $V$ as a bottom element (which we will also refer to as $\bot$.
Note that we have the sequence of embeddings $\R^n \hookrightarrow \IR^n \hookrightarrow \Cvx(\R^n)$: every vector $v \in \R^n$ can be treated as a tuple of singleton intervals $\IR^n$, and every element $x \in \IR^n$ can be treated as a (convex) hyperrectangle, where some dimensions of the hyperrectangle may be infinite.
We use the notation $\iota_{\R^n \hookrightarrow \IR^n}$ and $\iota_{\IR^n \hookrightarrow \Cvx(\R^n)}$ to denote these embeddings, respectively.

\paragraph{The Clarke Derivative.}

Let $f : \R^n \to \R^m$. If $f$ is locally Lipschitz on $X \subseteq U$, let $Z_f \subseteq X$ be the points of nondifferentiability of $f$. The \emph{Bouligand subdifferential} of $f$ at $x \in X$ is the \emph{set} of matrices
\[
\partial_Bf(x) \triangleq \left\{ H : \R^{m \times n} \mid
\begin{matrix}
H = \lim_{j \to \infty} Jf(x_j) \text{ for some sequence } (x_j)_{j \in \nat}
\\ \text{ where } x_j \in X \setminus Z_f \text{ for all } j \in \nat \text{ and } \lim_{j \to \infty} x_j = x
\end{matrix}\right\},
\]
where $J$ is the Jacobian operator defining the derivative of a function at a point where it is differentiable.
%\begin{definition}
The \emph{Clarke Jacobian} of $f$ at $x$ is given by the convex hull
$
\partial f(x) \triangleq \text{hull}(\partial_B f(x)).
$
%\end{definition}
The Clarke Jacobian $\partial f(x) \in \Cvx(\R^{m \times n})$ is always compact (since $f$ is locally Lipschitz).

%\begin{definition}
Given $f : \R^n \to {\R^m}_\bot$, let $U$ be the largest open set on which $f$ is both defined and locally Lipschitz. We can define the \emph{partial Clarke Jacobian} of $f$ to   be
\[
\partial_\bot f(x) = \begin{cases}
\partial f(x) & x \in U
\\ \bot & x \notin U
\end{cases}
\]
such that $\partial_\bot : (\R^n \to {\R^m}_\bot) \to \R^n \to \Cvx(\R^{m \times n})$.
%\end{definition}
We can map values of $\Cvx(A)$ to $A_\bot$ (for any $A$) by mapping maximal elements $\{x\} \in   \Cvx(A)$ to $x \in A_\bot$ and everything else to $\bot$.
Using this conversion, we can also give the partial Clarke Jacobian the type
$\partial_\bot : (\R^n \to \Cvx(\R^m)) \to \R^n \to \Cvx(\R^{m \times n})$,
and thus we can also iterate the partial Clarke Jacobian construction to get higher-order derivatives $\partial_\bot^k : (\R^n \to {\R^m}_\bot) \to \R^n \to \Cvx(\R^{m \times n^k})$.
Note that the $k+1$th-order Clarke Jacobian is $\bot$ unless the $k$th-order Clarke Jacobian is maximal; thus, when defined, higher-order Clarke Jacobians are just Clarke Jacobians of conventional higher-order derivatives.

When a function is differentiable, its partial Clarke Jacobian is a maximal element. When it is locally Lipschitz but not differentiable, the partial Clarke Jacobian is a compact convex set. When it is not locally Lipschitz, the partial Clarke Jacobian is the entire space, corresponding to $\bot$.

\subsection{Smoothish Maps}

\NA{We will now define \cat{AD}.
The objects of \cat{AD} are the natural numbers, where $n \in \nat$ corresponds to $n$-dimensional Euclidean space. To emphasize that we are thinking of Euclidean space, we write the object $n \in \nat$ as $\R^n$.} \fTBD{MC: I'd like to delete whatever isn't necessary here and instead just start with the sentence below.}
A morphism of \cat{AD} is a \emph{smoothish map}: a \emph{derivative tower} that is \emph{successively consistent}.
A \emph{derivative tower} $f$ between spaces $\R^n$ and $\R^m$, $f : \cont{\R^n \sto \R^m}$, is a collection of continuous maps (taking the Scott topology for $\IR$)
\[
f^{(k)} : \R^n \times (\R^n)^k \to \IR^m
\]
for each $k \in \nat$, where $f^{(k)}$ represents the $k$th-order derivative.
This defines a smoothish map as a power series, where the first $\R^n$ argument is the point where the map is evaluated, and the remaining $k$ arguments represent the inputs to a multilinear map representing the derivative.\footnote{This representation as \emph{derivative towers} is largely drawn from \citep{elliott-higher-ad}.
%However, \cite{elliott-higher-ad} defines composition of derivative towers incorrectly.
}
Given vectors $x \in \R^n$ and $y \in \R^k$, let $x \otimes y \in \R^{n \times k}$ denote the tensor product.
Define $\text{Mat}_k : (\R^n \times (\R^n)^k \to \IR^m) \to \R^n \to {\R^{m \times n^k}}_\bot$ at a point $x \in \R^n$ such that $\text{Mat}_k(f)(x) = M$ if there is a matrix $M \in \R^{m \times n^k}$ such that for all $dx_1, \ldots, dx_k \in \R^n$, we have
\begin{equation}
  f(x; dx_1, \ldots, dx_k) = \iota_{\R^m \hookrightarrow \IR^m} \left( M \cdot (dx_1 \otimes \ldots \otimes dx_k) \right),
  \label{eq:multilinearity}
\end{equation}
and $\text{Mat}_k(f)(x) = \bot$ if there is no such matrix ($M$).
Equation~\ref{eq:multilinearity} requires that $f$ is multilinear in its $dx_1, \ldots, dx_k$ arguments, which means that $f$ has a form that permits differentiation.

\newcommand{\hatdot}{\mathbin{\hat\cdot}}
\newcommand{\hatotimes}{\mathbin{\hat\otimes}}

\begin{definition}
\label{derivative-tower-consistent}
We define a consistency relation $\text{Cons}_k(g, f)$ for a function $g : \R^n \times (\R^n)^{k} \to \IR^m$ and a function $f : \R^n \to \Cvx(\R^{m \times n^k})$ to hold if
for all $x \in \R^n$ and for all $dx_1, \ldots, dx_{k} \in \R^n$,
\[
\iota_{\IR^m \hookrightarrow \Cvx(\R^m)}\left(
g(x; dx_1, \ldots, dx_{k})
\right)
\sqsubseteq
f(x) \cdot (dx_1 \otimes \ldots \otimes dx_{k}).
\]
A derivative tower $f$ is \emph{successively consistent} if for all $k \in \nat$, we have
\[
\text{Cons}_{k+1}(f^{(k+1)}, \partial_\bot \text{Mat}_k(f^{(k)})),
\]
meaning that each successive derivative $f^{(k + 1)}$ is consistent with the value-level behavior of $f^{(k)}$.
\end{definition}

A \emph{smoothish map} $f$ is a successively consistent derivative tower.
We call a smoothish map \emph{smooth} if $f^{(k)}$ is maximal for all $k$ (which agrees with the standard definition of a smooth map).
We will later show (in \Cref{smooth:categorical-operations}) that smoothish maps form a category \cat{AD}, and then by categorical semantics, that all expressions in the first-order language map to that category.

\subsection{Primitives}

Any first-order primitive may be implemented by giving its power-series representation.
We use the notation $f^{(k)}(x; \vec{v})$ to denote the $k$th derivative of $f$ at $x$ in directions $\vec{v}$; a smoothish map $f$ is defined by the collection of these functions for all $k \in \nat$.
These data provide power-series expansions around any input point.
There is a map $\cont{0 : \Gamma \sto A}$ (for any $\Gamma, A \in \cat{AD}$) that always returns zero regardless of its input.
A linear map $\cont{f : A \to B}$ determines a smooth map $\mathsf{linear}(f) : A \sto B$ by
\begin{align*}
\mathsf{linear}(f)^{(0)}(x) &\triangleq f(x)
\\
\mathsf{linear}(f)^{(1)}(x; v) &\triangleq f(v)
\\
\mathsf{linear}(f)^{(k + 2)}(x; \vec{v}) &\triangleq 0
\end{align*}

\paragraph{Derivative-Tower Construction.}

A derivative tower can be viewed as a stream of a function and all of its derivatives.
Streams are characterized by the isomorphism
\[
\mathsf{Stream}(A) \cong A \times \mathsf{Stream}(A)
\] that says that a stream $s : \mathsf{Stream}(A)$ is exactly composed of its head, $\mathsf{head}(s) : A$, and its tail, $\mathsf{tail}(s) : \mathsf{Stream}(A)$. To construct a derivative tower, we define the map
$
\cont{\mathsf{foldDer}}% : (A \cto B) \times (A \times A \sto B) \cto (A \sto B)
$
as an analogue to the \emph{cons} operation on streams. For instance, given value-level definitions of sine and cosine, $\sin$ and $\cos$,
it is well-founded to define their derivative towers as
\begin{align*}
\begin{qvect}\ESHInline{\ESHBol{}sin}\end{qvect} &\triangleq \mathsf{foldDer}(\sin,\begin{qvect}x, dx \vdash \ESHInline{\ESHBol{}cos}(x) \:\ESHInline{\ESHBol{}*}\: dx \end{qvect})
\\ \begin{qvect}\ESHInline{\ESHBol{}cos}\end{qvect} &\triangleq \mathsf{foldDer}(\cos, \begin{qvect}x, dx \vdash -\ESHInline{\ESHBol{}sin}(x) \:\ESHInline{\ESHBol{}*}\: dx\end{qvect}),
\end{align*}
just as it would be to define two mutually recursive streams $\mathsf{evens} = \mathsf{cons}(0, \mathsf{map}(\lambda x.\ x + 1, \mathsf{odds}))$ and $\mathsf{odds} = \mathsf{cons}(1, \mathsf{map}(\lambda x.\ x + 3, \mathsf{evens})))$.

We define $\mathsf{foldDer}$ as follows, where $f : A \to B$ and $g : A \stimes A \sto B$, such that $\mathsf{foldDer}(f, g) : A \sto B$.
\begin{align*}
\mathsf{foldDer}(f, g)^{(0)}(x) &\triangleq f(x)
\\ \mathsf{foldDer}(f, g)^{(k+1)}(x; v_1, \ldots, v_{k + 1}) &\triangleq
g^{(k)}((x, v_1); (v_2, 0), \ldots, (v_{k+1}, 0)) \qquad\qquad \qquad (k \in \nat)
\end{align*}
One of the perturbations $\cont{v_1}$ is passed in as the value to $\cont{g}$, and then that perturbation is not considered to have any derivatives itself, hence the 0s in the second components of the perturbation passed to $g$.
Setting the first components of the derivatives to $v_2, \ldots, v_{k+1}$ establishes these as independent infinitesimal perturbations of the first value component, $x$.

\subsubsection{Arithmetic Operations}
\label{sec:arith}

The binary arithmetic operations are first-order functions and so can be represented in \cat{AD} as functions with the type $\begin{contc}\R \stimes \R \sto \R\end{contc}$.
Addition and subtraction are linear, so their semantics is simply
$\begin{qvect}\ESHInline{\ESHBol{}+}\end{qvect} \triangleq \mathsf{linear}(+)$
and $\begin{qvect}\ESHInline{\ESHBol{}\ESHDash{}}\end{qvect} \triangleq \mathsf{linear}(-)$.
We define the smooth multiplication operator by
\begin{align*}
\begin{qvect}\ESHInline{\ESHBol{}*}\end{qvect}
&
\triangleq
\mathsf{foldDer}(
\lambda (x, y).\ x \times y
,
\begin{qvect}
(x, y), (dx, dy) \vdash
x \:\ESHInline{\ESHBol{}*}\: dy
\:\ESHInline{\ESHBol{}+}\:
y \:\ESHInline{\ESHBol{}*}\: dx
\end{qvect}
),
\end{align*}
whose derivative is the familiar product rule.
Note that our definition of $\begin{qvect}\ESHInline{\ESHBol{}*}\end{qvect}$ has two recursive references to multiplication's own \emph{smooth} map.
This recursive reference is well-founded because the result is used in a way that does not demand any further differentiation.
This recursive pattern is similar to defining the stream of natural numbers $\mathsf{nats} : \mathsf{Stream}(\nat)$ by
\[
\mathsf{nats} \triangleq \mathsf{cons}(0, \mathsf{map}\ (\lambda x.\ x + 1)\ \mathsf{nats}),
\]
where mapping a function over $\mathsf{nats}$ does not demand any further calls to $\mathsf{tail}$.
Reciprocals (used for division) can be defined using $\mathsf{foldDer}$ as well, where all $k$th-order derivatives will return $\bot$ when the input is 0.

\subsubsection{Lipschitz but Nonsmooth Functions}
Many functions, such as $\max$, $\min$, and $\mathsf{ReLU}$, are locally Lipschitz but not smooth.
These functions are used pervasively in contexts that require differentiation,
so their admissibility in a differential-programming semantics is paramount.
Whereas most differential-programming semantics say that derivative of $\max$ is undefined when its arguments are equal, our use of Clarke derivatives permits a non-$\bot$ result.
\fTBD{JM: mention applications using, e.g., ReLU}

We define \ESHInline{\ESHBol{}max} as follows, where \text{hull} computes the interval corresponding to the convex hull of the union of a set of points.
\begin{align*}
\begin{qvect}\mathsf{max}\end{qvect}^{(0)}(x, y) &\triangleq \max(x, y)
\\
\begin{qvect}\mathsf{max}\end{qvect}^{(1)}((x, y); (dx, dy)) &\triangleq
\begin{cases}
dx & x > y
\\ dy & y < x
\\ \text{hull}(\{dx, dy\}) & x = y
\end{cases}
\\
\begin{qvect}\mathsf{max}\end{qvect}^{(k + 2)}((x, y); \vec{v}) &\triangleq
\begin{cases}
0 & x \neq y
\\ \bot & x = y
\end{cases}
\end{align*}

\subsubsection{Differentiation Operator}

To give a semantics to the syntax $\frac{\partial e_y}{\partial x}\mid_{x = e_{x}} \cdot e_{dx}$ for differentiation, we first define a differentiation operator, postfix $'$, on smoothish maps, where $f : A \sto B$ maps to $f' : A \stimes A \sto B$.
Defining this operator is nontrivial, because all the derivatives of $\cont{f'}$ must consider not only perturbations to the function value but also perturbations to the derivative argument,
which are not accounted for in the original derivative tower:
note that the $k$th derivative of $\cont{f}$ is a multilinear map from $\cont{A^k}$, whereas the $k$th derivative of $\cont{f'}$ is a multilinear map from $\cont{A^{2k}}$.
We show the value and first few derivatives; because $x$ will always be applied as the value argument to derivatives of $f$, we elide those arguments:
\begin{align*}
{f'}^{(0)}(x, v) &= f^{(1)}(v)
\\
{f'}^{(1)}((x, v); (dx_a, dv_a)) &=
f^{(2)}(v, dx_a)
+ f^{(1)}(dv_a)
\\
{f'}^{(2)}((x, v); (dx_a, dv_a), (dx_b, dv_b)) &=
f^{(3)}(v, dx_a, dx_b)
+ f^{(2)}(dv_a, dx_b)
+ f^{(2)}(dx_a, dv_b)
%\\
%{f'}^{(3)}((x, v); (dx_a, dv_a), (dx_b, dv_b), (dx_c, dv_c)) &=
%f^{(4)}(v, dx_a, dx_b, dx_c)
%+ f^{(3)}(dv_a, dx_b, dx_c)
%+ f^{(3)}(dx_a, dv_b, dx_c)
%+ f^{(3)}(dx_a, dx_b, dv_c)
\end{align*}
The general formula is:
\begin{align*}
&{f'}^{(k)}((x, v); (dx_1,dv_1),\ldots,(dx_k,dv_k)) \triangleq
\\ & \qquad f^{(k + 1)}(x; v, dx_1, \ldots, dx_k)
  + \sum_{j = 1}^k f^{(k)}(x; dx_1, \ldots, dx_{j-1}, dv_j, dx_{j+1}, \ldots dx_k).
\end{align*}

\TBD{Explain differentiation syntax and give its semantics}

\subsubsection{Revisiting Derivative Tower Construction}

The $'$ operator is analogous to the $\mathsf{tail}$ operator of a stream,
in that derivative towers have the section-retraction pair
\begin{center}
\begin{tikzcd}[sep=7em]
A \sto B
\arrow[r, tail, "\lambda f. (f^{(0)}{,} f')", shift left]
& (A \to B) \times (A \stimes A \sto B)
\arrow[l, two heads, "\mathsf{foldDer}", shift left]
\end{tikzcd}
\end{center}
that characterizes a derivative tower $\cont{f : A \sto B}$ as a function $\cont{f^{(0)} : A \to B}$ for the evaluation map of $f$ together with a derivative tower $\cont{f' : A \stimes A \sto B}$ where $f'(x, v)$ represents the directional derivative of $f$ at $x$ in the direction $v$.

%taking as input a function value as well as an infinitesimal perturbation.

Given this observation, we may for convenience in the rest of the paper define a smoothish map $f$ by its value-level function $f^{(0)}$ and its smoothish derivative $f'$, denoting an implicit use of $\mathsf{foldDer}$. For example, we can equivalently define the smooth multiplication operator (\Cref{sec:arith}) by
\begin{align*}
\begin{qvect}\ESHInline{\ESHBol{}*}\end{qvect}^{(0)}
&
\triangleq
\lambda (x, y).\ x \times y
\\
\begin{qvect}\ESHInline{\ESHBol{}*}\end{qvect}'
&
\triangleq
\begin{qvect}
(x, y), (dx, dy) \vdash
x \:\ESHInline{\ESHBol{}*}\: dy
\:\ESHInline{\ESHBol{}+}\:
y \:\ESHInline{\ESHBol{}*}\: dx
\end{qvect}
.
\end{align*}
%This characterization is convenient for defining smooth maps.

\subsection{Categorical Operations}
\label{smooth:categorical-operations}

\cat{AD} forms a Cartesian monoidal category.
We describe the categorical operations here, and prove that they satisfy the expected properties in \Cref{smooth:ad-category}.
The maps $\mathrm{id} : A \to A$ (for all $A$), $\begin{contc}! : \Gamma \to \ast\end{contc}$ (for all $\cont{\Gamma}$), $\begin{contc}\mathsf{fst} : A \stimes B \to A\end{contc}$ and $\begin{contc}\mathsf{snd} : A \stimes B \to B\end{contc}$ (for all $\begin{contc}A, B\end{contc}$) are all in fact linear maps and so can be made into smooth maps with the $\cont{\mathsf{linear}}$ operator described above.\fTBD{MC: missing fst and snd in syntax}
Given $\begin{contc}f : \Gamma \sto A\end{contc}$ and $\begin{contc}g : \Gamma \sto B\end{contc}$, we define their product $\begin{contc}\langle f, g \rangle : \Gamma \sto A \stimes B\end{contc}$ by
\[
\langle f, g \rangle^{(k)}(x; \vec{v}) \triangleq (f^{(k)}(x; \vec{v}), g^{(k)}(x; \vec{v}))
,
\]

It only remains to define composition.
Composition of smooth maps is given by Faà di Bruno's formula.
The definition is perhaps easier to understand by example for small $k$.
The following shows derivatives of $\begin{contc}g \circ f\end{contc}$ at $\cont{x}$; since $\cont{g}$ is always differentiated at $\cont{f(x)}$ and $\cont{f}$ is always differentiated at $\cont{x}$, we elide those arguments:
\begin{align*}
(g \circ f)^{(0)}() &= g^{(0)}()
\\
(g \circ f)^{(1)}(v_a) &= g^{(1)}(f^{(1)}(v_a))
\\
(g \circ f)^{(2)}(v_a, v_b) &= g^{(2)}(f^{(1)}(v_a), f^{(1)}(v_b)) + g^{(1)}(f^{(2)}(v_a, v_b))
\\
(g \circ f)^{(3)}(v_a, v_b, v_c) &=
g^{(3)}(f^{(1)}(v_a), f^{(1)}(v_b), f^{(1)}(v_c))
\\ &+ g^{(2)}(f^{(2)}(v_a, v_b), f^{(1)}(v_c)) + g^{(2)}(f^{(2)}(v_a, v_c), f^{(1)}(v_b))
\\ &+ g^{(2)}(f^{(2)}(v_b, v_c), f^{(1)}(v_a))
+ g^{(1)}(f^{(3)}(v_a, v_b, v_c))
\end{align*}
The general formula is
\begin{align*}
(g \circ f)^{(k)}(x; \vec{v}) &\triangleq
\sum_{\pi \in \mathcal{H}(\{1, \ldots, k\})}
\text{let } n \triangleq |\pi| \text{ in }
g^{(n)}\left(f(x);
\begin{matrix}
f^{(|\pi_1|)}(x; v_{{\pi_1}_1}, \ldots, v_{{\pi_1}_{|\pi_1|}}),
\\ \vdots,
\\ f^{(|\pi_{n}|)}(x; v_{{\pi_{n}}_1}, \ldots, v_{{\pi_{n}}_{\left|\pi_{n}\right|}})
\end{matrix}
\right)
,
\end{align*}
where $\mathcal{H}(S)$ is the set of partitions of a set $S$, and $|S|$ is the cardinality of a set.
Note that in the general case, the inputs to $g^{(n)}$ may be elements of $\IR^b$ rather than $\R^b$ (for some $b \in \nat$).
Given any $n$th derivative $g^{(n)} : \R^b \times (\R^b)^k \to \IR^m$, we extend it to apply to inputs $x \in \IR^b$ and $dx_1, \ldots, dx_k \in \IR^b$ by
\[
g^{(n)}(x; dx_1, \ldots, dx_k) \triangleq \mathrm{hull}\left\{ g^{(n)}(y; dy_1, \ldots, dy_k) \mid y \in x, dy_1 \in dx_1, \ldots, dy_k \in dx_k \right\}.
\]

Faà di Bruno's formula simplifies drastically in the case that either function is linear:
\begin{proposition}
\label{linear-postcompose}
For any $g : B \to C$ and any derivative tower $f : A \sto B$,
for any $k \in \nat$ and any $x \in A$ and $v_1, \ldots, v_k \in A$,
\[
(\mathsf{linear}(g) \circ f)^{(k)}(x; v_1, \ldots, v_k) = g(f^{(k)}(v_1, \ldots, v_k))
\]
\end{proposition}
\begin{proof}[Proof sketch]
Because $\mathsf{linear}(g)^{(j)}(v_1, \ldots, v_j) = 0$ whenever $j > 1$ by definition of $\mathsf{linear}$,
all terms in the sum given by the Faà di Bruno formula where $|\pi| > 1$ will be 0.
We can thus remove those terms, and the only term in the sum that will remain is the one where $|\pi| = 1$.
\end{proof}
\begin{proposition}
\label{linear-precompose}
For any consistent derivative tower $g : B \sto C$ and any $f : A \to B$ that maps maximal elements to maximal elements,
for any $k \in \nat$ and any maximal $x \in A$ and maximal $v_1, \ldots, v_k \in A$,
\[
(g \circ \mathsf{linear}(f))^{(k)}(x; v_1, \ldots, v_k) = g^{(k)}(f(v_1), \ldots, f(v_k)).
\]
\end{proposition}
\begin{proof}[Proof sketch]
Note that the term in the sum given by the Faà di Bruno formula where $|\pi| = k$ gives the right-hand side $g^{(k)}(f(v_1), \ldots, f(v_k))$.
For all other terms in the sum, where $|\pi| < k$, we have that one of the inputs to $g^{(|\pi|)}$ will be 0, because we have $\mathsf{linear}(f)^{(j)}(v_1, \ldots, v_j) = 0$ whenever $j > 1$ by definition of $\mathsf{linear}$.

We need to know that adding all these terms to the term $|\pi| = k$ makes no difference to the sum,
which can happen either if all of the terms are 0, or if already $g^{(k)}(f(v_1), \ldots, f(v_k)) = \bot$,
in which case the addition of any elements will not change the result.
Thus, it suffices to prove that if $g^{(k)}(f(v_1), \ldots, f(v_k)) \ne \bot$, then all of those other terms in the sum are 0.
A detailed technical argument can show that this is the case.
\end{proof}

The chain rule for Clarke derivatives is a specialization relation rather than an equality:
\begin{proposition}[Chain rule for $\partial_\bot$]
\label{chain-rule-partial-clarke}
Given $f : \R^n \to {\R^m}_\bot$,
and $g : \R^m \to {\R^k}_\bot$
for all $x \in \R^n$ and all $dx \in \R^n$,
\[
\mathrm{hull}\left(\left\{ G \cdot F \cdot dx \mid G \in \left(\partial_\bot g\right)(f(x)), F \in \partial_\bot f(x) \right\}\right)
\sqsubseteq
\partial_\bot(g \circ f)(x) \cdot dx.
\]
\end{proposition}
\begin{proof}[Proof sketch]
A minor variation of \cite[Corollary on page 75]{clarke1990}.
\end{proof}
For example, at the value level, \ESHInline{\ESHBol{}max\ESHSpace{}x\ESHSpace{}0\ESHSpace{}+\ESHSpace{}min\ESHSpace{}0\ESHSpace{}x\ESHSpace{}=\ESHSpace{}x}, but the derivative of the left-hand side at 0 is $[0, 2]$ while the derivative at the right-hand side is $1$, noting $[0,2] \sqsubseteq 1$.
This has important ramifications for \dlang{}, where we construct functions as compositions of others and need composition to be computable.
Because of the specialization relation, we know that any behavior of a function in \dlang{} (e.g., $[0, 2]$) will be \emph{compatible} with the ideal derivative of its value-level function (e.g., $1$), but it may not return the maximal such value.

\begin{proposition}
\label{smooth:ad-category}
These operations (identity, composition, pairing, projections) give \cat{AD} the structure of a Cartesian monoidal category.
Therefore, \cat{AD} admits the internal language described in \Cref{smooth:fo-syntax}.
\end{proposition}
\begin{proof}[Proof sketch]
There are two main classes of properties we must confirm about these categorical operations.
First, we must verify that all of the operations preserve successive consistency, taking consistent derivative towers to consistent derivative towers.
Second, we must confirm that the algebraic laws of a Cartesian monoidal category.
\begin{enumerate}[leftmargin=1em]
\item[1.] Operations preserve consistency. Because several of the categorical operations are of the form $\mathsf{linear}(f)$ we first prove a lemma that these maps are consistent:
\begin{lemma}
\label{linear-consistent}
Call a map $f : \R^n \to \IR^k$ \emph{linear} if it always outputs values in $\R^k$ and if it is linear in the traditional sense, i.e., $f(u) + f(v) = f(u + v)$ for all $u, v \in \R^n$ and $c \cdot f(v) = f(c \cdot v)$ for all $c \in \R$ and all $v \in \R^n$.
Whenever $f : \R^n \to \IR^k$ is linear, $\mathsf{linear}(f)$ is consistent.
\end{lemma}
\begin{proof}
Since $f$ is linear in the above-defined sense, it is smooth, and so its derivatives will always be maximal, and will coincide with the traditional derivatives, which is exactly what $\mathsf{linear}(f)$ computes.
\end{proof}
\begin{itemize}[leftmargin=0.5em]
\item \emph{Identity maps are consistent.}
Follows from \Cref{linear-consistent}.
\item \emph{Product projections are consistent.}
Also follows from \Cref{linear-consistent}.
\item \emph{Pairing preserves consistency.}
Essentially reduces to the following lemma:
\begin{proposition}
\label{pairing-consistency-one}
Given two maps $f : \R^n \to {\R^m}_\bot$ and $g : \R^n \to {\R^k}_\bot$,
for any $x \in \R^n$,
\[
\{ \begin{bmatrix} u & v \end{bmatrix} \mid u \in \partial_\bot f (x), v \in \partial_\bot g (x) \}
\sqsubseteq
\partial_\bot(\lambda z. (f(z), g(z)))(x),
\]
where the pairing operation $(\cdot, \cdot) : {\R^m}_\bot \times {\R^k}_\bot \to {\R^{m + k}}_\bot$ returns $\bot$ if either of its arguments it $\bot$, or the pair of values if both inputs are not $\bot$.
\end{proposition}
\begin{proof}[Proof sketch]
Note that the set defined by the set comprehension on the left-hand side of the relation is convex, since both $\partial_\bot f (x)$ and $\partial_\bot g (x)$ are.
Suppose $\begin{bmatrix} H & L \end{bmatrix}$ is in the Bouligand subdifferential of $\lambda z. (f(z), g(z))$ at $x$.
Then it must be the case that $H$ is in the Bouligand subdifferential for $f$ and that $L$ is in the Bouligand subdifferential for $g$.
\end{proof}
\item \emph{Composition preserves consistency.}
The full proof is quite detailed and technical.
At its core, the proof proceeds much like the proof of the conventional Faà di Bruno formula, which can proceed by induction on the order of differentiation. However, whereas conventionally there is an equality between the Faá di Bruno formula and the derivative, in our case, their is an order relation that the Faá di Bruno formula is at most the derivative.
The base case is the general chain rule of calculus, which in our case corresponds to the chain rule for Clarke derivatives, \Cref{chain-rule-partial-clarke}.
The key step in the inductive case is the tensor product rule:
\begin{proposition}[Tensor product rule for $\partial_\bot$]
\label{tensor-product-rule-partial-clarke}
Given $g : D \to {\R^{m \times n_j \times \ldots \times n_1}}_\bot$ and for all $i \in \{1, \ldots, j \}$,
$f_i : D \to {\R^{n_i}}_\bot$,
for all $x \in D$,
% Margin friendlier alternatives
% \bigotimes_{k=1}^j f_k(x)
% \partial_\bot g(x) \cdot (f_1(x) \otimes \ldots \otimes f_j(x))
% + g(x) \cdot \sum_{i = 1}^j f_1(x) \otimes \ldots \otimes f_{i-1}(x) \otimes \partial_\bot f_i(x) \otimes f_{i + 1}(x) \otimes \ldots \otimes f_j(x)
\begin{gather*}
\partial_\bot g(x) \cdot \left(\bigotimes_{k=1}^j f_k(x)\right)
+ g(x) \cdot \sum_{i = 1}^j \left(\bigotimes_{k=1}^{i-1} f_k(x)\right) \otimes \partial_\bot f_i(x) \otimes \left(\bigotimes_{k=i+1}^{j} f_k(x)\right)
\\ \sqsubseteq
\\ \partial_\bot(g \cdot (f_1 \otimes \ldots \otimes f_j))(x).
\end{gather*}
\end{proposition}
\begin{proof}
\TBD{Now need to show that the tensor products are the same as the nested products}
By repeated application of the product rule for Clarke derivatives.
\end{proof}
\end{itemize}

\item[2.] Algebraic laws hold.
\begin{itemize}[leftmargin=0.5em]
\item \emph{Composition is associative.}
Follows from associativity and commutativity of $+$ and the associativity of taking partitions of partitions (in the appropriate sense).
\item \emph{$f \circ \mathrm{id} = f = \mathrm{id} \circ f$}.
Follows from \Cref{linear-precompose} and \Cref{linear-postcompose}, since $\mathrm{id}$ is linear.
\item \emph{$\beta$ and $\eta$ laws for product projections}. Follows from the fact that $\mathsf{linear}$ commutes with pairing, i.e.,
 $
\mathsf{linear}(\langle f, g \rangle)
=
\langle \mathsf{linear}(f), \mathsf{linear}(g) \rangle,
$
and from \Cref{linear-precompose} and \Cref{linear-postcompose}. \popQED
\end{itemize}
\end{enumerate}
\end{proof}

\subsection{Consistency}

The derivatives that our semantics defines are \emph{consistent}: the behaviors of $k$th derivative that is computed, $\begin{qvect}e\end{qvect}^{(k)}$, are compatible with the derivatives that would be abstractly defined by looking at its value-level behavior, $\partial_\bot^k \text{Mat}_0(\begin{qvect}e\end{qvect}^{(0)})$.
This proposition follows by first demonstrating that derivative towers are \emph{successively consistent}.

\begin{proposition}
\label{terms-consistent}
Given any term $\Gamma \vdash e : \tau$, the derivative tower $\begin{qvect}e\end{qvect}$ is successively consistent.
\end{proposition}
\begin{proof}[Proof sketch]
By induction on the typing derivation of $e$.
We then see that, to know the proposition is true, we must know that the derivative towers for all primitives are consistent (including product projections)
and that pairing and composition preserve successive consistency (proof sketch in \Cref{smooth:ad-category}).
\end{proof}

\begin{proposition}[Consistency of differentiation in the first-order language]
\label{fo-soundness}
Given any term $\Gamma \vdash e : \tau$, for all $k \in \nat$, $\mathrm{Cons}_{k+1}\left(\begin{qvect}e\end{qvect}^{(k+1)} , \partial_\bot^{k+1} \mathrm{Mat}_0(\begin{qvect}e\end{qvect}^{(0)}) \right)$.
\end{proposition}
\begin{proof}[Proof sketch]
By \Cref{terms-consistent}, $\begin{qvect}e\end{qvect}$ is successively consistent.
Then the proof proceeds by a simple induction on $k$.
\end{proof}

\section{Higher-Order Semantics (\cat{HAD})}
\label{ho-semantics}

\begin{figure}
\begin{subfigure}[t]{0.4\textwidth}
\center{\emph{Syntax}}
\begin{align*}
\text{variables } x
\\ \text{constant types } K &\in \mathsf{Ob}(\cat{HAD})
\\ \text{types } \tau &\mathbin{::=} \ast \gor \tau_1 \times \tau_2 \gor \tau_1 \to \tau_2
\\ &\gor \begin{uq}K\end{uq}
\\ \text{contexts } \Gamma &\mathbin{::=} \cdot \gor \Gamma, x : \tau
\\ \text{constants } k &\in \mathsf{Arr}(\cat{HAD})
\\ \text{expressions } e &\mathbin{::=} x \gor \begin{uq}k\end{uq} \gor e\ e \gor \lambda x.\ e
\\ &\gor ! \gor (e, e)
\\ &\gor \text{let } x \triangleq e \text{ in } e
\end{align*}
\end{subfigure}
\rulesep
\begin{subfigure}[t]{0.5\textwidth}
\center{\emph{Typing rules}}
\begin{mathpar}
\inferrule*
{(x : \tau) \in \Gamma}
{\Gamma \vdash x : \tau}

\inferrule*
{\Gamma \vdash e_1 : \tau_1 \to \tau_2
\\ \Gamma \vdash e_2 : \tau_1}
{\Gamma \vdash e_1 \, e_2 : \tau_2}

\inferrule*
{\Gamma, x : \tau_1 \vdash e : \tau_2}
{\Gamma \vdash \lambda\ x : \tau_1.\ e : \tau_1 \to \tau_2}

\inferrule*{k \in \denote{\Gamma} \to_\cat{HAD} \denote{\tau}}
{\Gamma \vdash \begin{uq}k\end{uq} : \tau}

\inferrule*
{ }
{\Gamma \vdash \,! : \ast}

\inferrule*
{\Gamma \vdash e_1 : \tau_1
\\ \Gamma \vdash e_2 : \tau_2}
{\Gamma \vdash (e_1, e_2) : \tau_1 \times \tau_2}

\inferrule*
{\Gamma \vdash e_1 : \tau_1
\\ \Gamma, x : \tau_1 \vdash e_2 : \tau_2}
{\Gamma \vdash \text{let } x \triangleq e_1 \text{ in } e_2 : \tau_2}
\end{mathpar}
\end{subfigure}
\caption[Syntax and typing rules for \dlang{}.]{Syntax and typing rules for \dlang{}. The constants are those listed in \Cref{dlang-constants}.}
\label{smooth:ho-syntax}
\vspace{-.25cm}
\end{figure}

\begin{figure}
\begin{subfigure}[t]{0.6\textwidth}
\small
\center{\emph{Types}}
\begin{align*}
\denote{\ast}(\Gamma) &\triangleq 1_\cat{Set}
\\ \denote{\tau_1 \times \tau_2}(\Gamma) &\triangleq \denote{\tau_1}(\Gamma) \times \denote{\tau_2}(\Gamma)
\\ \denote{\begin{uq}K\end{uq}}(\Gamma) &\triangleq K(\Gamma)
\\ \denote{\tau_1 \to \tau_2}(\Gamma) &\triangleq \int_{\Delta \in \cat{AD}} (\Delta \sto \Gamma) \times  \denote{\tau_1}(\Delta) \to \denote{\tau_2}(\Delta)
\end{align*}
\end{subfigure}
%\hfill
\rulesep
%\hfill
\begin{subfigure}[t]{0.38\textwidth}
\small
\center{\emph{Terms}}
\begin{align*}
%\denote{\Gamma \vdash x : \tau} &\triangleq \text{projection from $\Gamma$ to $\tau$} \\
\denote{e_1 \, e_2}(\gamma)&\triangleq
 \denote{e_1}(\gamma)(\text{id}, \denote{e_2}(\gamma))
\\ \denote{\lambda\ x : \tau_1.\ e} &\triangleq \mathsf{abstract}(\denote{e})
\\ \denote{\begin{uq}k\end{uq}}(\gamma) &\triangleq k(\gamma)
\\ \denote{!} &\triangleq !
\\ \denote{(e_1, e_2)}(\gamma) &\triangleq (\denote{e_1}(\gamma), \denote{e_2}(\gamma))
\\ \denote{\text{let } x \triangleq e_1 \text{ in } e_2} &\triangleq \denote{(\lambda x.\ e_2)\ e_1}
\end{align*}
\end{subfigure}
\caption[Semantics of \dlang{}.]{The semantics of \dlang{}. }
\label{smooth:fig:ho-semantics}
\vspace{-0.8em}
\end{figure}

The category \cat{AD} does not admit exponentiation (function spaces), since its objects are limited to $\R^n$.
However, higher-order functions yield novel expressive power that is critical for \Cref{smooth:applications}.
To admit higher-order functions, \dlang{} uses a category \cat{HAD} of \emph{presheaves} over \cat{AD}, i.e., $\cat{HAD} = [\cat{AD}^\mathrm{op}, \cat{Set}]$.

%For instance, \dlang{} has the primitive higher-order function \ESHInline{\ESHBol{}integral01\ESHSpace{}:\ESHSpace{}(\textcolor[HTML]{204A87}{\ESHUnicodeSubstitution{\ESHMathSymbol{\Re}}}\ESHSpace{}\ESHUnicodeSubstitution{\ESHMathSymbol{\rightarrow}}\ESHSpace{}\textcolor[HTML]{204A87}{\ESHUnicodeSubstitution{\ESHMathSymbol{\Re}}})\ESHSpace{}\ESHUnicodeSubstitution{\ESHMathSymbol{\rightarrow}}\ESHSpace{}\textcolor[HTML]{204A87}{\ESHUnicodeSubstitution{\ESHMathSymbol{\Re}}}} that computes the integral of a function over the unit interval $[0, 1]$.

\paragraph{Syntax and Semantics.} The basic syntax of \cat{HAD} is that of the simply-typed lambda calculus, shown in \Cref{smooth:ho-syntax}, where the constants are those listed in \Cref{dlang-constants}.
\Cref{smooth:fig:ho-semantics} presents the semantics of \dlang{} (generic to any Cartesian closed category of presheaves). However, the categorical semantics in \cat{HAD} means that \dlang{} is inherently extensible and not limited to just those constants in \Cref{dlang-constants}; any object or morphism in \cat{HAD} could be added to the language.

We now proceed to describe the semantics of the higher-order constants in \Cref{dlang-constants}.

\subsection{Ground Types and First-Order Primitives}
Any space $X \in \cat{AD}$ can be lifted into a presheaf \cat{HAD} by the \emph{Yoneda embedding}, $\y{X} \in \cat{HAD}$, which acts as $\y{X}(\Gamma) \triangleq \Gamma \sto X$.
Because the Yoneda embedding is full and faithful and preserves products,
ground types (and their products) in \cat{HAD} represent Cartesian spaces and first-order functions represent smoothish maps.
In particular, $\denote{\ESHInline{\ESHBol{}\textcolor[HTML]{204A87}{\ESHUnicodeSubstitution{\ESHMathSymbol{\Re}}}}} \triangleq \y{\IR}$.
Note that all first-order functions from \cat{AD} can be lifted into \cat{HAD} by the Yoneda embedding.

\subsection{Smoothish Higher-Order Primitive Functions}
\label{higher-order-primitives}

Each of the smoothish higher-order primitive functions has a type \ESHInline{\ESHBol{}(\textcolor[HTML]{204A87}{\ESHUnicodeSubstitution{\ESHMathSymbol{\Re}}}\ESHSpace{}\ESHUnicodeSubstitution{\ESHMathSymbol{\rightarrow}}\ESHSpace{}\textcolor[HTML]{204A87}{\ESHUnicodeSubstitution{\ESHMathSymbol{\Re}}})\ESHSpace{}\ESHUnicodeSubstitution{\ESHMathSymbol{\rightarrow}}\ESHSpace{}\textcolor[HTML]{204A87}{\ESHUnicodeSubstitution{\ESHMathSymbol{\Re}}}} in \dlang{}.
To construct primitives of this type, we can equivalently construct maps $\begin{contc}(\Gamma \stimes \IR \sto \IR) \to (\Gamma \sto \IR)\end{contc}$ for all $\cont{\Gamma \in \cat{AD}}$ in \clang{}.\footnote{
In any category of presheaves $[\mathcal{C}^\text{op}, \cat{Set}]$, letting $\y{\cdot}$ denote the Yoneda embedding and letting $\Rightarrow$ denote the internal hom,
then there is an equivalence between constants with the second-order type $(\y{A} \Rightarrow \y{B}) \Rightarrow \y{C}$ and the end\thesisonly{(see \Cref{preliminaries:presheaves})} $\int_\Gamma (\Gamma \times A \to_\mathcal{C} B) \to (\Gamma \to_\mathcal{C} C)$:
\[
\small
1 \to_{[\mathcal{C}^\text{op}, \cat{Set}]} (\y{A} \Rightarrow \y{B}) \Rightarrow \y{C}
\cong (\y{A} \Rightarrow \y{B}) \to_{[\mathcal{C}^\text{op}, \cat{Set}]}  \y{C}
= \int_\Gamma (\y{A} \Rightarrow \y{B})(\Gamma) \to  \y{C}(\Gamma)
\cong \int_\Gamma (\Gamma \times A \to_\mathcal{C} B) \to (\Gamma \to_\mathcal{C} C)
\]
}
Such a map takes as input $\cont{\IR}$-valued expression in a context $\cont{\Gamma \stimes \IR}$ and produce an $\cont{\IR}$-valued expression in the context $\cont{\Gamma}$ (for any $\cont{\Gamma}$).

Accordingly, we defined these second-order primitives with parametrically polymorphic mappings of derivative towers.
We must confirm that these definitions preserve successive consistency, i.e., they must map successively consistent derivative towers to successively consistent derivative towers.
In general, this boils down to confirming that taking the derivative of the value-level definitions of each of these primitives (when applied to any possible function $f : \Gamma \stimes \IR \sto \IR$) yields the definitions for the derivatives of these primitives.
It is possible to confirm for each definition that this is the case.

\subsubsection{Smooth integral}
The integral \ESHInline{\ESHBol{}integral01} is defined as follows for any $\begin{contc}f : \Gamma \stimes \R \sto \R\end{contc}$:
\[
\begin{qsmooth}\ESHInline{\ESHBol{}integral01}\end{qsmooth}(f)^{(k)}(\gamma; d\gamma_1, \ldots, d\gamma_k)
\triangleq
\int_0^1 f^{(k)}(\gamma, x; (d\gamma_1, 0), \ldots, (d\gamma_k, 0))\ dx
.
\]
Since integration is a linear operator, we essentially just integrate the first-order infinitesimal perturbations arising from $f$ at every order of derivative.
Integration is \emph{smooth} in the sense that if its input is smooth, its output will be smooth as well.
Note the similarity between the above AD tower and the result of postcomposing a linear function $\ell$ after a function $f$ arising from Faà di Bruno's formula described previously.
The reader may wonder how a semantics invoking integration might be computable; we discuss this in \Cref{smooth:implementation}.

\subsubsection{Smoothish Root Finding}
\label{root-finding}
The primitive \ESHInline{\ESHBol{}cutRoot\ESHSpace{}:\ESHSpace{}(\textcolor[HTML]{204A87}{\ESHUnicodeSubstitution{\ESHMathSymbol{\Re}}}\ESHSpace{}\ESHUnicodeSubstitution{\ESHMathSymbol{\rightarrow}}\ESHSpace{}\textcolor[HTML]{204A87}{\ESHUnicodeSubstitution{\ESHMathSymbol{\Re}}})\ESHSpace{}\ESHUnicodeSubstitution{\ESHMathSymbol{\rightarrow}}\ESHSpace{}\textcolor[HTML]{204A87}{\ESHUnicodeSubstitution{\ESHMathSymbol{\Re}}}} smoothly finds the root of any function with a single isolated root that is positive to its left and negative to its right.

Equivalently, \ESHInline{\ESHBol{}cutRoot} is a map $\cont{(\Gamma \stimes \R \sto \R) \to (\Gamma \sto \R)}$.
We will define \ESHInline{\ESHBol{}cutRoot} by using the stream characterization of smooth maps,
defining it with a function for its evaluation map and a smooth map for its derivative:
\begin{align*}
\begin{contc}
\begin{qsmooth}\ESHInline{\ESHBol{}cutRoot}\end{qsmooth} (f)^{(0)} \end{contc}
&
\triangleq
\begin{contc}
\lambda \gamma.\
[\sup \{ x : \R \mid f^{(0)}(\gamma, x) > 0 \}
,\inf \{x : \R \mid f^{(0)}(\gamma, x) < 0 \}]
\end{contc}
\\
\begin{contc}
\begin{qsmooth}\ESHInline{\ESHBol{}cutRoot}\end{qsmooth}(f)'
\end{contc} &
\triangleq
\begin{contc}
\begin{qvect}
\gamma, d\gamma \vdash
\text{let } y \triangleq \begin{uqcont}\begin{qsmooth}\ESHInline{\ESHBol{}cutRoot}\end{qsmooth}(f)\end{uqcont}(\gamma) \text{ in }
-\frac{\begin{uqcont}f'\end{uqcont}((\gamma, y),(d\gamma, 0))}{\begin{uqcont}f'\end{uqcont}((\gamma, y), (0, 1))}
\end{qvect}
\end{contc}
\end{align*}
The formula for the derivative is a simple application of the \emph{implicit function theorem}.\fTBD{Mike will want a citation?}
Note that we have a well-founded recursive reference following the same pattern as with multiplication.

\ESHInline{\ESHBol{}cutRoot} enables root-finding only for functions that have only one root.
In graphics, for ray tracing of implicit surfaces, it is useful to be able to find for a function \ESHInline{\ESHBol{}f\ESHSpace{}:\ESHSpace{}\textcolor[HTML]{204A87}{\ESHUnicodeSubstitution{\ESHMathSymbol{\Re}}}\ESHSpace{}\ESHUnicodeSubstitution{\ESHMathSymbol{\rightarrow}}\ESHSpace{}\textcolor[HTML]{204A87}{\ESHUnicodeSubstitution{\ESHMathSymbol{\Re}}}} the least root $x \in [0, 1]$ such that $f$ switches from positive for values just less than $x$ to negative for values just greater than $x$.
\ESHInline{\ESHBol{}firstRoot\ESHSpace{}:\ESHSpace{}(\textcolor[HTML]{204A87}{\ESHUnicodeSubstitution{\ESHMathSymbol{\Re}}}\ESHSpace{}\ESHUnicodeSubstitution{\ESHMathSymbol{\rightarrow}}\ESHSpace{}\textcolor[HTML]{204A87}{\ESHUnicodeSubstitution{\ESHMathSymbol{\Re}}})\ESHSpace{}\ESHUnicodeSubstitution{\ESHMathSymbol{\rightarrow}}\ESHSpace{}\textcolor[HTML]{204A87}{\ESHUnicodeSubstitution{\ESHMathSymbol{\Re}}}} accomplishes this:
\begin{align*}
\begin{contc}
\begin{qsmooth}\ESHInline{\ESHBol{}firstRoot}\end{qsmooth} (f)^{(0)} \end{contc}
&\triangleq
\begin{contc}
\lambda \gamma.\
\begin{matrix}
[\sup \{ x \in [0, 1] \mid \forall q \in [0, x].\ f^{(0)}(\gamma, q) > 0 \}
\\
,\inf \{x \in [0, 1] \mid \exists q \in [0, x].\ f^{(0)}(\gamma, x) < 0 \}]
\end{matrix}
\end{contc}
\\
\begin{contc}
\begin{qsmooth}\ESHInline{\ESHBol{}firstRoot}\end{qsmooth}(f)'
\end{contc}
&\triangleq
\begin{contc}
\begin{qvect}
\gamma, d\gamma \vdash
\begin{matrix}
\text{let } y \triangleq \begin{uqcont}\begin{qsmooth}\ESHInline{\ESHBol{}firstRoot}\end{qsmooth}(f)\end{uqcont}(\gamma) \text{ in }
\\
-\frac{\begin{uqcont}f'\end{uqcont}((\gamma, y),(d\gamma, 0))}{\begin{uqcont}f'\end{uqcont}((\gamma, y),(0, 1))}
\end{matrix}
\end{qvect}
\end{contc}
\end{align*}
Like with \ESHInline{\ESHBol{}cutRoot}, its derivatives are determined by the implicit function theorem; the only difference is in the definition of the value of the root.

\vspace{-.25em}
\subsubsection{Smoothish Optimization}
\label{smoothopt}

\dlang{} admits primitives \ESHInline{\ESHBol{}argmax01,\ESHSpace{}max01\ESHSpace{}:\ESHSpace{}(\textcolor[HTML]{204A87}{\ESHUnicodeSubstitution{\ESHMathSymbol{\Re}}}\ESHSpace{}\ESHUnicodeSubstitution{\ESHMathSymbol{\rightarrow}}\ESHSpace{}\textcolor[HTML]{204A87}{\ESHUnicodeSubstitution{\ESHMathSymbol{\Re}}})\ESHSpace{}\ESHUnicodeSubstitution{\ESHMathSymbol{\rightarrow}}\ESHSpace{}\textcolor[HTML]{204A87}{\ESHUnicodeSubstitution{\ESHMathSymbol{\Re}}}} that find the maximizing argument and the maximum, respectively, of a function \ESHInline{\ESHBol{}f\ESHSpace{}:\ESHSpace{}\textcolor[HTML]{204A87}{\ESHUnicodeSubstitution{\ESHMathSymbol{\Re}}}\ESHSpace{}\ESHUnicodeSubstitution{\ESHMathSymbol{\rightarrow}}\ESHSpace{}\textcolor[HTML]{204A87}{\ESHUnicodeSubstitution{\ESHMathSymbol{\Re}}}} over the unit interval.
Equivalently, each of \ESHInline{\ESHBol{}argmax01} and \ESHInline{\ESHBol{}max01} are maps $\begin{contc}(\Gamma \stimes \R \sto \R) \to (\Gamma \sto \R)\end{contc}$.

We first describe \ESHInline{\ESHBol{}argmax01}, which is defined as follows:
\begin{align*}
&\begin{contc}
\begin{qsmooth}\ESHInline{\ESHBol{}argmax01}\end{qsmooth} (f)^{(0)} \end{contc}
\triangleq
\begin{contc}
\lambda \gamma.\
\text{hull}\left( \{ x \in [0, 1] \mid f(\gamma, x) = \max_{z \in [0, 1]} f(\gamma, z) \} \right)
\end{contc}
\\&
\begin{contc}
\begin{qsmooth}\ESHInline{\ESHBol{}argmax01}\end{qsmooth}(f)'
\end{contc}
\triangleq
\begin{contc}
\begin{qvect}
\gamma, d\gamma \vdash
\begin{matrix}
\text{let } y \triangleq \begin{uqcont}\begin{qsmooth}\ESHInline{\ESHBol{}argmax01}\end{qsmooth}(f)\end{uqcont}(\gamma) \text{ in }
\\
\text{let } f'_y \triangleq \begin{uqcont}f'\end{uqcont}((\gamma, y), (0, 1)) \text{ in }
\\
\begin{cases}
-\frac{\begin{uqcont}f''\end{uqcont}(((\gamma, y),(0, 1)),((d\gamma, 0), (0, 0)))}{\begin{uqcont}f''\end{uqcont}(((\gamma, y),(0, 1)),((0, 1),(0,0)))}
& 0 < y < 1
\\
0 & y = 0 \wedge f'_y < 0
\\
0 & y = 1 \wedge f'_y > 0
\\
\bot & \text{otherwise}
\end{cases}
\end{matrix}
\end{qvect}
\end{contc}
\end{align*}
The input is a smooth map $f : \Gamma \times \R \sto \R$.
In general, for a $\gamma \in \Gamma$, there may be many values of $x$ achieving the same maximum $f(\gamma, x)$, so the value-level definition takes the convex hull of the set of those maximizing arguments.
The derivative of \ESHInline{\ESHBol{}argmax01} is not $\bot$ only when its value is maximal, i.e., there is only one maximizing argument, which we will call $y$.
There are three possibilities for $y$: either $0 < y < 1$, or $y = 0$, or $y = 1$.
In the case that $0 < y < 1$, then if $f''$ is defined at $(\gamma, y)$, then we know that $f'_y(\gamma, y) = 0$ and that this argmax is an isolated root of $f'_y$, where $f'_y$ is the derivative of $f$ with respect to its latter argument.
Any infinitesimal perturbation $d\gamma$ to $\gamma$ results in an infinitesimal perturbation to the root of $f'_y$, so the implicit function theorem defines how the root changes.
If the maximizing argument $y$ is on the boundary, i.e., $y = 0$ or $y = 1$, then if we additionally know that either $f'_y(\gamma, y) < 0$ or $f'_y(\gamma, y) > 0$, respectively, then it must be the case that the derivative of the argmax is 0, because the argmax will be stuck at the boundary no matter how $\gamma$ might be infinitesimally perturbed.
\TBD{Haven't defined cases within \cat{AD}}
\TBD{Discuss computability?}
%If $f :\Gamma \times \R \sto \R$ is smooth,
%then $\ESHInline{\ESHBol{}max01}(f) : \Gamma \sto \R$ is semismooth \fTBD{Not sure this is true. Add argument or citation.}.
%Defining the derivative is trickier.

We can now proceed to describe \ESHInline{\ESHBol{}max01}, whose derivative is defined in terms of \ESHInline{\ESHBol{}argmax01}:
\begin{align*}
\begin{contc}
\begin{qsmooth}\ESHInline{\ESHBol{}max01}\end{qsmooth} (f)^{(0)} \end{contc}
&
\triangleq
\begin{contc}
\lambda \gamma.\ \max_{x \in [0, 1]} f(\gamma, x)
\end{contc}
\\
\begin{contc}
\begin{qsmooth}\ESHInline{\ESHBol{}max01}\end{qsmooth}(f)'
\end{contc} &
\triangleq
\begin{contc}
\left(f \circ \begin{qsmooth}\ESHInline{\ESHBol{}argmax01}\end{qsmooth}(f) \right)'
\end{contc}
\end{align*}
Just as the derivative of \ESHInline{\ESHBol{}max} depends on which argument results in the max, similarly the derivative of \ESHInline{\ESHBol{}max01} is a function of the maximizing argument.
If we can isolate a single argmax, then \ESHInline{\ESHBol{}max01\ESHSpace{}f\ESHSpace{}=\ESHSpace{}f\ESHSpace{}(argmax01\ESHSpace{}f)},
and thus all the derivatives of \ESHInline{\ESHBol{}max01\ESHSpace{}f} follow from the chain rule and the smooth derivatives of \ESHInline{\ESHBol{}f} and \ESHInline{\ESHBol{}argmax01\ESHSpace{}f}.
\fTBD{What about a non-analytic smooth function, whose derivative may not have isolated roots? e.g., imagine minimizing the function here: \url{https://en.wikipedia.org/wiki/Non-analytic_smooth_function}}

%\vspace{-.35em}
\subsection{Internal Derivatives of Functions at All Types}
\label{smooth:internal-deriv}

The primitive \ESHInline{\ESHBol{}tangent\ESHSpace{}A\ESHSpace{}B\ESHSpace{}:\ESHSpace{}(A\ESHSpace{}\ESHUnicodeSubstitution{\ESHMathSymbol{\rightarrow}}\ESHSpace{}B)\ESHSpace{}\ESHUnicodeSubstitution{\ESHMathSymbol{\rightarrow}}\ESHSpace{}Tan\ESHSpace{}A\ESHSpace{}\ESHUnicodeSubstitution{\ESHMathSymbol{\rightarrow}}\ESHSpace{}Tan\ESHSpace{}B} permits the expression of the derivative of any function in \dlang{}, with any input type \ESHInline{\ESHBol{}A} and output type \ESHInline{\ESHBol{}B}.
These types are much more general than those on which differentiation in classically defined in mathematics.
In this section, we will explain the semantics of \ESHInline{\ESHBol{}tangent} and \ESHInline{\ESHBol{}Tan}, which generalize the notion of differentiation from \cat{AD} to apply to all objects in \cat{HAD} (i.e., all types in \dlang{}).

We need to systematically generalize the derivative of \cat{AD}, as expressed with the postfix $'$ operator, to apply to \cat{HAD}.
Following \citet{vakar}, we can apply the categorical technique of left Kan extensions, which extend a functor on a base category to one that acts on presheaves over that category.
Our definition of generalized tangent spaces and its properties will also be similar to the dvs diffeology on internal tangent bundles as described by \citet{difftangent}.
Accordingly, we can lift the operation of forward-mode differentiation from the first-order language \cat{AD} to the higher-order language \cat{HAD}.
Defining
\begin{align*}
\begin{contc}\mathsf{valueWithDer}\end{contc}
&\begin{contc}: (A \sto B) \to (A \stimes A \sto B \stimes B)\end{contc}
\\
\begin{contc}\mathsf{valueWithDer}(f)\end{contc}
&\triangleq
\begin{qvect} x, dx \vdash (\begin{uqcont}f\end{uqcont}(x), \begin{uqcont}f'\end{uqcont}(x, dx))\end{qvect},
\end{align*}
we find that $\cont{\mathsf{valueWithDer}}$ defines a functor on \cat{AD} acting on objects by $\cont{X \mapsto X \stimes X}$ from a space $X$ to its tangent bundle $X \stimes X$, where the tangent bundle $X \stimes X$ represents a point of $X$ together with an infinitesimal perturbation of that point.
The functoriality of $\mathsf{valueWithDer}$ follows from the chain rule of differentiation (and that $\begin{uqcont}\mathrm{id}'\end{uqcont}(x, dx) = dx$).

This functor can be extended to \cat{HAD} via a left Kan extension to produce a functor \ESHInline{\ESHBol{}Tan} and its functorial map \ESHInline{\ESHBol{}tangent\ESHSpace{}A\ESHSpace{}B\ESHSpace{}:\ESHSpace{}(A\ESHSpace{}\ESHUnicodeSubstitution{\ESHMathSymbol{\rightarrow}}\ESHSpace{}B)\ESHSpace{}\ESHUnicodeSubstitution{\ESHMathSymbol{\rightarrow}}\ESHSpace{}Tan\ESHSpace{}A\ESHSpace{}\ESHUnicodeSubstitution{\ESHMathSymbol{\rightarrow}}\ESHSpace{}Tan\ESHSpace{}B}, which runs generalized forward-mode derivatives, interpreted geometrically as a pushforward of the tangent bundles.
Concretely, we define the \emph{tangent bundle} functor \ESHInline{\ESHBol{}Tan}, as a left Kan extension, corresponds to a \emph{coend}:
\[
\ESHInline{\ESHBol{}Tan}(F)(\Gamma) \cong \coend{\Delta}{(\Gamma \sto \Delta^2) \times F(\Delta)}.
\]
Informally, the tangent bundle over the \dlang{} type $F$ in a context $\Gamma$ is represented by a pair of a value and infinitesimal perturbation $\Gamma \sto \Delta^2$ for some Cartesian space $\Delta$ (i.e., $\Delta = \R^n$ for some $n \in \nat$), together with a map from the space $\Delta$ into the type $F$.
Thus, if we wish to define an infinitesimal perturbation into a complicated type $F$,
we are able to do it by choosing a Cartesian space $\Delta$ to express that infinitesimal perturbation, and then we construct a map from $\Delta$ to $F$.
All elements of the tangent bundle of $F$ arise in that way.

We now explain how these tangent bundles work with an example.
Suppose $F = \R^2$ and we want to represent the tangent bundle $((0, 1), (1, 0)) \in \R^2 \times \R^2$, i.e., the vector $(0, 1)$ moving infinitesimally in the $(1, 0)$ direction.
Since there are no variables in the context, we can define the tangent bundle at once for all $\Gamma$.
The type of generalized tangent bundles is
\[
\ESHInline{\ESHBol{}Tan}(\y{\R^2})(\Gamma) \cong \coend{\Delta}{(\ast \sto \Delta^2) \times (\Delta \sto \R^2)}.
\]
We can represent the tangent bundle $((0, 1), (1, 0)) \in \R^2 \times \R^2$ in two equivalent ways.
The straightforward way is to take $\Delta = \R^2$ and put the point and its perturbation in the first component and the identity map in the second,
\[
(\R^2, (\langle (0, 1), (1, 0) \rangle, \mathrm{id})).
\]
Alternatively, we can represent it with a parametric function $f : \R \to \R^2$ defined by $f(t) = (0, 1) + t \cdot (1, 0)$, describing a point that moves from $(0, 1)$ at $t = 0$ in the direction of $(1, 0)$ as $t$ increases:
\[
(\R, (\langle 0, 1 \rangle, \lambda t.\ (0, 1) + t \cdot (1, 0))).
\]
Two members $(\tau_1, (f_1, g_1))$ and $(\tau_2, (f_2, g_2))$ of $\ESHInline{\ESHBol{}Tan}(\y{\R^2})(\Gamma)$ are equivalent if
\[
\mathsf{valueWithDer}(g_1) \circ f_1 = \mathsf{valueWithDer}(g_2) \circ f_2.
\]
Indeed, this is the case for the two examples above, as both compositions yield $\langle (0, 1), (1, 0) \rangle$.
This criterion for equivalence is for representable types such as $\y{\R^2}$ but generalizes for tangent bundles over types that are not representable.
It intuitively captures the notion that the first component of the tuple represents a tangent bundle of a representable space, whereas the second is a map that applies to that output but is yet to be differentiated.
This is the justification for applying $\mathsf{valueWithDer}$ above.

The Kan extension is genuinely an extension of the underlying functor $\mathsf{valueWithDer}$.
That is, we have the equivalence $\ESHInline{\ESHBol{}Tan}(\y{A}) \cong \y{A^2}$, where $\y{\cdot}$ is the Yoneda embedding (\Cref{tan-yoneda}).
This means that the generalized tangent bundle for Cartesian spaces $\R^n$ is indeed $\R^n \times \R^n$: one $\R^n$ for the point and one $\R^n$ for the infinitesimal perturbation.

The generalized tangent bundle functor supports other operations as well.
A polymorphic function \ESHInline{\ESHBol{}tangentValue\ESHSpace{}A\ESHSpace{}:\ESHSpace{}Tan\ESHSpace{}A\ESHSpace{}\ESHUnicodeSubstitution{\ESHMathSymbol{\rightarrow}}\ESHSpace{}A} projects out the base point.
The primitive \ESHInline{\ESHBol{}tangentProd\ESHSpace{}A\ESHSpace{}B\ESHSpace{}:\ESHSpace{}Tan\ESHSpace{}(A\ESHSpace{}*\ESHSpace{}B)\ESHSpace{}\ESHUnicodeSubstitution{\ESHMathSymbol{\cong}}\ESHSpace{}Tan\ESHSpace{}A\ESHSpace{}*\ESHSpace{}Tan\ESHSpace{}B} implements the following isomorphism:
\begin{proposition}
\label{tan-product}
Tangent bundles commute with products, i.e., $\ESHInline{\ESHBol{}Tan}(F \times G) \cong \ESHInline{\ESHBol{}Tan}(F) \times \ESHInline{\ESHBol{}Tan}(G)$.
\end{proposition}
\begin{proof}[Proof sketch]
First, we construct mappings in both directions:
We easily have the product projections
$\ESHInline{\ESHBol{}Tan}(F \times G) \to \ESHInline{\ESHBol{}Tan}(F)$
and
$\ESHInline{\ESHBol{}Tan}(F \times G) \to \ESHInline{\ESHBol{}Tan}(G)$.
Conversely, given $(\ESHInline{\ESHBol{}Tan}(F) \times \ESHInline{\ESHBol{}Tan}(G))(\Gamma)$,
we get $\Delta_1$ and $\Delta_2$ with
$f : \Gamma \sto \Delta_1^2$
and
$g : \Gamma \sto \Delta_2^2$
and $F(\Delta_1)$ and $G(\Delta_2)$.
Taking $\Delta \triangleq \Delta_1 \times \Delta_2$, we can define
$
h(\gamma) \triangleq ((x, y), (dx, dy))
$
where $(x, dx) = f(\gamma)$
and $(y, dy) = g(\gamma)$.
Using the pullbacks $\pi_1^* : F(\Delta_1) \to F(\Delta_1 \times \Delta_2)$
and $\pi_2^* : G(\Delta_2) \to G(\Delta_1 \times \Delta_2)$,
we can produce $\ESHInline{\ESHBol{}Tan}(F \times G)(\Gamma)$.

Next, it is possible to confirm that these mappings are mutually inverse,
using the fact that
\[\mathsf{valueWithDer}(\lambda (x, y). (f(x), g(y)))((x, y), (dx, dy)) = ((f(x), g(x)), (f'(x, dx),g'(y, dy))),\] together with general properties of limits and functoriality of \ESHInline{\ESHBol{}Tan}, $F$, and $G$.
\end{proof}

The primitive \ESHInline{\ESHBol{}tangent\_R\ESHSpace{}:\ESHSpace{}Tan\ESHSpace{}\textcolor[HTML]{204A87}{\ESHUnicodeSubstitution{\ESHMathSymbol{\Re}}}\ESHSpace{}\ESHUnicodeSubstitution{\ESHMathSymbol{\cong}}\ESHSpace{}\textcolor[HTML]{204A87}{\ESHUnicodeSubstitution{\ESHMathSymbol{\Re}}}\ESHSpace{}*\ESHSpace{}\textcolor[HTML]{204A87}{\ESHUnicodeSubstitution{\ESHMathSymbol{\Re}}}} that implements the following isomorphism (for the special case of \ESHInline{\ESHBol{}\textcolor[HTML]{204A87}{\ESHUnicodeSubstitution{\ESHMathSymbol{\Re}}}}):
\begin{proposition}
\label{tan-yoneda}
We have the equivalence $\ESHInline{\ESHBol{}Tan}(\y{A}) \cong \y{A^2}$, where $\y{\cdot}$ is the Yoneda embedding.
\end{proposition}
\begin{proof}
\begin{align*}
\ESHInline{\ESHBol{}Tan}(\y{A})(\Gamma)
\cong  \coend{\Delta}{ (\Gamma \sto \Delta^2) \times (\Delta \sto A) }.
\end{align*}
Given $f : \y{A^2}(\Gamma) = \Gamma \sto A^2$, we can take $\Delta = A$ and use $(f, \text{id})$.
Given an element of $\ESHInline{\ESHBol{}Tan}(\y{A})(\Gamma)$, i.e., some $\Delta$ and $f : \Gamma \sto \Delta^2$ and $g : \Delta \sto A$, then $\mathsf{valueWithDer}(g) \circ f : \y{A^2}(\Gamma)$.
\end{proof}
Note that we have $\ESHInline{\ESHBol{}tangentValue} \circ \ESHInline{\ESHBol{}tangent\_R} = \ESHInline{\ESHBol{}fst}$, i.e., the first component is the base point and the second is the infinitesimal perturbation.

Note that the types to represent isomorphisms of tangent bundles are not necessarily isomorphisms in \dlang{}:
the type \ESHInline{\ESHBol{}\ESHUnicodeSubstitution{\ESHMathSymbol{\cong}}} just corresponds to pairs of maps back and forth.
The primitive \ESHInline{\ESHBol{}tangentTo\_R\ESHSpace{}A\ESHSpace{}:\ESHSpace{}Tan\ESHSpace{}(A\ESHSpace{}\ESHUnicodeSubstitution{\ESHMathSymbol{\rightarrow}}\ESHSpace{}\textcolor[HTML]{204A87}{\ESHUnicodeSubstitution{\ESHMathSymbol{\Re}}})\ESHSpace{}\ESHUnicodeSubstitution{\ESHMathSymbol{\cong}}\ESHSpace{}(A\ESHSpace{}\ESHUnicodeSubstitution{\ESHMathSymbol{\rightarrow}}\ESHSpace{}\textcolor[HTML]{204A87}{\ESHUnicodeSubstitution{\ESHMathSymbol{\Re}}})\ESHSpace{}*\ESHSpace{}(A\ESHSpace{}\ESHUnicodeSubstitution{\ESHMathSymbol{\rightarrow}}\ESHSpace{}\textcolor[HTML]{204A87}{\ESHUnicodeSubstitution{\ESHMathSymbol{\Re}}})}, in which tangent bundles distribute over functions into \ESHInline{\ESHBol{}\textcolor[HTML]{204A87}{\ESHUnicodeSubstitution{\ESHMathSymbol{\Re}}}}, implements mappings that are an isomorphism only when we restrict \ESHInline{\ESHBol{}\textcolor[HTML]{204A87}{\ESHUnicodeSubstitution{\ESHMathSymbol{\Re}}}} to $\R$ (rather than all of $\IR$):

\begin{proposition}
\label{tan-arrow-real}
There is an isomorphism
$\ESHInline{\ESHBol{}Tan}(A \Rightarrow \y{\R}) \cong A \Rightarrow \ESHInline{\ESHBol{}Tan}(\y{\R}) \cong A \Rightarrow \y{\R^2}$.
\end{proposition}
\begin{proof}
First, we construct the mappings in each direction. Note that these types are:
\begin{align*}
\ESHInline{\ESHBol{}Tan}(A \Rightarrow \y{\R})(\Gamma) &\cong \coend{\Delta}{(\Gamma \sto \Delta^2) \times \cend{X}{ (X \sto \Delta) \to A(X) \to (X \sto \R)}}
\\
(A \Rightarrow \y{\R^2})(\Gamma) &\cong \cend{X}{ (X \sto \Gamma) \to A(X) \to (X \sto \R^2)}
\end{align*}

Given $f : (A \Rightarrow \y{\R^2})$, we take $\Delta = \Gamma \times \R$, and use
\begin{align*}
\exists \Gamma \times \R.\ (\lambda \gamma.\ ((\gamma, 0), (0, 1)),
\Lambda X.\ \lambda (e : X \sto \Gamma \times \R).\ \lambda (a : A(X)).\
\\
\text{let }(g, dg) = f(X,\pi_1 \circ e, a)\text{ in } \lambda x : X.\ g(x) + \pi_2(e(x)) \cdot dg(x))
\end{align*}

Conversely, given a member of
$\ESHInline{\ESHBol{}Tan}(A \Rightarrow \y{\R})(\Gamma)$,
i.e.,
a $\Delta$ with $d : \Gamma \sto \Delta^2$ and $f : \int_X (X \sto \Delta) \to A(X) \to (X \sto \R)$,
we can provide\fTBD{I believe this is incorrect.}
\[
\Lambda X.\ \lambda (e : X \sto \Gamma).\ \lambda (a : A(X)).\
\mathsf{valueWithDer}(f(X,\pi_1
\langle f(X,\pi_1 \circ d \circ e,a), f(X,\pi_2 \circ d \circ e,a) \rangle)).
\]

Next, we must confirm that these mappings are mutually inverse.
This boils down to the basic identity
$
f'(x; v) = \frac{\partial (f(x + t \cdot v))}{\partial t} \mid_{t=0}.
$
\end{proof}
Note that it is not an isomorphism for all of \ESHInline{\ESHBol{}\textcolor[HTML]{204A87}{\ESHUnicodeSubstitution{\ESHMathSymbol{\Re}}}}, because we rely on the algebraic law $x + 0 \cdot y = x$ for all $y$, but if we allow $y \in \IR \setminus \R$, there is the counterexample $x + 0 \cdot \bot = \bot$.

\subsection{Consistency}

\begin{proposition}[Consistency of differentiation in the higher-order language]
\label{ho-soundness}
Given any term $\Gamma \vdash e : \tau$ in \dlang{} where $\Gamma$ is a context of all ground types and $\tau$ is a ground type, then $\begin{qsmooth}e\end{qsmooth}$ is equivalent to some first-order smoothish map $f$, i.e., successively consistent derivative tower.
\end{proposition}
\begin{proof}
Since the Yoneda embedding is full and faithful, first-order terms in \cat{HAD} correspond to morphisms in \cat{AD}, so this statement reduces to \Cref{fo-soundness}.
\end{proof}

\section{Computability and Numerically-Sound Implementation}
\label{smooth:implementation}

It is not obvious that the categorical semantics of \dlang{} we present in \Crefrange{fo-semantics}{ho-semantics} is actually implementable (in a sound manner).
The semantics critically uses reals and real arithmetic, rather than some approximation like floating point (which would fail to give even the most basic equalities such as $1/5+2/5=3/5$).
And value-level definitions of higher-order primitives in \dlang{} are expressed in terms of mathematical operations for integration, optimization, and root finding applied to arbitrary continuous maps.
In fact, our semantic development is computable, and we have implemented it in a numerically sound manner as an embedded DSL in Haskell.

%This section explains how that works.

Our semantics can be developed constructively and interpreted within the \emph{internal language} of another topos, which we call \emph{\clang}, in order to provide a computable interpretation.
We base \clang{} on MarshallB \citep{icfp19}.
Our implementation of \dlang{} more-or-less directly follows interpreting the semantics of \dlang{} within \clang{} and in turn implementing \clang{} in Haskell.

\clang{} is a topos of sheaves over a Cartesian monoidal category that we call \cat{CTop}.
\cat{CTop} is a category of computably presented topological spaces and computable continuous maps.
\clang{} is the topos of sheaves over \cat{CTop} with the open cover topology (along the lines of \cite{continuoustruth}).

What results is a stack of languages: \dlang{} reducing to \cat{AD}, implemented in \clang{}, which reduces to \cat{CTop}, which carries the final executable content of ground terms.
We can view it like a stack of metaprogramming languages on top of \cat{CTop}:
ultimately, when a closed term of \dlang{} (or any other language in the stack) of ground type is evaluated and displayed as a sequence of improving approximations, it is in fact a closed term of \cat{CTop}, i.e., a computable point of a topological space.

\Comment{
Each higher layer in the stack provides additional expressiveness.
In \cat{CTop} we can represent a continuous map such as the cube root function, but we need the expressivity of \clang{} to represent a higher-order function for general root-finding, such as the \ESHInline{\ESHBol{}cutRoot} function described.
Ultimately, when \ESHInline{\ESHBol{}cutRoot} is applied to find the roots of a particular function $f$, it essentially reduces to a term in \cat{CTop}.
However, functions in \cat{CTop} and \clang{} do not admit differentiation; we can represent the cube root function in \cat{AD} as a smooth function, with implementations of all its derivatives.
However, we cannot \emph{internally} represent the differentiation function \ESHInline{\ESHBol{}deriv}, nor the more general function for root-finding \ESHInline{\ESHBol{}cutRoot}, since these are higher-order functions.
These can only be represented in \dlang{}.

\begin{figure}
\begin{tikzpicture}
\tikzset{block/.style= {draw, minimum height=2em,minimum width=4em}}
\node[block, fill=smoothcolor!20] (lambdas) {$\dlang{}$ implements $\cat{HAD} = \cat{Sh}_J(\cat{AD})$};
\node[block, below=1em of lambdas] (lambdasfo) {\cat{AD} (\cont{$\sto$})};
\node[block, below=1em of lambdasfo, fill=contcolor!20] (lambdac) {$\clang{}$ implements $\cat{Sh}_{J'}(\cat{CTop})$};
\node[block, below=1em of lambdac] (lambdacfo) {\cat{CTop}};
\draw[->] (lambdas.south) -- (lambdasfo.north);
\draw[->] (lambdasfo.south) -- (lambdac.north);
\draw[->] (lambdac.south) -- (lambdacfo.north);
\end{tikzpicture}
\caption{The stack of languages by which \dlang{} is defined.}
\label{system-diagram}
\vspace{-0.8em}
\end{figure}
}

\paragraph{Semantics of \clang{} and Implications for \dlang{}.}

\clang{} is a language whose types are (generalized) topological spaces with computable structure and whose functions are (generalized) computable continuous maps.
\clang{} permits all the higher-order functions and higher-order types that we will seek to define in \dlang{} and enables their computation to arbitrary precision.
This section describes \clang{} by example.
In \clang{}, the type \ESHInline{\ESHBol{}\textcolor[HTML]{204A87}{\ESHUnicodeSubstitution{\ESHMathSymbol{\Re}}}} in \clang{} represents the interval reals $\IR$.
One closed term, or value, of type \ESHInline{\ESHBol{}\textcolor[HTML]{204A87}{\ESHUnicodeSubstitution{\ESHMathSymbol{\Re}}}} is \ESHInline{\ESHBol{}sqrt\ESHSpace{}2}.
A value of \ESHInline{\ESHBol{}\textcolor[HTML]{204A87}{\ESHUnicodeSubstitution{\ESHMathSymbol{\Re}}}} represents a point of the space $\cont{\IR}$
and is computationally represented by streams of increasingly precise approximations (i.e., monotone with respect to $\sqsubseteq$):
\vspace{-0.7em}
\begin{center}
\begin{minipage}[t]{0.7\textwidth}
\begin{prompt}
\ESHInline{\ESHBol{}{>}\ESHSpace{}sqrt\ESHSpace{}2\ESHSpace{}:\ESHSpace{}\textcolor[HTML]{204A87}{\ESHUnicodeSubstitution{\ESHMathSymbol{\Re}}}}
\end{prompt}
\begin{repl}
\small
\ESHInline{\ESHBol{}[1.4142135619,\ESHSpace{}1.4142135624]}
\\ \ESHInline{\ESHBol{}[1.414213562370,\ESHSpace{}1.414213562384]}
\\ \ESHInline{\ESHBol{}[1.4142135623729,\ESHSpace{}1.4142135623733]}
\\ ...
\end{repl}
\end{minipage}
\end{center}

Note that these streams of increasingly precise approximations can be used to provide the arbitrary-precision interface where one asks for a precision tolerance and gets a result.
Each interval $\kint{\underline{x}}{\overline{x}}$, where $\underline{x} \in \{ - \infty \} \cup \mathbb{D}, \overline{x} \in \mathbb{D} \cup \{ \infty \}$, has either infinite or dyadic-rational ($\mathbb{D} = \{ k / 2^n \mid k \in \mathbb{Z}, n \in \mathbb{N} \}$) endpoints
and represents partial information about \ESHInline{\ESHBol{}sqrt\ESHSpace{}2}:
the first component represents a rational lower bound (with $-\infty$ being a vacuous bound) and the second an upper bound (with $\infty$ vacuous).
\clang{} is \emph{sound} in the sense that these bounds are guaranteed to hold of the true value.
Two closed terms of \ESHInline{\ESHBol{}\textcolor[HTML]{204A87}{\ESHUnicodeSubstitution{\ESHMathSymbol{\Re}}}} in \clang{} are considered equivalent if their streams always overlap, even if the streams are not identical.
For instance, $\ESHInline{\ESHBol{}(sqrt\ESHSpace{}2)\ESHRaise{0.30}{2}} = \ESHInline{\ESHBol{}2}$:
\vspace{-0.7em}
\begin{center}
\begin{minipage}[t]{0.48\textwidth}
\begin{prompt}
\ESHInline{\ESHBol{}{>}\ESHSpace{}(sqrt\ESHSpace{}2)\ESHRaise{0.30}{2}}
\end{prompt}
\begin{repl}
\small
\ESHInline{\ESHBol{}[1.9999999986,\ESHSpace{}2.0000000009]}
\\ \ESHInline{\ESHBol{}[1.999999999985,\ESHSpace{}2.000000000058]}
\\ \ESHInline{\ESHBol{}[1.9999999999991,\ESHSpace{}2.0000000000009]}
\\ ...
\end{repl}
\end{minipage}
\begin{minipage}[t]{0.48\textwidth}
\begin{prompt}
\ESHInline{\ESHBol{}{>}\ESHSpace{}2}
\end{prompt}
\begin{repl}
\small
\ESHInline{\ESHBol{}[2.0000000000,\ESHSpace{}2.0000000000]}
\\ \ESHInline{\ESHBol{}[2.000000000000,\ESHSpace{}2.000000000000]}
\\ \ESHInline{\ESHBol{}[2.0000000000000,\ESHSpace{}2.0000000000000]}
\\ ...
\end{repl}
\end{minipage}
\end{center}
The equivalence means that one can substitute \ESHInline{\ESHBol{}(sqrt\ESHSpace{}2)\ESHRaise{0.30}{2}} for \ESHInline{\ESHBol{}2} within any program without affecting its meaning.
In contrast, the floating-point computation for many languages and CPUs returns 2.0000000000000004, which is not 2 and does
not itself indicate a larger range of possible values that includes 2,
and would not validate the equation $\texttt{(sqrt 2)\^{}2} = \texttt{2}$.

First-order functions in \clang{} are stream transformers of their approximations.
For instance, applying the squaring function \ESHInline{\ESHBol{}(\ESHDash{})\ESHRaise{0.30}{2}\ESHSpace{}:\ESHSpace{}\textcolor[HTML]{204A87}{\ESHUnicodeSubstitution{\ESHMathSymbol{\Re}}}\ESHSpace{}\ESHUnicodeSubstitution{\ESHMathSymbol{\rightarrow}}\ESHSpace{}\textcolor[HTML]{204A87}{\ESHUnicodeSubstitution{\ESHMathSymbol{\Re}}}} to \ESHInline{\ESHBol{}sqrt\ESHSpace{}2} yields the following result:
\vspace{-0.7em}
\begin{center}
\begin{minipage}[t]{0.48\textwidth}
\begin{prompt}
\ESHInline{\ESHBol{}{>}\ESHSpace{}sqrt\ESHSpace{}2\ESHSpace{}:\ESHSpace{}\textcolor[HTML]{204A87}{\ESHUnicodeSubstitution{\ESHMathSymbol{\Re}}}}
\end{prompt}
\begin{repl}
\small
\ESHInline{\ESHBol{}[1.4142135619,\ESHSpace{}1.4142135624]}
\\ \ESHInline{\ESHBol{}[1.414213562370,\ESHSpace{}1.414213562384]}
\\ \ESHInline{\ESHBol{}[1.4142135623729,\ESHSpace{}1.4142135623733]}
\\ ...
\end{repl}
\end{minipage}
\begin{minipage}[t]{0.48\textwidth}
\begin{prompt}
\ESHInline{\ESHBol{}{>}\ESHSpace{}(sqrt\ESHSpace{}2)\ESHRaise{0.30}{2}\ESHSpace{}:\ESHSpace{}\textcolor[HTML]{204A87}{\ESHUnicodeSubstitution{\ESHMathSymbol{\Re}}}}
\end{prompt}
\begin{repl}
\small
\ESHInline{\ESHBol{}[1.9999999986,\ESHSpace{}2.0000000009]}
\\ \ESHInline{\ESHBol{}[1.999999999985,\ESHSpace{}2.000000000058]}
\\ \ESHInline{\ESHBol{}[1.9999999999991,\ESHSpace{}2.0000000000009]}
\\ ...
\end{repl}
\end{minipage}
\end{center}
\noindent
In this case, the squaring function squares each input interval to produce output intervals.
The computation is continuous in the sense that the computation of each interval result of \ESHInline{\ESHBol{}(sqrt\ESHSpace{}2)\ESHRaise{0.30}{2}} needs only an interval approximation of \ESHInline{\ESHBol{}sqrt\ESHSpace{}2}.
First-order functions such as \ESHInline{\ESHBol{}(\ESHDash{})\ESHRaise{0.30}{2}} are \emph{continuous maps},
meaning that in order to approximate the output to any finite level of precision,
it suffices to inspect the input to only a finite level of precision.

\paragraph{Implementing Higher-Order Primitives.}

The value-level definitions of higher-order primitives in \dlang{} are expressed in terms of mathematical operations for integration, optimization, and root finding.
It's not obvious that these are computable.
However, MarshallB \citep{icfp19} demonstrates how to endow a language with computable
implementations of Riemannian integration, maximization over compact sets, as well as a \emph{Dedekind cut} primitive that is essentially equivalent to the root finding of \ESHInline{\ESHBol{}cutRoot} and can be used to implement the root finding of \ESHInline{\ESHBol{}firstRoot}.
We were able to implement these MarshallB primitives in \clang{} and use them to implement the higher-order primitives in \dlang{}.

\vspace{-.15cm}
\paragraph{Haskell Implementation.}

We implemented \dlang{} as an embedded language within Haskell.
Because $\R^n$ and $\IR^n$ are representable within \cat{CTop},
we actually implement \cat{AD} directly using \cat{CTop} within Haskell,
rather than working internally to \clang{}.
We implement \cat{CTop} using an interval-arithmetic library that in turn uses MPFR \cite{mpfr}, a library for multi-precision floating-point arithmetic.
We include this implementation and all the code examples as supplementary material,
\paperonly{and will make it publicly available.}
\thesisonly{found at \url{github.com/bmsherman/phd-thesis-supplemental}.}
See the readme file for more information about the code.

\vspace{-.15cm}
\paragraph{Computability and Numerical Soundness.}

The semantics for \dlang{} supports a realistic machine model for computing real-valued results to arbitrary precision.
This is in contrast to semantics that permit Boolean-valued comparison of real numbers, and computational models like Real RAM, in which a machine can compare real numbers in constant time.
When algorithms are designed based on such models but implemented with floating-point arithmetic,
those implementations may fail to be robust to floating-point error (e.g., \citep{examplesrobustness}).
In contrast, the continuity inherent in \dlang{}'s semantics provides a robustness guarantee:
arbitrary-precision approximations of the output can be produced by inspecting only finite-precision approximations of the input.

\section{Higher-Order Datatypes and Libraries}
\label{smooth:applications}

This section demonstrates the unique expressivity and computability of \dlang{}.
We use the novel higher-order primitives available in \dlang{} -- including integration, optimization, and root-finding -- to build libraries for constructing and computing with three different higher-order datatypes: probability distributions (and measures), implicit surfaces, and generalized parametric surfaces.
Since these libraries are implemented in \dlang{}, they are differentiable (arbitrarily many times).
For each library, we compute an example differentiation task.
\Cref{fig:examples} shows a high-level overview of each example.
We now detail the implementation of each of the libraries and provide the implementations for each of the corresponding examples.
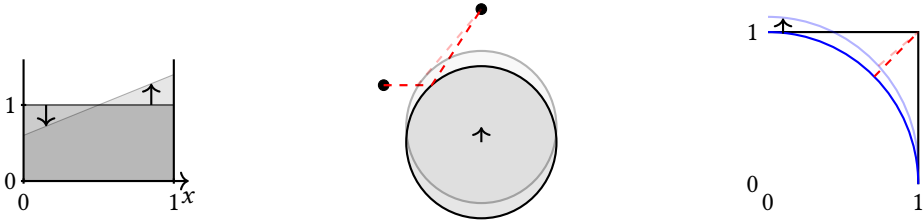
\TBD{Additional novelty/genericity claim about the libraries?}

\begin{figure}

\begin{subfigure}[t]{0.29\textwidth}
\begin{center}
  \begin{tikzpicture}
    \fill [black, opacity=0.2, domain=0:2, variable=\x]
      (0, 0)
      -- plot ({\x}, {1})
      -- (2, 0)
      -- cycle;

    \fill[black, opacity=0.1, domain=0:2, variable=\x]
      (0, 0)
      -- plot ({\x}, {1 + 0.4 * (\x - 1)})
      -- (2, 0)
      -- cycle;

    \draw [thick, ->] (-0,0)--(2.2,0) node[right, below] {$x$};
     \foreach \x in {0,...,1}
       \draw[xshift={2*\x cm}, thick] (0pt,1pt)--(0pt,-1pt) node[below] {$\x$};

    \draw [thick] (0,-0)--(0,1.6);
    \draw [thick] (2,-0)--(2,1.6);
     \foreach \y in {0,...,1}
       \draw[yshift=\y cm, thick] (-1pt,0pt)--(1pt,0pt) node[left] {$\y$};

    \draw [domain=0:2, variable=\x]
      plot ({\x}, {1});
    \draw[domain=0:2, variable=\x, opacity=0.3]
      plot ({\x}, {1 + 0.4 * (\x - 1)});

     \draw[thick, ->] (1.7, 1) -- (1.7, 1.28);
     \draw[thick, ->] (0.3, 1) -- (0.3, 0.72);
  \end{tikzpicture}
\end{center}
\caption{\emph{Probability distributions}: How does the mean and variance of the uniform distribution change as you weight its mass to tilt more towards higher values and away from lower values?}
\label{fig:probability}
\end{subfigure}
\hfill
\begin{subfigure}[t]{0.29\textwidth}
\begin{center}
\begin{tikzpicture}
\draw[fill=black] (1.3, 1) circle(0.2em);
\draw[fill=black] (0, 0) circle(0.2em);
\draw[thick, fill=black, fill opacity=0.1] (1.3, -3/4) circle (1);
\draw[dashed, thick, red] (0, 0) -- (0.63856, 0) -- (1.3, 1);
\draw[dashed, thick, red, opacity=0.3] (0, 0) -- (0.45, 0) -- (1.3, 1);
\draw[thick, fill=black, opacity=0.3, fill opacity=0.03] (1.3, -0.55) circle (1);
\draw[thick,->] (1.3, -3/4) -- (1.3, -0.55);
\end{tikzpicture}
\end{center}
\caption{\emph{Implicit surfaces}: A ray of light from a source above bounces off a circle before hitting a camera. How does the brightness change when the circle is moved up?}
\label{fig:ray-tracing}
\end{subfigure}
\hfill
\begin{subfigure}[t]{0.29\textwidth}
\begin{center}
\begin{tikzpicture}
\draw[thick, fill=none, fill opacity=0.4]  (1, -1) -- (1, 1) -- (-1, 1);
\draw (-1, -1) node[below] {0};
\draw (1, -1) node[below] {1};
\draw (-1, -1) node[left] {0};
\draw (-1, 1) node [left] {1};
\draw[thick, blue, fill=none, fill opacity=0.4] (1, -1) arc[start angle=0, end angle=90, radius=2];
\draw[thick, blue, opacity=0.3] (1, -0.8) arc[start angle=0, end angle=90, radius=2];
\draw[thick,->] (-0.8, 0.98) -- (-0.8, 1.18);
\draw[red, thick, dashed] (0.4142, 0.4142) -- (1,1);;
\draw[red, thick, dashed, opacity=0.3] (0.47, 0.555) -- (1,1);;
\end{tikzpicture}
\end{center}
\caption{\emph{Generalized parametric surfaces}: How does the Hausdorff distance between the quarter circle and the ``L'' shape change as the quarter circle is moved up?}
\label{fig:disk-square}
\end{subfigure}
\vspace{-.15cm}
\caption{Three example differentiation problems we will express and compute with libraries in \dlang{}.}
\label{fig:examples}
\vspace{-0.8em}
\end{figure}

\subsection{Probability Distributions (and Measures)}

Probability is central to many machine-learning applications.
Loss functions for Bayesian neural networks, GANs, \NA{etc.} involve expectations over probability distributions.\fTBD{Cite \cite{probability-functional-descent}?}
However, no previous work on the semantics of AD supports probability distributions\footnote{
While other works can represent expectations over distributions with finite support as sums, this would not work for distributions with infinite support. Loss functions frequently involve expectations over distributions with infinite support.
}.
The interaction between probabilistic choice and differentiation is nontrivial,
and the lack of a semantic treatment of their interaction has real consequences for machine-learning practitioners using AD libraries who seek to combine them.
Practitioners often use Monte Carlo sampling to approximate expectations,
but because derivatives cannot be propagated through the samplers in common frameworks such as PyTorch and TensorFlow,
code that \emph{looks} correct and produces appropriate approximations of its value-level output can end up producing incorrect derivatives when AD is applied (as mentioned in the introduction).
This common pitfall, which can be difficult to detect, necessitates \emph{the reparameterization trick}, where code is rewritten such that samplers do not depend on any parameters that are to be differentiated.
%The name ``trick'' points to the confusion that arises from the interaction of probability and differentiation.

\dlang{} can represent a monad of probability distributions $\smooth{\mathcal{P}}$, making it the first language semantics to support differentiation through probabilistic choice, including through distributions such as the uniform distribution on the unit interval.\fTBD{MC: buried claim. This needs to be in introduction.}
%\TBD{, eliminating the need for the reparameterization trick}.
Supporting probability distributions is hard because they must involve higher-order functions:
expectations are higher-order functions $\begin{smoothc}\mathcal{P}(A) \times (A \to \R) \to \R\end{smoothc}$, as is the monadic bind operator $\begin{smoothc}\mathcal{P}(A) \times (A \to \mathcal{P}(B)) \to \mathcal{P}(B)\end{smoothc}$ that supports compositional construction of complex probability distributions from simple ones.

\paragraph{A \dlang{} Library for Probability Distributions and Measures.}

Probability distributions, measures, and distributions (in the sense of generalized functions) can all be described as integrals,
\[
\ESHInline{\ESHBol{}\textcolor[HTML]{204A87}{type}\ESHSpace{}Integral\ESHSpace{}A\ESHSpace{}=\ESHSpace{}(A\ESHSpace{}\ESHUnicodeSubstitution{\ESHMathSymbol{\rightarrow}}\ESHSpace{}\textcolor[HTML]{204A87}{\ESHUnicodeSubstitution{\ESHMathSymbol{\Re}}})\ESHSpace{}\ESHUnicodeSubstitution{\ESHMathSymbol{\rightarrow}}\ESHSpace{}\textcolor[HTML]{204A87}{\ESHUnicodeSubstitution{\ESHMathSymbol{\Re}}}},
\]
detailed in \Cref{code:expecter}.
Integrals are functions $\smooth{i : (A \to \R) \to \R}$ which are linear in their arguments.
Measures are those integrals $\smooth{i}$ satisfying $\smooth{i(f) \ge 0}$ whenever $\smooth{f(x) \ge 0}$ for all $x \in \R$.
Probability distributions are those measures $\smooth{i}$ satisfying $\smooth{i(\lambda x.\ 1) = 1}$; the integral for a probability distribution computes the expectation of a real-valued function under that distribution.

\begin{figure}[htbp]
\small
\begin{ESHBlock}
\ESHBol{}\textcolor[HTML]{204A87}{type}\ESHSpace{}Integral\ESHSpace{}A\ESHSpace{}=\ESHSpace{}(A\ESHSpace{}\ESHUnicodeSubstitution{\ESHMathSymbol{\rightarrow}}\ESHSpace{}\textcolor[HTML]{204A87}{\ESHUnicodeSubstitution{\ESHMathSymbol{\Re}}})\ESHSpace{}\ESHUnicodeSubstitution{\ESHMathSymbol{\rightarrow}}\ESHSpace{}\textcolor[HTML]{204A87}{\ESHUnicodeSubstitution{\ESHMathSymbol{\Re}}}\ESHEol
\ESHBol{}\ESHEmptyLine{}\ESHEol
\ESHBol{}\textcolor[HTML]{346604}{let}\ESHSpace{}dirac\ESHSpace{}A\ESHSpace{}(x\ESHSpace{}:\ESHSpace{}A)\ESHSpace{}:\ESHSpace{}Integral\ESHSpace{}A\ESHSpace{}=\ESHSpace{}\textcolor[HTML]{346604}{\ESHUnicodeSubstitution{\ESHMathSymbol{\lambda}}}\ESHSpace{}f\ESHSpace{}:\ESHSpace{}A\ESHSpace{}\ESHUnicodeSubstitution{\ESHMathSymbol{\rightarrow}}\ESHSpace{}\textcolor[HTML]{204A87}{\ESHUnicodeSubstitution{\ESHMathSymbol{\Re}}}\ESHSpace{}\ESHUnicodeSubstitution{\ESHMathSymbol{\Rightarrow}}\ESHSpace{}f\ESHSpace{}x\ESHEol
\ESHBol{}\textcolor[HTML]{346604}{let}\ESHSpace{}bind\ESHSpace{}A\ESHSpace{}B\ESHSpace{}(x\ESHSpace{}:\ESHSpace{}Integral\ESHSpace{}A)\ESHSpace{}(f\ESHSpace{}:\ESHSpace{}A\ESHSpace{}\ESHUnicodeSubstitution{\ESHMathSymbol{\rightarrow}}\ESHSpace{}Integral\ESHSpace{}B)\ESHSpace{}:\ESHSpace{}Integral\ESHSpace{}B\ESHEol
\ESHBol{}\ESHSpace{}\ESHSpace{}=\ESHSpace{}\textcolor[HTML]{346604}{\ESHUnicodeSubstitution{\ESHMathSymbol{\lambda}}}\ESHSpace{}k\ESHSpace{}:\ESHSpace{}B\ESHSpace{}\ESHUnicodeSubstitution{\ESHMathSymbol{\rightarrow}}\ESHSpace{}\textcolor[HTML]{204A87}{\ESHUnicodeSubstitution{\ESHMathSymbol{\Re}}}\ESHSpace{}\ESHUnicodeSubstitution{\ESHMathSymbol{\Rightarrow}}\ESHSpace{}x\ESHSpace{}(\textcolor[HTML]{346604}{\ESHUnicodeSubstitution{\ESHMathSymbol{\lambda}}}\ESHSpace{}a\ESHSpace{}:\ESHSpace{}A\ESHSpace{}\ESHUnicodeSubstitution{\ESHMathSymbol{\Rightarrow}}\ESHSpace{}f\ESHSpace{}a\ESHSpace{}k)\ESHEol
\ESHBol{}\textcolor[HTML]{346604}{let}\ESHSpace{}zero\ESHSpace{}A\ESHSpace{}:\ESHSpace{}Integral\ESHSpace{}A\ESHSpace{}=\ESHSpace{}\textcolor[HTML]{346604}{\ESHUnicodeSubstitution{\ESHMathSymbol{\lambda}}}\ESHSpace{}f\ESHSpace{}:\ESHSpace{}A\ESHSpace{}\ESHUnicodeSubstitution{\ESHMathSymbol{\rightarrow}}\ESHSpace{}Real\ESHSpace{}\ESHUnicodeSubstitution{\ESHMathSymbol{\Rightarrow}}\ESHSpace{}0\ESHEol
\ESHBol{}\textcolor[HTML]{346604}{let}\ESHSpace{}add\ESHSpace{}A\ESHSpace{}(x\ESHSpace{}y\ESHSpace{}:\ESHSpace{}Integral\ESHSpace{}A)\ESHSpace{}:\ESHSpace{}Integral\ESHSpace{}A\ESHSpace{}=\ESHSpace{}\textcolor[HTML]{346604}{\ESHUnicodeSubstitution{\ESHMathSymbol{\lambda}}}\ESHSpace{}f\ESHSpace{}:\ESHSpace{}A\ESHSpace{}\ESHUnicodeSubstitution{\ESHMathSymbol{\rightarrow}}\ESHSpace{}\textcolor[HTML]{204A87}{\ESHUnicodeSubstitution{\ESHMathSymbol{\Re}}}\ESHSpace{}\ESHUnicodeSubstitution{\ESHMathSymbol{\Rightarrow}}\ESHSpace{}x\ESHSpace{}f\ESHSpace{}+\ESHSpace{}y\ESHSpace{}f\ESHEol
\ESHBol{}\textcolor[HTML]{346604}{let}\ESHSpace{}map\ESHSpace{}A\ESHSpace{}B\ESHSpace{}(f\ESHSpace{}:\ESHSpace{}A\ESHSpace{}\ESHUnicodeSubstitution{\ESHMathSymbol{\rightarrow}}\ESHSpace{}B)\ESHSpace{}(e\ESHSpace{}:\ESHSpace{}Integral\ESHSpace{}A)\ESHSpace{}:\ESHSpace{}Integral\ESHSpace{}B\ESHSpace{}=\ESHEol
\ESHBol{}\ESHSpace{}\ESHSpace{}\textcolor[HTML]{346604}{\ESHUnicodeSubstitution{\ESHMathSymbol{\lambda}}}\ESHSpace{}k\ESHSpace{}:\ESHSpace{}B\ESHSpace{}\ESHUnicodeSubstitution{\ESHMathSymbol{\rightarrow}}\ESHSpace{}\textcolor[HTML]{204A87}{\ESHUnicodeSubstitution{\ESHMathSymbol{\Re}}}\ESHSpace{}\ESHUnicodeSubstitution{\ESHMathSymbol{\Rightarrow}}\ESHSpace{}e\ESHSpace{}(\textcolor[HTML]{346604}{\ESHUnicodeSubstitution{\ESHMathSymbol{\lambda}}}\ESHSpace{}x\ESHSpace{}:\ESHSpace{}A\ESHSpace{}\ESHUnicodeSubstitution{\ESHMathSymbol{\Rightarrow}}\ESHSpace{}k\ESHSpace{}(f\ESHSpace{}x))\ESHEol
\ESHBol{}\textcolor[HTML]{346604}{let}\ESHSpace{}factor\ESHSpace{}(x\ESHSpace{}:\ESHSpace{}\textcolor[HTML]{204A87}{\ESHUnicodeSubstitution{\ESHMathSymbol{\Re}}})\ESHSpace{}:\ESHSpace{}Integral\ESHSpace{}unit\ESHSpace{}=\ESHSpace{}\textcolor[HTML]{346604}{\ESHUnicodeSubstitution{\ESHMathSymbol{\lambda}}}\ESHSpace{}f\ESHSpace{}:\ESHSpace{}unit\ESHSpace{}\ESHUnicodeSubstitution{\ESHMathSymbol{\rightarrow}}\ESHSpace{}\textcolor[HTML]{204A87}{\ESHUnicodeSubstitution{\ESHMathSymbol{\Re}}}\ESHSpace{}\ESHUnicodeSubstitution{\ESHMathSymbol{\Rightarrow}}\ESHSpace{}f\ESHSpace{}()\ESHSpace{}*\ESHSpace{}x\ESHEol
\ESHBol{}\textcolor[HTML]{346604}{let}\ESHSpace{}measToProb\ESHSpace{}A\ESHSpace{}(e\ESHSpace{}:\ESHSpace{}Integral\ESHSpace{}A)\ESHSpace{}:\ESHSpace{}Integral\ESHSpace{}A\ESHSpace{}=\ESHSpace{}\textcolor[HTML]{346604}{\ESHUnicodeSubstitution{\ESHMathSymbol{\lambda}}}\ESHSpace{}f\ESHSpace{}:\ESHSpace{}A\ESHSpace{}\ESHUnicodeSubstitution{\ESHMathSymbol{\rightarrow}}\ESHSpace{}\textcolor[HTML]{204A87}{\ESHUnicodeSubstitution{\ESHMathSymbol{\Re}}}\ESHSpace{}\ESHUnicodeSubstitution{\ESHMathSymbol{\Rightarrow}}\ESHSpace{}e\ESHSpace{}f\ESHSpace{}/\ESHSpace{}e\ESHSpace{}(\textcolor[HTML]{346604}{\ESHUnicodeSubstitution{\ESHMathSymbol{\lambda}}}\ESHSpace{}x\ESHSpace{}:\ESHSpace{}A\ESHSpace{}\ESHUnicodeSubstitution{\ESHMathSymbol{\Rightarrow}}\ESHSpace{}1)\ESHEol
\ESHBol{}\textcolor[HTML]{346604}{let}\ESHSpace{}bernoulli\ESHSpace{}(p\ESHSpace{}:\ESHSpace{}\textcolor[HTML]{204A87}{\ESHUnicodeSubstitution{\ESHMathSymbol{\Re}}})\ESHSpace{}:\ESHSpace{}Integral\ESHSpace{}\textcolor[HTML]{204A87}{\ESHUnicodeSubstitution{\ESHMathSymbol{\mathfrak{B}}}}\ESHSpace{}=\ESHSpace{}\textcolor[HTML]{346604}{\ESHUnicodeSubstitution{\ESHMathSymbol{\lambda}}}\ESHSpace{}f\ESHSpace{}:\ESHSpace{}\textcolor[HTML]{204A87}{\ESHUnicodeSubstitution{\ESHMathSymbol{\mathfrak{B}}}}\ESHSpace{}\ESHUnicodeSubstitution{\ESHMathSymbol{\rightarrow}}\ESHSpace{}\textcolor[HTML]{204A87}{\ESHUnicodeSubstitution{\ESHMathSymbol{\Re}}}\ESHSpace{}\ESHUnicodeSubstitution{\ESHMathSymbol{\Rightarrow}}\ESHSpace{}p\ESHSpace{}*\ESHSpace{}f\ESHSpace{}tt\ESHSpace{}+\ESHSpace{}(1\ESHSpace{}\ESHDash{}\ESHSpace{}p)\ESHSpace{}*\ESHSpace{}f\ESHSpace{}ff\ESHEol
\ESHBol{}\textcolor[HTML]{346604}{let}\ESHSpace{}uniform\ESHSpace{}:\ESHSpace{}Integral\ESHSpace{}\textcolor[HTML]{204A87}{\ESHUnicodeSubstitution{\ESHMathSymbol{\Re}}}\ESHSpace{}=\ESHSpace{}integral01\ESHEol
\ESHBol{}\ESHEmptyLine{}\ESHEol
\ESHBol{}\textcolor[HTML]{346604}{let}\ESHSpace{}total\_mass\ESHSpace{}A\ESHSpace{}(mu\ESHSpace{}:\ESHSpace{}Integral\ESHSpace{}A)\ESHSpace{}=\ESHSpace{}mu\ESHSpace{}(\textcolor[HTML]{346604}{\ESHUnicodeSubstitution{\ESHMathSymbol{\lambda}}}\ESHSpace{}x\ESHSpace{}:\ESHSpace{}A\ESHSpace{}\ESHUnicodeSubstitution{\ESHMathSymbol{\Rightarrow}}\ESHSpace{}1)\ESHEol
\ESHBol{}\textcolor[HTML]{346604}{let}\ESHSpace{}mean\ESHSpace{}(mu\ESHSpace{}:\ESHSpace{}Integral\ESHSpace{}\textcolor[HTML]{204A87}{\ESHUnicodeSubstitution{\ESHMathSymbol{\Re}}})\ESHSpace{}=\ESHSpace{}mu\ESHSpace{}(\textcolor[HTML]{346604}{\ESHUnicodeSubstitution{\ESHMathSymbol{\lambda}}}\ESHSpace{}x\ESHSpace{}:\ESHSpace{}\textcolor[HTML]{204A87}{\ESHUnicodeSubstitution{\ESHMathSymbol{\Re}}}\ESHSpace{}\ESHUnicodeSubstitution{\ESHMathSymbol{\Rightarrow}}\ESHSpace{}x)\ESHEol
\ESHBol{}\textcolor[HTML]{346604}{let}\ESHSpace{}variance\ESHSpace{}(mu\ESHSpace{}:\ESHSpace{}Integral\ESHSpace{}\textcolor[HTML]{204A87}{\ESHUnicodeSubstitution{\ESHMathSymbol{\Re}}})\ESHSpace{}=\ESHSpace{}mu\ESHSpace{}(\textcolor[HTML]{346604}{\ESHUnicodeSubstitution{\ESHMathSymbol{\lambda}}}\ESHSpace{}x\ESHSpace{}:\ESHSpace{}\textcolor[HTML]{204A87}{\ESHUnicodeSubstitution{\ESHMathSymbol{\Re}}}\ESHSpace{}\ESHUnicodeSubstitution{\ESHMathSymbol{\Rightarrow}}\ESHSpace{}(x\ESHSpace{}\ESHDash{}\ESHSpace{}mean\ESHSpace{}mu)\ESHRaise{0.30}{2})
\end{ESHBlock}
\codethensec
\caption{Integrals and \dlang{} programs that manipulate them.}
\label{code:expecter}
\vspace{-0.8em}
\end{figure}

\paragraph{Example.}

What happens if we make an infinitesimal perturbation to the uniform distribution as in \Cref{fig:probability}?
How will its mean and variance change?
Differentiation answers these questions.

%Let's consider some probability distributions and computations on them.
The uniform distribution over the interval $[0, 1]$ is equivalent to the integral of $[0, 1]$, namely \ESHInline{\ESHBol{}uniform\ESHSpace{}:\ESHSpace{}Integral\ESHSpace{}\textcolor[HTML]{204A87}{\ESHUnicodeSubstitution{\ESHMathSymbol{\Re}}}\ESHSpace{}=\ESHSpace{}integral01}.
\NA{It satisfies $\int_0^1 1 dx = 1$ (as any probability distribution must),} and has mean $\int_0^1 x dx = 1/2$ and variance
%\[
$\int_0^1 (x - 1/2)^2 dx = 1/12$.
%\]

Next, we must craft a perturbation to consider. There is an isomorphism \ESHInline{\ESHBol{}Tan\ESHSpace{}(Integral\ESHSpace{}A)\ESHSpace{}\ESHUnicodeSubstitution{\ESHMathSymbol{\cong}}\ESHSpace{}Integral\ESHSpace{}A\ESHSpace{}*\ESHSpace{}Integral\ESHSpace{}A}, which says that a perturbation to an integral itself has the form of an integral as well. Hence, our \NA{perturbation must also be an integral}.
In addition, because we are perturbing a probability distribution, whose total mass must sum to 1, the total mass of our perturbation must be 0: if we are to increase mass somewhere, we must decrease it elsewhere.
% there are  additional restrictions on these infinitesimal perturbations, one of which is that
Given these design considerations, consider the following perturbation to the uniform distribution that makes 1 more likely, 0 less likely, 1/2 equally likely as before, and interpolates between these:\footnote{\Cref{fig:probability} shows a schematic of this perturbation.}
\begin{ESHBlock}
\ESHBol{}\textcolor[HTML]{346604}{let}\ESHSpace{}change\ESHSpace{}:\ESHSpace{}Integral\ESHSpace{}\textcolor[HTML]{204A87}{\ESHUnicodeSubstitution{\ESHMathSymbol{\Re}}}\ESHSpace{}=\ESHSpace{}\textcolor[HTML]{346604}{\ESHUnicodeSubstitution{\ESHMathSymbol{\lambda}}}\ESHSpace{}f\ESHSpace{}:\ESHSpace{}\textcolor[HTML]{204A87}{\ESHUnicodeSubstitution{\ESHMathSymbol{\Re}}}\ESHSpace{}\ESHUnicodeSubstitution{\ESHMathSymbol{\rightarrow}}\ESHSpace{}\textcolor[HTML]{204A87}{\ESHUnicodeSubstitution{\ESHMathSymbol{\Re}}}\ESHSpace{}\ESHUnicodeSubstitution{\ESHMathSymbol{\Rightarrow}}\ESHSpace{}integral01\ESHSpace{}(\textcolor[HTML]{346604}{\ESHUnicodeSubstitution{\ESHMathSymbol{\lambda}}}\ESHSpace{}x\ESHSpace{}:\ESHSpace{}\textcolor[HTML]{204A87}{\ESHUnicodeSubstitution{\ESHMathSymbol{\Re}}}\ESHSpace{}\ESHUnicodeSubstitution{\ESHMathSymbol{\Rightarrow}}\ESHSpace{}(x\ESHSpace{}\ESHDash{}\ESHSpace{}1/2)\ESHSpace{}*\ESHSpace{}f\ESHSpace{}x)
\end{ESHBlock}
The perturbation is an integral with total mass 0: $\int_0^1 (x-1/2) dx = 0$.

Returning to our question of how this perturbation changes the mean and variance of \ESHInline{\ESHBol{}uniform}, for convenience let \ESHInline{\ESHBol{}der\ESHSpace{}:\ESHSpace{}(Integral\ESHSpace{}A\ESHSpace{}\ESHUnicodeSubstitution{\ESHMathSymbol{\rightarrow}}\ESHSpace{}\textcolor[HTML]{204A87}{\ESHUnicodeSubstitution{\ESHMathSymbol{\Re}}})\ESHSpace{}\ESHUnicodeSubstitution{\ESHMathSymbol{\rightarrow}}\ESHSpace{}Integral\ESHSpace{}A\ESHSpace{}\ESHUnicodeSubstitution{\ESHMathSymbol{\rightarrow}}\ESHSpace{}Integral\ESHSpace{}A\ESHSpace{}\ESHUnicodeSubstitution{\ESHMathSymbol{\rightarrow}}\ESHSpace{}\textcolor[HTML]{204A87}{\ESHUnicodeSubstitution{\ESHMathSymbol{\Re}}}} compute the derivative of its argument at a point and infinitesimal perturbation, using the appropriate coercions and projections to and from tangent spaces.\footnote{\ESHInline{\ESHBol{}\textcolor[HTML]{346604}{let}\ESHSpace{}der\ESHSpace{}f\ESHSpace{}x\ESHSpace{}dx\ESHSpace{}=\ESHSpace{}snd\ESHSpace{}(tangetTo\_R.to\ESHSpace{}(tangent\ESHSpace{}f\ESHSpace{}(tangetTo\_R.from\ESHSpace{}(x,\ESHSpace{}dx))))}}
Since \ESHInline{\ESHBol{}mean} is linear, its derivative is independent of the current value and is just the original \ESHInline{\ESHBol{}mean} function applied to the infinitesimal perturbation:
\begin{ESHBlock}
\ESHBol{}der\ESHSpace{}mean\ESHSpace{}uniform\ESHSpace{}change\ESHEol
\ESHBol{}\ESHSpace{}\ESHSpace{}=\ESHSpace{}mean\ESHSpace{}change\ESHEol
\ESHBol{}\ESHSpace{}\ESHSpace{}=\ESHSpace{}integral01\ESHSpace{}(\textcolor[HTML]{346604}{\ESHUnicodeSubstitution{\ESHMathSymbol{\lambda}}}\ESHSpace{}x\ESHSpace{}:\ESHSpace{}\textcolor[HTML]{204A87}{\ESHUnicodeSubstitution{\ESHMathSymbol{\Re}}}\ESHSpace{}\ESHUnicodeSubstitution{\ESHMathSymbol{\Rightarrow}}\ESHSpace{}(x\ESHSpace{}\ESHDash{}\ESHSpace{}1/2)\ESHSpace{}*\ESHSpace{}x)\ESHEol
\ESHBol{}\ESHSpace{}\ESHSpace{}=\ESHSpace{}1/12
\end{ESHBlock}
And indeed, that's what we compute:

\vspace{-.5em}
\begin{center}
\begin{minipage}[t]{0.88\textwidth}
\begin{prompt}
\ESHInline{\ESHBol{}eps=1e\ESHDash{}3{>}\ESHSpace{}der\ESHSpace{}mean\ESHSpace{}uniform\ESHSpace{}change}
\end{prompt}
\begin{repl}
\ESHInline{\ESHBol{}[0.0829,\ESHSpace{}0.0837]}
\end{repl}
\end{minipage}
\end{center}
\vspace{.25em}
However, \ESHInline{\ESHBol{}variance} is nonlinear, so its derivative does depend on the current point.
Let's compute it and then reason about the answer:
\vspace{-.5em}
\begin{center}
\begin{minipage}[t]{0.88\textwidth}
\begin{prompt}
\ESHInline{\ESHBol{}eps=1e\ESHDash{}2{>}\ESHSpace{}der\ESHSpace{}variance\ESHSpace{}uniform\ESHSpace{}change}
\end{prompt}
\begin{repl}
\ESHInline{\ESHBol{}[\ESHDash{}0.005,\ESHSpace{}0.004]}
\end{repl}
\end{minipage}
\end{center}
\vspace{.25em}
We can reason about the change in the variance with the laws about derivatives, just as we would in first-order cases:
\begin{ESHBlock}
\ESHBol{}der\ESHSpace{}variance\ESHSpace{}uniform\ESHSpace{}change\ESHEol
\ESHBol{}\ESHSpace{}\ESHSpace{}=\ESHSpace{}der\ESHSpace{}(\textcolor[HTML]{346604}{\ESHUnicodeSubstitution{\ESHMathSymbol{\lambda}}}\ESHSpace{}mu\ESHSpace{}:\ESHSpace{}Integral\ESHSpace{}\textcolor[HTML]{204A87}{\ESHUnicodeSubstitution{\ESHMathSymbol{\Re}}}\ESHSpace{}\ESHUnicodeSubstitution{\ESHMathSymbol{\Rightarrow}}\ESHSpace{}mu\ESHSpace{}(\textcolor[HTML]{346604}{\ESHUnicodeSubstitution{\ESHMathSymbol{\lambda}}}\ESHSpace{}x\ESHSpace{}:\ESHSpace{}\textcolor[HTML]{204A87}{\ESHUnicodeSubstitution{\ESHMathSymbol{\Re}}}\ESHSpace{}\ESHUnicodeSubstitution{\ESHMathSymbol{\Rightarrow}}\ESHSpace{}x\ESHRaise{0.30}{2})\ESHSpace{}\ESHDash{}\ESHSpace{}(mean\ESHSpace{}mu)\ESHRaise{0.30}{2})\ESHSpace{}uniform\ESHSpace{}change\ESHEol
\ESHBol{}\ESHSpace{}\ESHSpace{}=\ESHSpace{}change\ESHSpace{}(\textcolor[HTML]{346604}{\ESHUnicodeSubstitution{\ESHMathSymbol{\lambda}}}\ESHSpace{}x\ESHSpace{}:\ESHSpace{}\textcolor[HTML]{204A87}{\ESHUnicodeSubstitution{\ESHMathSymbol{\Re}}}\ESHSpace{}\ESHUnicodeSubstitution{\ESHMathSymbol{\Rightarrow}}\ESHSpace{}x\ESHRaise{0.30}{2})\ESHSpace{}\ESHDash{}\ESHSpace{}2\ESHSpace{}*\ESHSpace{}mean\ESHSpace{}uniform\ESHSpace{}*\ESHSpace{}mean\ESHSpace{}change\ESHEol
\ESHBol{}\ESHSpace{}\ESHSpace{}=\ESHSpace{}integral01\ESHSpace{}(\textcolor[HTML]{346604}{\ESHUnicodeSubstitution{\ESHMathSymbol{\lambda}}}\ESHSpace{}x\ESHSpace{}:\ESHSpace{}\textcolor[HTML]{204A87}{\ESHUnicodeSubstitution{\ESHMathSymbol{\Re}}}\ESHSpace{}\ESHUnicodeSubstitution{\ESHMathSymbol{\Rightarrow}}\ESHSpace{}(x\ESHDash{}1/2)*x\ESHRaise{0.30}{2})\ESHSpace{}\ESHDash{}\ESHSpace{}2\ESHSpace{}*\ESHSpace{}1/2\ESHSpace{}*\ESHSpace{}1/12\ESHEol
\ESHBol{}\ESHSpace{}\ESHSpace{}=\ESHSpace{}1/12\ESHSpace{}\ESHDash{}\ESHSpace{}1/12\ESHEol
\ESHBol{}\ESHSpace{}\ESHSpace{}=\ESHSpace{}0
\end{ESHBlock}
So it turns out that this infinitesimal perturbation will actually not change the variance.

\Comment{
In the previous case, we chose an infinitesimal perturbation \ESHInline{\ESHBol{}change} that can be viewed as a \emph{signed measure}, meaning that it is the difference of two measures.
However, not all perturbations of measures or probability distributions are of this form.
Let's consider the derivative of the Dirac delta function \ESHInline{\ESHBol{}dirac}.
Letting \ESHInline{\ESHBol{}der\ESHSpace{}:\ESHSpace{}(A\ESHSpace{}\ESHUnicodeSubstitution{\ESHMathSymbol{\rightarrow}}\ESHSpace{}Integral\ESHSpace{}A)\ESHSpace{}\ESHUnicodeSubstitution{\ESHMathSymbol{\rightarrow}}\ESHSpace{}Tan\ESHSpace{}A\ESHSpace{}\ESHUnicodeSubstitution{\ESHMathSymbol{\rightarrow}}\ESHSpace{}Integral\ESHSpace{}A} produce the derivative by the appropriate coercion and projection, we have
\begin{ESHBlock}
\ESHBol{}der\ESHSpace{}dirac\ESHSpace{}(xdx\ESHSpace{}:\ESHSpace{}Tan\ESHSpace{}A)\ESHSpace{}:\ESHSpace{}Integral\ESHSpace{}A\ESHSpace{}=\ESHEol
\ESHBol{}\ESHSpace{}\ESHSpace{}\textcolor[HTML]{346604}{\ESHUnicodeSubstitution{\ESHMathSymbol{\lambda}}}\ESHSpace{}f\ESHSpace{}:\ESHSpace{}A\ESHSpace{}\ESHUnicodeSubstitution{\ESHMathSymbol{\rightarrow}}\ESHSpace{}\textcolor[HTML]{204A87}{\ESHUnicodeSubstitution{\ESHMathSymbol{\Re}}}\ESHSpace{}\ESHUnicodeSubstitution{\ESHMathSymbol{\Rightarrow}}\ESHSpace{}\textcolor[HTML]{346604}{let}\ESHSpace{}(y,\ESHSpace{}dy)\ESHSpace{}=\ESHSpace{}tangent\ESHSpace{}f\ESHSpace{}xdx\ESHSpace{}\textcolor[HTML]{346604}{in}\ESHSpace{}dy
\end{ESHBlock}
Unlike previous infinitesimal perturbations we've seen, this one is a linear operator that differentiates its argument \ESHInline{\ESHBol{}f}.
%Accordingly, it is what is called a \emph{distribution}, rather than a measure (all measures are distributions).
}

\TBD{Conclusion?}

\subsection{Implicit Surfaces and Root-Finding}
\label{smooth-implicit-surfaces}

\begin{figure}
\small
\begin{ESHBlock}
\ESHBol{}\textcolor[HTML]{204A87}{type}\ESHSpace{}Surface\ESHSpace{}A\ESHSpace{}=\ESHSpace{}A\ESHSpace{}\ESHUnicodeSubstitution{\ESHMathSymbol{\rightarrow}}\ESHSpace{}\textcolor[HTML]{204A87}{\ESHUnicodeSubstitution{\ESHMathSymbol{\Re}}}\ESHEol
\ESHBol{}\ESHEmptyLine{}\ESHEol
\ESHBol{}\textcolor[HTML]{346604}{let}\ESHSpace{}circle\ESHSpace{}(c\ESHSpace{}:\ESHSpace{}\textcolor[HTML]{204A87}{\ESHUnicodeSubstitution{\ESHMathSymbol{\Re}}}\ESHRaise{0.30}{2})\ESHSpace{}(r\ESHSpace{}:\ESHSpace{}\textcolor[HTML]{204A87}{\ESHUnicodeSubstitution{\ESHMathSymbol{\Re}}})\ESHSpace{}:\ESHSpace{}Surface\ESHSpace{}(\textcolor[HTML]{204A87}{\ESHUnicodeSubstitution{\ESHMathSymbol{\Re}}}\ESHRaise{0.30}{2})\ESHSpace{}=\ESHEol
\ESHBol{}\ESHSpace{}\ESHSpace{}\textcolor[HTML]{346604}{\ESHUnicodeSubstitution{\ESHMathSymbol{\lambda}}}\ESHSpace{}x\ESHSpace{}:\ESHSpace{}\textcolor[HTML]{204A87}{\ESHUnicodeSubstitution{\ESHMathSymbol{\Re}}}\ESHRaise{0.30}{2}\ESHSpace{}\ESHUnicodeSubstitution{\ESHMathSymbol{\Rightarrow}}\ESHSpace{}r\ESHRaise{0.30}{2}\ESHSpace{}\ESHDash{}\ESHSpace{}(x[0]\ESHSpace{}\ESHDash{}\ESHSpace{}c[0])\ESHRaise{0.30}{2}\ESHSpace{}\ESHDash{}\ESHSpace{}(x[1]\ESHSpace{}\ESHDash{}\ESHSpace{}c[1])\ESHRaise{0.30}{2}\ESHEol
\ESHBol{}\textcolor[HTML]{346604}{let}\ESHSpace{}halfplane\ESHSpace{}A\ESHSpace{}(normal\ESHSpace{}:\ESHSpace{}\textcolor[HTML]{204A87}{\ESHUnicodeSubstitution{\ESHMathSymbol{\Re}}}\ESHRaise{0.30}{2})\ESHSpace{}:\ESHSpace{}Surface\ESHSpace{}(\textcolor[HTML]{204A87}{\ESHUnicodeSubstitution{\ESHMathSymbol{\Re}}}\ESHRaise{0.30}{2})\ESHSpace{}=\ESHSpace{}\textcolor[HTML]{346604}{\ESHUnicodeSubstitution{\ESHMathSymbol{\lambda}}}\ESHSpace{}x\ESHSpace{}:\ESHSpace{}\textcolor[HTML]{204A87}{\ESHUnicodeSubstitution{\ESHMathSymbol{\Re}}}\ESHRaise{0.30}{2}\ESHSpace{}\ESHUnicodeSubstitution{\ESHMathSymbol{\Rightarrow}}\ESHSpace{}dot\ESHSpace{}normal\ESHSpace{}x\ESHEol
\ESHBol{}\textcolor[HTML]{346604}{let}\ESHSpace{}union\ESHSpace{}A\ESHSpace{}(s\ESHSpace{}s{'}\ESHSpace{}:\ESHSpace{}Surface\ESHSpace{}A)\ESHSpace{}:\ESHSpace{}Surface\ESHSpace{}A\ESHSpace{}=\ESHSpace{}\textcolor[HTML]{346604}{\ESHUnicodeSubstitution{\ESHMathSymbol{\lambda}}}\ESHSpace{}x\ESHSpace{}:\ESHSpace{}A\ESHSpace{}\ESHUnicodeSubstitution{\ESHMathSymbol{\Rightarrow}}\ESHSpace{}max\ESHSpace{}(s\ESHSpace{}x)\ESHSpace{}(s{'}\ESHSpace{}x)\ESHEol
\ESHBol{}\textcolor[HTML]{346604}{let}\ESHSpace{}intersection\ESHSpace{}A\ESHSpace{}(s\ESHSpace{}s{'}\ESHSpace{}:\ESHSpace{}Surface\ESHSpace{}A)\ESHSpace{}:\ESHSpace{}Surface\ESHSpace{}A\ESHSpace{}=\ESHSpace{}\ESHSpace{}\textcolor[HTML]{346604}{\ESHUnicodeSubstitution{\ESHMathSymbol{\lambda}}}\ESHSpace{}x\ESHSpace{}:\ESHSpace{}A\ESHSpace{}\ESHUnicodeSubstitution{\ESHMathSymbol{\Rightarrow}}\ESHSpace{}min\ESHSpace{}(s\ESHSpace{}x)\ESHSpace{}(s{'}\ESHSpace{}x)\ESHEol
\ESHBol{}\textcolor[HTML]{346604}{let}\ESHSpace{}complement\ESHSpace{}A\ESHSpace{}(s\ESHSpace{}:\ESHSpace{}Surface\ESHSpace{}A)\ESHSpace{}:\ESHSpace{}Surface\ESHSpace{}A\ESHSpace{}=\ESHSpace{}\textcolor[HTML]{346604}{\ESHUnicodeSubstitution{\ESHMathSymbol{\lambda}}}\ESHSpace{}x\ESHSpace{}:\ESHSpace{}A\ESHSpace{}\ESHUnicodeSubstitution{\ESHMathSymbol{\Rightarrow}}\ESHSpace{}\ESHDash{}\ESHSpace{}(s\ESHSpace{}x)
\end{ESHBlock}
\codethensec
\caption{A \dlang{} library for implicit surfaces.}
\label{code:surfaces}
\vspace{-0.8em}
\end{figure}

\Cref{dlang-intro} and \Cref{fig:diff-ray-tracing} presented a library for implicit surfaces and a function for performing ray tracing on scenes represented by implicit surfaces.

\Cref{code:surfaces} presents a library for constructing implicit surfaces.
An \emph{implicit surface} is a representation of a surface (such as a sphere or plane) with the zero-set of a differentiable function $\smooth{f : \R^n \to \R}$ (where usually we consider $n=3$ for 3-dimensional space).
Whether $f(x, y)$ is positive, negative, or zero indicates whether $(x, y)$ is inside, outside, or on the border of the surface, respectively.
The angle at which a ray deflects is determined by the \emph{surface normal} at the location where the ray hits the surface, which is the vector that is orthogonal to the plane that is tangent to the surface.

In \dlang{}, we can represent implicit surfaces as \ESHInline{\ESHBol{}\textcolor[HTML]{204A87}{type}\ESHSpace{}Surface\ESHSpace{}A\ESHSpace{}=\ESHSpace{}A\ESHSpace{}\ESHUnicodeSubstitution{\ESHMathSymbol{\rightarrow}}\ESHSpace{}\textcolor[HTML]{204A87}{\ESHUnicodeSubstitution{\ESHMathSymbol{\Re}}}}.
\Cref{code:surfaces} presents a small library for constructing implicit surfaces.
The Boolean operations of Constructive Solid Geometry (CSG) -- union, intersection, and complement -- are available for these implicit surfaces.
Because \dlang{} permits nonsmooth functions, it is able to represent implicit surfaces that don't necessarily correspond to manifolds, such as the union of two spheres that are offset and equally sized.
Where they touch, there is a corner, and thus there is no (unique) surface normal.

Our smooth ray tracer, shown in \Cref{code:raytracer}, renders the image of an implicit surface with a single light source and a Lambertian reflectance model, computing the angle at which light reflects off of the surface using automatic differentiation.
The code in \Cref{code:raytracer} reflects the contributions of \citet{niemeyer2019differentiable}, who use a differentiable ray-tracing renderer to learn implicit 3D representations of surfaces, noting their ``key insight is that depth gradients can be derived analytically using the concept of implicit differentiation.''

We can implement a smooth (and thus differentiable) ray tracer for implicit surfaces in \dlang{} in just a few lines of code, and the use of implicit differentiation automatically falls out.

\subsection{Generalized Parametric Surfaces and Optimization}
\label{parametric-surfaces}

We now build a library within \dlang{} for constructing shapes and computing operations on them.
For instance, we can represent the quarter disk and unit square in \Cref{fig:disk-square} as shapes and compute the Hausdorff distance between them, which equals $\sqrt{2} - 1$, as:
\vspace{-.5em}
\begin{center}
\begin{minipage}[t]{0.88\textwidth}
\begin{prompt}
\ESHInline{\ESHBol{}eps=1e\ESHDash{}3{>}\ESHSpace{}hausdorffDist\ESHSpace{}R2Dist\ESHSpace{}lShape\ESHSpace{}(quarterCircle\ESHSpace{}0)}
\end{prompt}
\begin{repl}
\ESHInline{\ESHBol{}[0.4138,\ESHSpace{}0.4145]}
\end{repl}
\end{minipage}
\end{center}
\vspace{.25em}

We can also compute derivatives, such as the infinitesimal perturbation in the Hausdorff distance that would result if the quarter circle were to infinitesimally move up by a unit magnitude:
\begin{center}
\begin{prompt}
\ESHInline{\ESHBol{}eps=1e\ESHDash{}1{>}\ESHSpace{}deriv\ESHSpace{}(\textcolor[HTML]{346604}{\ESHUnicodeSubstitution{\ESHMathSymbol{\lambda}}}\ESHSpace{}y\ESHSpace{}:\ESHSpace{}\textcolor[HTML]{204A87}{\ESHUnicodeSubstitution{\ESHMathSymbol{\Re}}}\ESHSpace{}\ESHUnicodeSubstitution{\ESHMathSymbol{\Rightarrow}}\ESHSpace{}hausdorffDist\ESHSpace{}R2Dist\ESHSpace{}lShape\ESHSpace{}(quarterCircle\ESHSpace{}y))\ESHSpace{}0}
\end{prompt}
\begin{repl}
\ESHInline{\ESHBol{}[\ESHDash{}0.752,\ESHSpace{}\ESHDash{}0.664]}
%[-infty, infty]
%[-infty, infty]
%[-infty, infty]
%[-infty, infty]
%[-infty, infty]
%[-infty, infty]
%[-infty, infty]
%[-infty, infty]
%[-infty, infty]
%[-infty, infty]
%[-1.4453436175471726262674599, -0.22347271151650204544111899]
%[-1.164475114332694297321243798, -0.3665490203256641472091622145]
%[-1.0010553488589931212932394640, -0.46752291340844785363310614232]
%[-0.9028955865032270836088556256219, -0.5374084303830795409544189251401]
%[-0.84050124162376018173343365875388, -0.58648438312965836647469094970107]
%[-0.7988488692510790590625026726412305, -0.6216336094068997724483610948130362]
%[-0.77078566162067609664077641576236704, -0.64652649569483010653891213841760331]
%[-0.7515973045396820224886373089321421844, -0.6641561255883687886832219076605117364]
%[-0.73830744437945037911894669647383378199,-0.67667298653379602333902449902937726249]
\end{repl}
\end{center}

This application is admittedly more speculative in its practical applications, but it demonstrates a novel domain in which we can define and compute derivatives.
We will now explain how this library for shapes works.

We represent these generalized parametric surfaces as \emph{maximizers}, represented in $\lambda_S$ as
\[
\ESHInline{\ESHBol{}\textcolor[HTML]{204A87}{type}\ESHSpace{}Maximizer\ESHSpace{}A\ESHSpace{}=\ESHSpace{}(A\ESHSpace{}\ESHUnicodeSubstitution{\ESHMathSymbol{\rightarrow}}\ESHSpace{}\textcolor[HTML]{204A87}{\ESHUnicodeSubstitution{\ESHMathSymbol{\Re}}})\ESHSpace{}\ESHUnicodeSubstitution{\ESHMathSymbol{\rightarrow}}\ESHSpace{}\textcolor[HTML]{204A87}{\ESHUnicodeSubstitution{\ESHMathSymbol{\Re}}}}.
\]
Maximizers are functions $F : (A \to \R) \to \R$ that satisfy the algebraic laws $F(\lambda x : A.\ k) = k$ for all $k \in \R$ and $F(\lambda x : A.\ \max(f(x), g(x))) = \max(F(f), F(g))$ (analogously to how integrals are functions that satisfy the algebraic laws of linearity).
A generalized parametric surface \ESHInline{\ESHBol{}k\ESHSpace{}:\ESHSpace{}Maximizer\ESHSpace{}A}, when applied to a function \ESHInline{\ESHBol{}f\ESHSpace{}:\ESHSpace{}A\ESHSpace{}\ESHUnicodeSubstitution{\ESHMathSymbol{\rightarrow}}\ESHSpace{}\textcolor[HTML]{204A87}{\ESHUnicodeSubstitution{\ESHMathSymbol{\Re}}}}, returns the maximum value that \ESHInline{\ESHBol{}f} attains on the region represented by \ESHInline{\ESHBol{}k}.

\begin{figure}[htbp]
\small
\begin{ESHBlock}
\ESHBol{}\textcolor[HTML]{204A87}{type}\ESHSpace{}Maximizer\ESHSpace{}A\ESHSpace{}=\ESHSpace{}(A\ESHSpace{}\ESHUnicodeSubstitution{\ESHMathSymbol{\rightarrow}}\ESHSpace{}\textcolor[HTML]{204A87}{\ESHUnicodeSubstitution{\ESHMathSymbol{\Re}}})\ESHSpace{}\ESHUnicodeSubstitution{\ESHMathSymbol{\rightarrow}}\ESHSpace{}\textcolor[HTML]{204A87}{\ESHUnicodeSubstitution{\ESHMathSymbol{\Re}}}\ESHEol
\ESHBol{}\textcolor[HTML]{346604}{let}\ESHSpace{}point\ESHSpace{}A\ESHSpace{}(x\ESHSpace{}:\ESHSpace{}A)\ESHSpace{}:\ESHSpace{}Maximizer\ESHSpace{}A\ESHSpace{}=\ESHSpace{}\textcolor[HTML]{346604}{\ESHUnicodeSubstitution{\ESHMathSymbol{\lambda}}}\ESHSpace{}f\ESHSpace{}:\ESHSpace{}A\ESHSpace{}\ESHUnicodeSubstitution{\ESHMathSymbol{\rightarrow}}\ESHSpace{}\textcolor[HTML]{204A87}{\ESHUnicodeSubstitution{\ESHMathSymbol{\Re}}}\ESHSpace{}\ESHUnicodeSubstitution{\ESHMathSymbol{\Rightarrow}}\ESHSpace{}f\ESHSpace{}x\ESHEol
\ESHBol{}\textcolor[HTML]{346604}{let}\ESHSpace{}indexedUnion\ESHSpace{}A\ESHSpace{}B\ESHSpace{}(ka\ESHSpace{}:\ESHSpace{}Maximizer\ESHSpace{}A)\ESHSpace{}(kb\ESHSpace{}:\ESHSpace{}A\ESHSpace{}\ESHUnicodeSubstitution{\ESHMathSymbol{\rightarrow}}\ESHSpace{}Maximizer\ESHSpace{}B)\ESHSpace{}:\ESHSpace{}Maximizer\ESHSpace{}B\ESHSpace{}=\ESHEol
\ESHBol{}\ESHSpace{}\ESHSpace{}\textcolor[HTML]{346604}{\ESHUnicodeSubstitution{\ESHMathSymbol{\lambda}}}\ESHSpace{}f\ESHSpace{}:\ESHSpace{}B\ESHSpace{}\ESHUnicodeSubstitution{\ESHMathSymbol{\rightarrow}}\ESHSpace{}\textcolor[HTML]{204A87}{\ESHUnicodeSubstitution{\ESHMathSymbol{\Re}}}\ESHSpace{}\ESHUnicodeSubstitution{\ESHMathSymbol{\Rightarrow}}\ESHSpace{}ka\ESHSpace{}(\textcolor[HTML]{346604}{\ESHUnicodeSubstitution{\ESHMathSymbol{\lambda}}}\ESHSpace{}a\ESHSpace{}:\ESHSpace{}A\ESHSpace{}\ESHUnicodeSubstitution{\ESHMathSymbol{\Rightarrow}}\ESHSpace{}kb\ESHSpace{}a\ESHSpace{}f)\ESHEol
\ESHBol{}\textcolor[HTML]{346604}{let}\ESHSpace{}union\ESHSpace{}A\ESHSpace{}(k1\ESHSpace{}k2\ESHSpace{}:\ESHSpace{}Maximizer\ESHSpace{}A)\ESHSpace{}:\ESHSpace{}Maximizer\ESHSpace{}A\ESHSpace{}=\ESHEol
\ESHBol{}\ESHSpace{}\ESHSpace{}\textcolor[HTML]{346604}{\ESHUnicodeSubstitution{\ESHMathSymbol{\lambda}}}\ESHSpace{}f\ESHSpace{}:\ESHSpace{}A\ESHSpace{}\ESHUnicodeSubstitution{\ESHMathSymbol{\rightarrow}}\ESHSpace{}\textcolor[HTML]{204A87}{\ESHUnicodeSubstitution{\ESHMathSymbol{\Re}}}\ESHSpace{}\ESHUnicodeSubstitution{\ESHMathSymbol{\Rightarrow}}\ESHSpace{}max\ESHSpace{}(k1\ESHSpace{}f)\ESHSpace{}(k2\ESHSpace{}f)\ESHEol
\ESHBol{}\textcolor[HTML]{346604}{let}\ESHSpace{}map\ESHSpace{}A\ESHSpace{}B\ESHSpace{}(g\ESHSpace{}:\ESHSpace{}A\ESHSpace{}\ESHUnicodeSubstitution{\ESHMathSymbol{\rightarrow}}\ESHSpace{}B)\ESHSpace{}(k\ESHSpace{}:\ESHSpace{}Maximizer\ESHSpace{}A)\ESHSpace{}:\ESHSpace{}Maximizer\ESHSpace{}B\ESHSpace{}=\ESHEol
\ESHBol{}\ESHSpace{}\ESHSpace{}\textcolor[HTML]{346604}{\ESHUnicodeSubstitution{\ESHMathSymbol{\lambda}}}\ESHSpace{}f\ESHSpace{}:\ESHSpace{}B\ESHSpace{}\ESHUnicodeSubstitution{\ESHMathSymbol{\rightarrow}}\ESHSpace{}\textcolor[HTML]{204A87}{\ESHUnicodeSubstitution{\ESHMathSymbol{\Re}}}\ESHSpace{}\ESHUnicodeSubstitution{\ESHMathSymbol{\Rightarrow}}\ESHSpace{}k\ESHSpace{}(\textcolor[HTML]{346604}{\ESHUnicodeSubstitution{\ESHMathSymbol{\lambda}}}\ESHSpace{}a\ESHSpace{}:\ESHSpace{}\textcolor[HTML]{204A87}{\ESHUnicodeSubstitution{\ESHMathSymbol{\Re}}}\ESHSpace{}\ESHUnicodeSubstitution{\ESHMathSymbol{\Rightarrow}}\ESHSpace{}f\ESHSpace{}(g\ESHSpace{}a))\ESHEol
\ESHBol{}\textcolor[HTML]{346604}{let}\ESHSpace{}sup\ESHSpace{}A\ESHSpace{}(k\ESHSpace{}:\ESHSpace{}Maximizer\ESHSpace{}A)\ESHSpace{}(f\ESHSpace{}:\ESHSpace{}A\ESHSpace{}\ESHUnicodeSubstitution{\ESHMathSymbol{\rightarrow}}\ESHSpace{}\textcolor[HTML]{204A87}{\ESHUnicodeSubstitution{\ESHMathSymbol{\Re}}})\ESHSpace{}:\ESHSpace{}\textcolor[HTML]{204A87}{\ESHUnicodeSubstitution{\ESHMathSymbol{\Re}}}\ESHSpace{}=\ESHSpace{}k\ESHSpace{}f\ESHEol
\ESHBol{}\textcolor[HTML]{346604}{let}\ESHSpace{}inf\ESHSpace{}A\ESHSpace{}(k\ESHSpace{}:\ESHSpace{}Maximizer\ESHSpace{}A)\ESHSpace{}(f\ESHSpace{}:\ESHSpace{}A\ESHSpace{}\ESHUnicodeSubstitution{\ESHMathSymbol{\rightarrow}}\ESHSpace{}\textcolor[HTML]{204A87}{\ESHUnicodeSubstitution{\ESHMathSymbol{\Re}}})\ESHSpace{}:\ESHSpace{}\textcolor[HTML]{204A87}{\ESHUnicodeSubstitution{\ESHMathSymbol{\Re}}}\ESHSpace{}=\ESHSpace{}\ESHDash{}\ESHSpace{}k\ESHSpace{}(\textcolor[HTML]{346604}{\ESHUnicodeSubstitution{\ESHMathSymbol{\lambda}}}\ESHSpace{}x\ESHSpace{}:\ESHSpace{}A\ESHSpace{}\ESHUnicodeSubstitution{\ESHMathSymbol{\Rightarrow}}\ESHSpace{}\ESHDash{}\ESHSpace{}(f\ESHSpace{}x))\ESHEol
\ESHBol{}\textcolor[HTML]{346604}{let}\ESHSpace{}hausdorffDist\ESHSpace{}A\ESHSpace{}(d\ESHSpace{}:\ESHSpace{}A\ESHSpace{}\ESHUnicodeSubstitution{\ESHMathSymbol{\rightarrow}}\ESHSpace{}A\ESHSpace{}\ESHUnicodeSubstitution{\ESHMathSymbol{\rightarrow}}\ESHSpace{}\textcolor[HTML]{204A87}{\ESHUnicodeSubstitution{\ESHMathSymbol{\Re}}})\ESHSpace{}(k1\ESHSpace{}k2\ESHSpace{}:\ESHSpace{}Maximizer\ESHSpace{}A)\ESHSpace{}:\ESHSpace{}\textcolor[HTML]{204A87}{\ESHUnicodeSubstitution{\ESHMathSymbol{\Re}}}\ESHSpace{}=\ESHEol
\ESHBol{}\ESHSpace{}\ESHSpace{}max\ESHSpace{}(sup\ESHSpace{}k1\ESHSpace{}(\textcolor[HTML]{346604}{\ESHUnicodeSubstitution{\ESHMathSymbol{\lambda}}}\ESHSpace{}x1\ESHSpace{}:\ESHSpace{}A\ESHSpace{}\ESHUnicodeSubstitution{\ESHMathSymbol{\Rightarrow}}\ESHSpace{}inf\ESHSpace{}k2\ESHSpace{}(\textcolor[HTML]{346604}{\ESHUnicodeSubstitution{\ESHMathSymbol{\lambda}}}\ESHSpace{}x2\ESHSpace{}:\ESHSpace{}A\ESHSpace{}\ESHUnicodeSubstitution{\ESHMathSymbol{\Rightarrow}}\ESHSpace{}d\ESHSpace{}x1\ESHSpace{}x2)))\ESHEol
\ESHBol{}\ESHSpace{}\ESHSpace{}\ESHSpace{}\ESHSpace{}\ESHSpace{}\ESHSpace{}(sup\ESHSpace{}k2\ESHSpace{}(\textcolor[HTML]{346604}{\ESHUnicodeSubstitution{\ESHMathSymbol{\lambda}}}\ESHSpace{}x2\ESHSpace{}:\ESHSpace{}A\ESHSpace{}\ESHUnicodeSubstitution{\ESHMathSymbol{\Rightarrow}}\ESHSpace{}inf\ESHSpace{}k1\ESHSpace{}(\textcolor[HTML]{346604}{\ESHUnicodeSubstitution{\ESHMathSymbol{\lambda}}}\ESHSpace{}x1\ESHSpace{}:\ESHSpace{}A\ESHSpace{}\ESHUnicodeSubstitution{\ESHMathSymbol{\Rightarrow}}\ESHSpace{}d\ESHSpace{}x1\ESHSpace{}x2)))\ESHEol
\ESHBol{}\ESHEmptyLine{}\ESHEol
\ESHBol{}\textcolor[HTML]{346604}{let}\ESHSpace{}unitInterval\ESHSpace{}:\ESHSpace{}Maximizer\ESHSpace{}\textcolor[HTML]{204A87}{\ESHUnicodeSubstitution{\ESHMathSymbol{\Re}}}\ESHSpace{}=\ESHSpace{}max01\ESHEol
\ESHBol{}\textcolor[HTML]{346604}{let}\ESHSpace{}quarterCircle\ESHSpace{}(y\ESHSpace{}:\ESHSpace{}\textcolor[HTML]{204A87}{\ESHUnicodeSubstitution{\ESHMathSymbol{\Re}}})\ESHSpace{}:\ESHSpace{}Maximizer\ESHSpace{}(\textcolor[HTML]{204A87}{\ESHUnicodeSubstitution{\ESHMathSymbol{\Re}}}\ESHRaise{0.30}{2})\ESHSpace{}=\ESHSpace{}map\ESHEol
\ESHBol{}\ESHSpace{}\ESHSpace{}(\textcolor[HTML]{346604}{\ESHUnicodeSubstitution{\ESHMathSymbol{\lambda}}}\ESHSpace{}theta\ESHSpace{}:\ESHSpace{}\textcolor[HTML]{204A87}{\ESHUnicodeSubstitution{\ESHMathSymbol{\Re}}}\ESHSpace{}\ESHUnicodeSubstitution{\ESHMathSymbol{\Rightarrow}}\ESHSpace{}(cos\ESHSpace{}(pi\ESHSpace{}/\ESHSpace{}2\ESHSpace{}*\ESHSpace{}theta),\ESHSpace{}sin\ESHSpace{}(pi\ESHSpace{}/\ESHSpace{}2\ESHSpace{}*\ESHSpace{}theta)\ESHSpace{}+\ESHSpace{}y))\ESHEol
\ESHBol{}\ESHSpace{}\ESHSpace{}unitInterval\ESHEol
\ESHBol{}\textcolor[HTML]{346604}{let}\ESHSpace{}lShape\ESHSpace{}:\ESHSpace{}Maximizer\ESHSpace{}(\textcolor[HTML]{204A87}{\ESHUnicodeSubstitution{\ESHMathSymbol{\Re}}}\ESHRaise{0.30}{2})\ESHSpace{}=\ESHEol
\ESHBol{}\ESHSpace{}\ESHSpace{}union\ESHSpace{}(map\ESHSpace{}(\textcolor[HTML]{346604}{\ESHUnicodeSubstitution{\ESHMathSymbol{\lambda}}}\ESHSpace{}x\ESHSpace{}:\ESHSpace{}\textcolor[HTML]{204A87}{\ESHUnicodeSubstitution{\ESHMathSymbol{\Re}}}\ESHSpace{}\ESHUnicodeSubstitution{\ESHMathSymbol{\rightarrow}}\ESHSpace{}(x,\ESHSpace{}1))\ESHSpace{}unitInterval)\ESHEol
\ESHBol{}\ESHSpace{}\ESHSpace{}\ESHSpace{}\ESHSpace{}\ESHSpace{}\ESHSpace{}\ESHSpace{}\ESHSpace{}(map\ESHSpace{}(\textcolor[HTML]{346604}{\ESHUnicodeSubstitution{\ESHMathSymbol{\lambda}}}\ESHSpace{}y\ESHSpace{}:\ESHSpace{}\textcolor[HTML]{204A87}{\ESHUnicodeSubstitution{\ESHMathSymbol{\Re}}}\ESHSpace{}\ESHUnicodeSubstitution{\ESHMathSymbol{\Rightarrow}}\ESHSpace{}(1,\ESHSpace{}y))\ESHSpace{}unitInterval)\ESHEol
\ESHBol{}\textcolor[HTML]{346604}{let}\ESHSpace{}R2Dist\ESHSpace{}(a\ESHSpace{}b\ESHSpace{}:\ESHSpace{}\textcolor[HTML]{204A87}{\ESHUnicodeSubstitution{\ESHMathSymbol{\Re}}}\ESHRaise{0.30}{2})\ESHSpace{}:\ESHSpace{}\textcolor[HTML]{204A87}{\ESHUnicodeSubstitution{\ESHMathSymbol{\Re}}}\ESHSpace{}=\ESHSpace{}sqrt\ESHSpace{}((a[0]\ESHSpace{}\ESHDash{}\ESHSpace{}b[0])\ESHRaise{0.30}{2}\ESHSpace{}+\ESHSpace{}(a[1]\ESHSpace{}\ESHDash{}\ESHSpace{}b[1])\ESHRaise{0.30}{2})
\end{ESHBlock}
\codethensec
\caption{Generalized parametric surfaces and \dlang{} programs that manipulate them.}
\label{code:maximizers}
\vspace{-0.8em}
\end{figure}

\Cref{code:maximizers} shows an excerpt of the library for generalized parametric surfaces.
Note that generalized parametric surfaces shapes form a monad (representing nondeterminism),
with \ESHInline{\ESHBol{}point} and \ESHInline{\ESHBol{}indexedUnion} as \emph{return} and \emph{bind},
yielding a programming model for constructing shapes.

\TBD{Explain more}

\NA{Returning to the earlier Hausdorff-distance example,
note that the maximal distance on the ``L'' shape occurs at the corner point, which is represented twice, as the endpoint of each line; thus, a maximum is taken over two equal distances.
In \cite{plotkin-paper}, because the maximum operator is defined with a partial conditional statement, the result --- not to mention the derivative --- would be undefined.
Because both the values and derivatives are the same for the two representations of this corner point, the derivative is a maximal element.
Also note that we need second derivatives to compute the derivative of the Hausdorff distance, due to the use of \ESHInline{\ESHBol{}max01}.}

\section{Discussion}
\label{smooth:discussion}

In this section, we discuss the capability of \dlang{} to represent control flow as well as the opportunity to soundly speed up execution of higher-order primitives using derivative information.

\subsection{Control Flow: Conditionals and Recursion}
\label{discrete}

\dlang{} supports discrete spaces, including in particular the Booleans $\bool$ and any well-founded set (such as the natural numbers).
The recursion principles for these yield, respectively, if-then-else expressions and well-founded recursion.
These control-flow expressions must be independent of ``continuous data'': all maps from \emph{connected} spaces to discrete spaces are constant\thesisonly{, just as in \clang{}}.
\paperonly{This property defines connected spaces.}
Connected spaces include all vector spaces, such as $\R^n$.
\citet{edalat2013} explain some particular issues that demonstrate why implementing piecewise-differentiable functions with branching is problematic.

\subsection{Optimizing Higher-Order Primitives with Derivative Information}

We can also use the fact that functions in \dlang{} come equipped with all their derivatives to opportunistically speed up some operations.
For instance, consider applying \ESHInline{\ESHBol{}cut\_root} to some function $f$.
Its value-level definition naturally maps to a bisection-like algorithm on the values of $f$.
However, since we have access to $f^{(1)}$, we can use a variation of Newton's method generalized to interval arithmetic to speed up the convergence drastically, and indeed we do this in our implementation.
Note that we are guaranteed that this optimization is sound, because consistency of differentiation ensures that $f^{(1)}$ appropriately reflects $f^{(0)}$.
It may be the case that $f^{(1)}$ returns $\bot$ at some points, or even everywhere, in which case the algorithm falls back on bisection to ensure progress.

Similarly, the literal interpretation of the value-level definition of Riemannian integration in \Cref{ho-semantics} maps to a quadrature method that uses only the values $f^{(0)}$ of $f$.
However, the availability of higher derivatives of $f$ makes it possible to use interval-based versions of higher-order integration methods, which can also drastically speed up the convergence.
We do not use these higher-order methods by default in our actual implementation.

\Comment{
\subsection{Future Work}

\citet{neural-odes} introduce deep learning models that incorporate \emph{neural ODEs}, i.e., ordinary differential equations (ODEs) to be differentiated through.
It should be possible to implement an arbitrary-precision ODE solver in \clang{}, following \cite{edalat2003}.
This could then be lifted to a higher-order function in \dlang{}, since its derivatives are themselves ODEs.
}

\section{Related Work}
\label{smooth:related-work}

\Cref{ad-semantics} illustrates the unique set of features that \dlang{} provides and their relationship to other approaches to AD semantics.
We note that these other approaches also have features that \dlang{} lacks.

\TBD{JM: I think that the related work section starts off too defensive and without enough motivation. I think that addressing Mike's comment on the value/growth of differentiable programming could be helpful.}

\TBD{Explain why so many people in PL semantics seem to care about differentiable programming semantics.}

\begin{table}[]
\caption[Summary of approaches to semantics of differentiable programming and their properties.]{Summary of other approaches to semantics of differentiable programming and their properties.
\\
\emph{Higher-order derivatives}: The differentiation operator can be iterated arbitrarily many times (when applied to smooth functions).
%\\
\emph{Higher-order functions}: A concrete test: is $\mathsf{twice}(f : \R \to \R)(x : \R) : \R \triangleq f(f(x))$ admitted?
\emph{Non-differentiable functions}: Some nondifferentiable functions are admitted. A concrete test: is $\max : \R^2 \to \R$ admitted? ``\emph{Clarke derivative}'' indicates that locally Lipschitz functions support derivatives in the sense of Clarke derivatives or L-derivatives \cite{edalat2004}, whereas ``\emph{partiality}'' indicates that nondifferentiable maps are supported by considering them to be partial at their discontinuities.
}
\begin{tabular}{l | p{2.1cm} p{2.1cm} p{3.2cm}}
 & higher-order functions & higher-order derivatives & nondifferentiable functions \\ \hline
\citet{vakar} & \Yes{} & \Yes{} & \No{} \\ \hline
\citet{edalat2013} & \Yes{} & \No{} & \Yes{} (Clarke derivative) \\ \hline
\citet{elliott-ad} & \No{} & \No{} & \No{} \\ \hline
\citet{plotkin-paper} & \No{} & \Yes{} & \Yes{} (partiality) \\ \hline
\citet{sigal} & \No{} & \Yes{} & \Yes{} (partiality) \\ \hline
\citet{diffcurry} & \Yes{} & \No{} & \No{} \\ \hline
\citet{huot} & \Yes{} & \Yes{} & \No{} \\ \hline
\citet{differential} & \Yes{} & \Yes{} & \No{} \\ \hline
\textbf{\dlang{} (this work)} & \Yes{} & \Yes{} & \Yes{} (Clarke derivative) \\ \hline
\end{tabular}
\label{ad-semantics}
\vspace{-0.8em}
\end{table}

\citet{edalat2013} describe a programming language for nonexpansive (i.e., Lipschitz constant 1) functions on the interval $[-1, 1]$ with a differentiation operator that applies to functions from $[-1, 1]$ to $[-1, 1]$.
The semantics of this differentiation operator are that of the L-derivative \citep{edalat2004, edalat2008}, which is closely related to the Clarke Jacobian definition we use.
Their domain-theoretic account ensures computability: in theory, results can be computed to arbitrary precision.
Their semantics is fundamentally limited to first-order derivatives: their interval type denotes $[-1,1] \times [-1, 1]$, corresponding to a \emph{dual-number} representation, baking in that limited capability.
It is unclear how that representation could be generalized directly to permit higher-order differentiation and appropriately handle nested differentiation (without the \emph{perturbation confusion} \citep{siskind2005perturbation} that may arise with nested differentiation).

\citet{elliott-higher-ad} presents a data type for representing smooth maps, where a smooth map $f$ is represented by the collection of its $k$th derivatives for all $k$.
\citet{elliott-higher-ad} defines the derivatives of some arithmetic functions as well as some categorical operations,
though the definition of composition of smooth maps is incorrect.
We support higher-order derivatives by adapting this representation for the Clarke derivative.

\citet{vakar} presents the semantics of a differentiable programming language that supports higher-order functions and higher-order derivatives using the quasitopos of diffeological spaces.
As a quasitopos, the semantics supports higher-order functions and quotient types.
\citet{vakar} show an internal derivative operator that can be applied to any function of any type, and thus can be applied repeatedly for higher-order derivatives.
We based our internal derivative operator on theirs.
Functions such as \ESHInline{\ESHBol{}max} that are not smooth are not admissible.
It is not made clear how one could implement a differentiable programming language supporting the expressive possibilities suggested by the semantics.

None of the works in \Cref{ad-semantics} describe higher-order functions for root-finding, optimization, or integration, nor do they describe datatypes for implicit surfaces, compact shapes, or probability distributions.
\citet{edalat2004} describe an integration operator in a domain-theoretic framework for differential calculus, but it does not handle higher-order derivatives.
\citet{icfp19} describe computable higher-order functions and libraries for root-finding, optimization, and integration, but does not admit differentiation of any sort.

We follow \citet{icfp19} in our approach to computability. We are unaware of any system that computes arbitrary-precision derivatives (given the definition of the function) in any capacity.

\section{Conclusion}
\label{smooth:conclusion}

This \paperonly{paper}\thesisonly{chapter} demonstrates how to compute and make sense of derivatives of higher-order functions, such as integration, optimization, and root-finding and at higher-order types, such as probability distributions, implicit surfaces, and generalized parametric surfaces.
Our libraries and case studies model existing differentiable algorithms, for instance, a differentiable ray tracer for implicit surfaces, without requiring any user-level differentiation code, in addition to demonstrating new differentiable algorithms, such as computing derivatives of the Hausdorff distance of generalized parametric surfaces.
Ideally, the ideas \dlang{} demonstrates may enable differentiable programming frameworks to support the new abstractions and expressivity suggested by this paper.

%Ideally, \dlang{} may enable deep-learning researchers to design models at a higher level of abstraction.

\begin{acks}
We thank Eric Atkinson, Tej Chajed, Alexander Lew, Alex Renda, and David Spivak, as well as the anonymous reviewers, for their helpful feedback and discussions.
This work was supported in part by the Office of Naval Research (ONR-N00014-17-1-2699). Any opinions, findings, and conclusions or recommendations expressed in this material are those of the author and do not necessarily reflect the views of the Office of Naval Research.
\end{acks}

\smoothpaperend

\end{document}